\documentclass[12pt]{article}
\usepackage{etex}
\usepackage{amsmath}
\usepackage{graphicx,psfrag,epsf}
\usepackage{enumerate}
\usepackage{natbib}
\usepackage{url}
\usepackage{setspace}
\usepackage{caption}

\newcommand{\blind}{1}

\usepackage{gyqcommands}
\usepackage[margin=1in]{geometry}

\allowdisplaybreaks

\usepackage{xcolor}
\hypersetup{
  colorlinks,
  linkcolor={red!50!black},
  citecolor={blue!70!black},
  urlcolor={blue}
}

\makeatletter
\let\amsmath@bigm\bigm
\renewcommand{\bigm}[1]{%
  \ifcsname fenced@\string#1\endcsname
    \expandafter\@firstoftwo
  \else
    \expandafter\@secondoftwo
  \fi
  {\expandafter\amsmath@bigm\csname fenced@\string#1\endcsname}%
  {\amsmath@bigm#1}%
}

\newcommand{\DeclareFence}[2]{\@namedef{fenced@\string#1}{#2}}
\makeatother
\DeclareFence{\mid}{|}

\begin{document}

\def\spacingset#1{\renewcommand{\baselinestretch}%
{#1}\small\normalsize} \spacingset{1}


\if1\blind
{
  \title{\bf A Joint MLE Approach to Large-Scale Structured Latent Attribute Analysis}
  \author{Yuqi Gu\footnote{Email: yg2811@columbia.edu.~Address:~Room 928 SSW, 1255 Amsterdam Avenue,~New York,~NY 10027.}
  \hspace{.2cm}\\
    Department of Statistics, Columbia University\\
    and \\
    Gongjun Xu\footnote{Email: gongjun@umich.edu.} \\
    Department of Statistics, University of Michigan}
    \date{}
  \maketitle
} \fi

\if0\blind
{
  \bigskip
  \bigskip
  \bigskip
  \begin{center}
    {\LARGE
    \bf 
     A Joint MLE Approach to Large-Scale Structured Latent Attribute Analysis\par}
\end{center}
  \medskip
} \fi

\begin{abstract}
Structured Latent Attribute Models (SLAMs) are a family of discrete latent variable models widely used in education, psychology, and epidemiology to model multivariate categorical data. 
A SLAM assumes that multiple discrete latent attributes explain the dependence of observed variables in a highly structured fashion. Usually, the maximum marginal likelihood estimation approach is adopted for SLAMs, treating the latent attributes as random effects. The increasing scope of modern assessment data involves large numbers of observed variables and high-dimensional latent attributes. This poses challenges to classical estimation methods and requires new methodology and understanding of latent variable modeling. Motivated by this, we consider the joint maximum likelihood estimation (MLE) approach to SLAMs, treating latent attributes as fixed unknown parameters. We investigate estimability, consistency, and computation in the regime where sample size, number of variables, and number of latent attributes all can diverge. We establish the statistical consistency of the joint MLE and propose efficient algorithms that scale well to large-scale data for several popular SLAMs. 
Simulation studies demonstrate the superior empirical performance of the proposed methods. An application to real data from an international educational assessment gives interpretable findings of cognitive diagnosis.
\end{abstract}

\noindent%


\spacingset{1.45} 

\section{Introduction}

\paragraph{A Modern Family of Fine-grained Discrete Latent Variable Models.}
Structured Latent Attribute Models (SLAMs) are discrete latent variable models that have attracted substantial attention in various applications, including cognitive diagnosis in educational assessments \citep{junker2001cognitive,HensonTemplin09,dela2011}, psychiatric diagnosis of mental disorders \citep{templin2006measurement,dela2018}, and epidemiological studies of disease etiology \citep{wu2017nested,o2019causes}. 
A SLAM assumes multiple binary latent {attributes} explain observed variables in a highly structured fashion. 
In particular, for each subject $i$ a SLAM models the $J$-dimensional observations $\bo r_i = (r_{i,1},\ldots,r_{i,J})$ using a $K$-dimensional  latent attribute profile $\bo a_i = (a_{i,1},\ldots,a_{i,K}) \in \{0,1\}^K$. 
 In many applications, each attribute $a_{i,k}=1$ or 0 carries substantive meanings; e.g., mastery/deficiency of some skill in an educational test, or presence/absence of some pathogen in epidemiological diagnosis.
An important ``{structured}'' feature of a SLAM comes from a binary loading matrix, the $\QQ$-matrix \citep{Tatsuoka1983}. 
The $J\times K$ matrix $\QQ=(q_{j,k})$ encodes how the observed variables depend on the latent attributes, where $q_{j,k}=1$ or 0 means whether or not the $j$th observed variable depends on the $k$th latent attribute.
By modeling the latent variables as multidimensional binary and incorporating structural constraints in the $\QQ$-matrix, SLAMs provide a powerful framework to infer subjects' fine-grained latent traits, and to perform clustering based on the inferred latent  profiles.

Since the latent variables are discrete, a SLAM can be viewed as a {mixture model}, where each subject's latent attribute profile $\bo a_i$ is a random variable following a categorical distribution with $|\{0,1\}^K|=2^K$ components. 
Over the past two decades when latent attribute models have attracted a great surge of interest, 
this perspective of treating subjects' latent attributes as \textit{random effects} is usually taken in the literature of modeling \citep{davier2008general,HensonTemplin09,dela2011}, estimation \citep{chen2015statistical,xu2018jasa,culpepper2019,gu2019jmlr}, and study of model identifiability \citep{xu2017,fang2019,gu2020partial,chen2020sparse}.
Taking this perspective, estimation is usually performed by maximizing the marginal likelihood. 
The corresponding estimators can be obtained via an EM algorithm for mixture models. 
But an obstacle to adopting such an approach in large-scale and high-dimensional data is that the number of latent patterns $2^K$ grows exponentially with the number of attributes $K$. 
This quickly becomes computationally cumbersome as $K$ grows large, which is commonly seen in modern large-scale assessment data.
{For example, the TIMSS 2003 8th grade dataset available in the R package \texttt{CDM} involves $K=13$ skill attributes, which gives rise to $2^{13}=8192$ binary skill patterns.}

\vspace{-3mm}
\paragraph{The Joint MLE Approach.}
On the other hand, the joint maximum likelihood estimation (joint MLE) approach treats the subjects' latent attributes {$\{\bo a_i: 1\leq i\leq N\}$} as \textit{fixed effects} and directly incorporates them into the likelihood as unknown parameters. 
{This approach would naturally avoid the need to model the joint distribution of the exponentially many latent attribute configurations.}
For traditional problems, joint maximum likelihood estimation was usually inconsistent when the sample size goes to infinity (large $N$) but the number of observed variables is fixed (fixed $J$) \citep[][]{neymanscott1948}. 
But in modern large-scale educational assessments, data are collected in an ever-increasing scope involving many student test-takers (large $N$) and many test items (large $J$). For example, the Trends in International Mathematics and Science Study (TIMSS), a series of international assessments of the mathematics and science knowledge, involve students in over 50 countries and have nearly 800 assessment items in total \citep{timss2015enc}.
This scope of data provides new opportunities and requires new methods and understanding of latent variable modeling.

The joint MLE's unique feature of directly incorporating subjects' latent attributes $\bo a_i$'s as parameters to estimate has important and useful practical implications. 
In the applications of SLAMs to cognitive diagnosis \citep{von2019handbook}, estimating each student's latent skill profiles $\bo a_i$ is of great interest as this can provide useful diagnosis of a student's strengths and weaknesses to facilitate better follow-up instructions.
However, most statistical developments of SLAMs \citep{chen2015statistical, xu2017, xu2018jasa, gu2019jmlr} focused on the random-effect versions which marginalize out the $\bo a_i$'s in the likelihood and focus on estimating other quantities, so their identifiability and estimation results do not apply to $\{\bo a_i\}$.
The important questions of what conditions can guarantee the $\{\bo a_i\}$ is consistently estimable and how to estimate this for large-scale data remain unaddressed.
The joint MLE approach considered in this work directly targets at estimating the unknown $\{\bo a_i\}$ and $\QQ$, and we will use this framework to address the aforementioned questions.

Recently, for structured latent factor analysis with {continuous} latent variables, \cite{chenlizhang2019joint} and \cite{chenlizhang2020} studied the joint MLE approach and established identifiability and estimability of continuous latent factors in the double asymptotic regime when $N$ and $J$ both go to infinity. 
However, SLAMs form a different landscape with all the latent variables being discrete. 
Establishing theory for statistical estimability and consistency for discrete latent variables in full generality requires different arguments from those in \cite{chenlizhang2019joint, chenlizhang2020}. 
In addition, new computational methods need to be developed to address the unique challenge of estimation with a large number of discrete latent attributes.

\vspace{-3mm}
\paragraph{Our Contributions.}
We investigate the joint MLE approach to large-scale structured latent attribute analysis, 
and make the following theoretical and methodological contributions.
\begin{enumerate}
	\item We consider the triple-asymptotic regime where all of the $N$, $J$, and $K$ can grow to infinity, for the first time in the literature of SLAMs. In this scenario, we establish the estimability and consistency of both the binary factor loadings in the $\QQ$-matrix and the latent attribute profiles of the subjects $\{a_{i,k}\}$. We also derive finite-sample error bounds for the considered estimators.
	\item We propose a scalable approximate algorithm to compute the joint MLE for two-parameter SLAMs (defined in Example \ref{exp-two}).
	We also propose an efficient two-step estimation procedure for general multi-parameter SLAMs (defined in Example \ref{exp-mult}). This two-step procedure is inspired by investigating a common and interesting type of model oversimplification of SLAMs. When misspecifying a general multi-parameter SLAM to the two-parameter submodel, we show the oversimplified joint MLE can consistently recover part, or even all, of the latent structure under certain conditions.
\end{enumerate}

The rest of the paper is organized as follows. 
Section \ref{sec-setup} introduces the setup of SLAMs  and discusses its connections with other latent variable models.
Section \ref{sec-consist} defines the joint MLE and studies its statistical properties.
Section \ref{sec-comp} proposes scalable algorithms for computing the joint MLE.
Section \ref{sec-simu} provides simulation studies and Section \ref{sec-real} applies our method to a dataset from the TIMSS 2011 Austrian assessment.
Section \ref{sec-disc} gives a discussion.
Technical proofs and additional discussion on computation are included in the Supplementary Material.

\section{Setup of Structured Latent Attribute Models}\label{sec-setup}

\paragraph{General Formulation and Concrete Examples.}
In this paper, we focus on SLAMs for multivariate binary data, which are ubiquitously encountered in educational assessments (correct/wrong answers), social science survey responses (yes/no responses), or biomedical and epidemiological diagnostic tests (positive/negative results).
For $N$ subjects and $J$ variables, collect the observed data in a $N\times J$ binary matrix $\mathbf R=(r_{i,j})$, where $r_{i,j}=1$ or 0 denotes whether the $i$th subject gives a positive response to the $j$th variable. 
Suppose there are $K$ binary latent attributes, then the $J\times K$ binary loading matrix $\QQ = (q_{j,k})$ encodes how the $J$ observed variables depend on the $K$ latent attributes. 
The $N\times K$ binary matrix $\AA = (a_{i,k})$ that stores the latent attribute profiles for the $N$ subjects. 
Both $\QQ$ and $\AA$ have binary entries, where $q_{j,k}=1$ or 0 represents whether the $j$th test item depends on the $k$th latent attribute, and $a_{i,k}=1$ or 0 represents whether the $i$th individual possesses the $k$th attribute.
Generally, a SLAM is a probabilistic model with discrete structures $\QQ$, $\AA$, and additional continuous parameters to specify the generative process of the response data $\RR$.

Denote the additional continuous parameters needed to complete the model specification by $\bo\Theta=\{\ttt_1,\ldots,\ttt_J\}$. Each observed variable $j$ has its continuous parameter vector which we generically denote by $\ttt_j$, whose form depends on the specific model and will be made concrete in Examples \ref{exp-two}--\ref{exp-mult}.
Each observed $r_{i,j}$ follows a Bernoulli distribution with parameter $f(\aa_i,\qq_j,\ttt_j)$ as a function of $\aa_i$, $\qq_j$, and $\ttt_j$.
Given the subjects' latent attribute matrix $\AA$, binary loading matrix $\QQ$, and parameters $\bo\Theta$, the observed responses are assumed to be conditionally independent.
In summary, a SLAM postulates the following statistical model,
\begin{align}\label{eq-bernoulli}
    (r_{i,j}\mid \AA,\QQ,\TT) 
    &\sim 
    \text{Bernoulli}(f(\aa_i,\qq_j,\ttt_j));
    \\ \label{eq-general}
	\mathbb P(\RR \mid \AA,\QQ,\TT)
	&=
	\prod_{i=1}^N \prod_{j=1}^J 
	\left( f(\aa_i,\qq_j,\ttt_j) \right)^{r_{i,j}}
	\left( 1 - f(\aa_i,\qq_j,\ttt_j) \right)^{1 - r_{i,j}}.
\end{align}

\color{black}
\begin{figure}[h!]\centering
\resizebox{0.9\textwidth}{!}{%
\begin{tikzpicture}[scale=1.9]

    \draw [very thick] (0,0) rectangle (2.4/2, 3.2/2);
    
    \filldraw [fill=black!20!white, draw=black!40!black] (0,0) rectangle (2.4/2, 3.2/2);
    
    \draw [step=0.4/2, very thin, color=gray] (0,0) grid (2.4/2, 3.2/2);

    \filldraw [fill=white] (0.4/2, 2.0/2) rectangle (0.8/2, 2.4/2);
    \draw (0.6/2, 2.2/2)  node[text=black] {{\tiny $r_{ij}$}};

    \draw (1.2/2, -0.3) node {{\tiny{$\mathbf R \in \{0,1\}^{N\times J}$}}};

    \draw (3.4/2, 1.7/2) node {\tiny{$\text{probabilistic}\atop\text{model}$}};
    \draw (3.4/2, 1.4/2) node {\tiny{$\Longleftarrow$}};
    \draw (3.4/2, 1.2/2) node {\tiny{$f(\bo a_i,\qq_j,\ttt_j)$}};

    \draw [very thick] (4.4/2, 0) rectangle (5.2/2, 3.2/2);
    \filldraw [fill=black!20!white,draw=black!40!black] (4.4/2,0) rectangle (5.2/2, 3.2/2);
    \draw [step=0.4/2, very thin, color=gray] (4.4/2, 0) grid (5.2/2, 3.2/2);

    \draw (4.1/2, 2.2/2) node[text=black] {{\tiny $\bo a_{i}$}};
    
    \filldraw [fill=white] (4.4/2, 2.0/2) rectangle (4.8/2, 2.4/2);
    \draw (4.6/2, 2.2/2) node[text=black] {{\tiny $a_{i1}$}};
    %
    \filldraw [fill=white] (4.8/2, 2.0/2) rectangle (5.2/2, 2.4/2);
    \draw (5/2, 2.2/2) node[text=black] {{\tiny $a_{i2}$}};
    
    \draw (4.8/2, -0.3) node {{\color{black}\tiny{$\mathbf A \in \{0,1\}^{N\times K}$}}};
    
    \draw (6/2, 1.6/2) node {\tiny{$\text{certain}\atop\text{function}$}};
    
    \draw [very thick] (6.8/2, 1.2/2) rectangle (9.2/2, 2/2);
    \filldraw [fill=black!20!white, draw=black!40!black] (6.8/2, 1.2/2) rectangle (9.2/2, 2/2);
    
    \draw [step=0.4/2, very thin, color=gray] (6.8/2, 1.2/2) grid (9.2/2, 2/2);
    
    \draw (7.4/2, 2.3/2) node[text=black] {{\tiny $\bo q_{j}^\top$}};
    
    \filldraw [fill=white] (7.2/2, 1.6/2) rectangle (7.6/2, 2/2);
    \draw (7.4/2, 1.8/2) node[text=black] {{\tiny $q_{j1}$}};
    %
    \filldraw [fill=white] (7.2/2, 1.2/2) rectangle (7.6/2, 1.6/2);
    \draw (7.4/2, 1.4/2) node[text=black] {{\tiny $q_{j2}$}};

    \draw (8/2, -0.3) node {{\tiny{$\QQ^\top \in \{0,1\}^{K\times J}$}}};
    
\end{tikzpicture}
} 
\caption{A visualization of a SLAM as taking $\AA$ and $\QQ$ as input in a probabilistic model and then generating the data $\RR$. The $K$ equals 2 in the figure.
The entry $r_{ij}$ follows a Bernoulli distribution with parameter $f(\bo a_i, \bo q_j, \ttt_j)$, which is a function of the $i$th row of $\AA$ (denoted by $\bo a_i$), the $j$th row of $\QQ$ (denoted by $\bo q_j$), and continuous parameters $\ttt_j$.
}
\label{fig-matrixfac}
\end{figure}
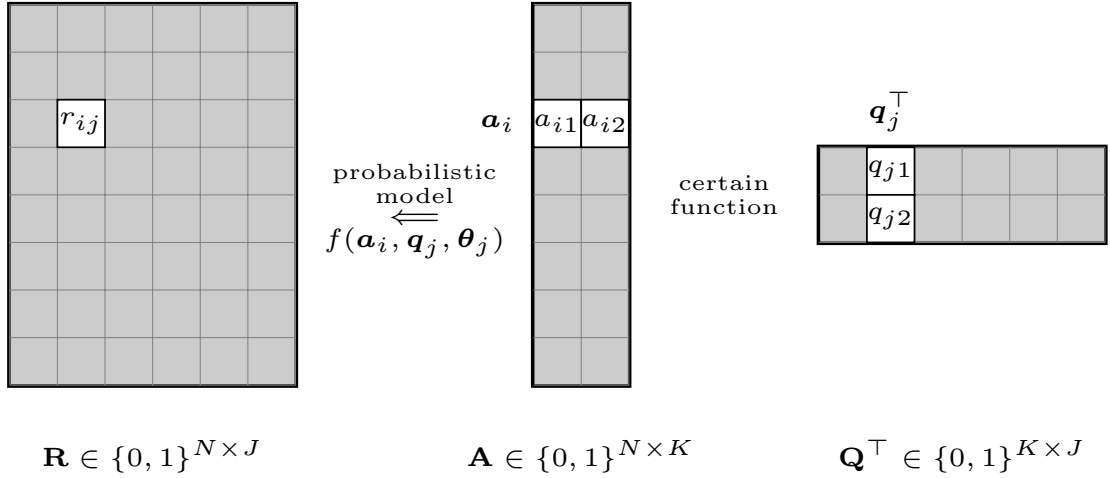

Figure \ref{fig-matrixfac} gives a visualization of a SLAM, making clear how the unknown binary matrices $\AA$ and $\QQ$ underlie the data generating process.
In this paper, we treat both $\AA$ and $\QQ$ as unknown fixed parameters and consider the large-scale scenarios where the number of subjects $N$, the number of observed variables $J$, and the number of latent attributes $K$ all can go to infinity, that is, a triple-asymptotic regime.

We next review two main types of SLAMs widely adopted in the cognitive diagnostic modeling literature: the two-parameter models and the multi-parameter models.

\begin{example}[Two-Parameter SLAMs]\label{exp-two}
For each item $j$, a two-parameter SLAM compactly uses two distinct Bernoulli parameters to model $r_{i,j}$, with $\ttt_j = (\theta_j^+,\theta_j^-)$.
There are two different types of two-parameter SLAMs, the Deterministic Input Noisy output ``And'' (DINA) model proposed in \cite{junker2001cognitive}, and the Deterministic Input Noisy output ``Or'' (DINO) model proposed in \cite{templin2006measurement}.
Under DINA and DINO models, the Bernoulli parameter $f(\aa_i,\qq_j,\ttt_j)$ in \eqref{eq-bernoulli} takes the following specific forms,
\begin{align*}
	f^{\text{DINA}}(\aa_i,\qq_j,\ttt_j)
	&=
	\begin{cases}
		\theta_j^+, &\text{if} ~~ a_{ik}=1 ~\text{for all}~k~\text{such that} ~ q_{jk}=1;\\
		\theta_j^-, &\text{otherwise}.
	\end{cases}
	\\[3mm]
	f^{\text{DINO}}(\aa_i,\qq_j,\ttt_j)
	&=
	\begin{cases}
		\theta_j^+, &\text{if} ~~ a_{ik}=1 ~\text{for at least one}~k~\text{such that} ~ q_{jk}=1;\\
		\theta_j^-, &\text{otherwise}.
	\end{cases}
\end{align*}
DINA is often used in educational testing with latent skills as attributes, and DINO often in psychiatric diagnosis with mental disorders as attributes \citep{dela2018}.
\end{example}

\cite{chen2015statistical} established duality between the DINA and DINO models with $\mathbb P(r_{ij}=1\mid \aa_i=\aaa,~ \qq_j,~ \ttt_j,~ \text{DINO})  = 1 - \mathbb P(r_{ij}=1\mid \aa_i=\one_K - \aaa,~ \qq_j,~ \ttt_j,~ \text{DINA})$ for any $\aaa\in\{0,1\}^K$, where $\one_K$ is a $K$-dimensional all-one vector. Thanks to this duality, identifiability and estimation results developed under DINA easily carry over to the DINO case. So without loss of generality, next we focus on the DINA model when studying two-parameter SLAMs.

\begin{example}[Multi-Parameter SLAMs]\label{exp-mult}
Unlike a two-parameter model, a multi-parameter SLAM models each observed variable $j$ using potentially more than two Bernoulli parameters. 
The $f(\aa_i,\qq_j,\ttt_j)$ in \eqref{eq-bernoulli} now takes the form
\begin{align}\label{eq-alleff}
	f^{\mult}(\aa_i,\qq_j,\ttt_j) = 
	f\Big( & \mu_{j,0} + \sum_{k=1}^K \mu_{j,k} (q_{jk} a_{ik}) + \sum_{1\leq k_1 < k_2\leq K} \mu_{j,k_1k_2} (q_{jk_1} a_{ik_1}) (q_{jk_2}a_{ik_2})
	\\ \notag
	& 
	+ \cdots + \mu_{j,12\cdots K} \prod_{k=1}^K (q_{jk} a_{ik})
	\Big)
\end{align}
where different link functions $f(\cdot)$ lead to different specific models; when $f(\cdot)$ is the identity, \eqref{eq-alleff} gives the Generalized DINA model \citep[GDINA,][]{dela2011}; when $f(\cdot)$ is the logistic function, \eqref{eq-alleff} gives the Log-linear Cognitive Diagnosis Models \citep[LCDMs,][]{HensonTemplin09}; see also the General Diagnostic Models \citep[GDMs,][]{davier2008general}. 
Note that in \eqref{eq-alleff}, not all the $\mu$-coefficients are meaningful and need to be incorporated into the model; for example, if $q_{j,k} = 0$ then $\mu_{j,k}$ is not needed and if $q_{jk_1} q_{jk_2} = 0$ then $\mu_{j,k_1k_2}$ is not needed, etc.
In multi-parameter SLAMs, the continuous parameter vector $\ttt_j$ is the collection of all the meaningful $\mu_{j,\bcdot}$-coefficients.
Multi-parameter models under \eqref{eq-alleff} are quite general, as they incorporate all the possible main and interaction effects of the meaningful latent attributes. 
\end{example}

Examples \ref{exp-two} and \ref{exp-mult} imply that the two-parameter model can be viewed as a submodel of the multi-parameter model.
To see this, just set all the $\mu$-coefficients in \eqref{eq-alleff} to zero except $\mu_{j,0}$ and the highest order term $\mu_{j, \,\text{high}} := \mu_{j, \,\{k:\, q_{jk}=1\}}$, then $\theta_{j}^- = f(\mu_{j,0})$ and $\theta_{j}^+ = f(\mu_{j,0} + \mu_{j,\, \text{high}})$ correspond to the two parameters for variable $j$ defined in Example \ref{exp-two}.

\paragraph{Connections between SLAMs and Other Latent Variable Models.}
We briefly review the family of latent variable models and locate SLAMs within this context.
Latent variable models can be categorized into four types according to the nature of the observed and the latent variables. 
When the observed and latent variables are both continuous, the factor analysis \citep{anderson1956fa} has been widely used. 
When the observed variables are discrete but the latent variables are continuous, the Item Response Theory (IRT) models \citep{embretson2013item, reckase2009mirt} are typical modeling choices.
On the other hand, to model continuous observed data using a discrete latent variable, researchers have employed mixture models such as the Gaussian mixtures \citep{reynolds2000gmm}. 
Finally, when both the observed variables and the latent one are discrete, the latent class model has been a popular modeling tool since decades ago \citep{lazarsfeld1968latent}.

SLAMs can be viewed as a modern generalization of latent class models (LCMs), in that both adopt discrete latent structure to model discrete data. 
Despite this similarity, the following two key characteristics distinguish SLAMs from traditional LCMs: 
(a) the discrete latent constructs in SLAM are multidimensional instead of  unidimensional as in LCMs;
and (b) a SLAM models dependence of the observed variables on the latent ones by a binary loading matrix $\QQ$.
Figure \ref{fig-graph} provides graphical model representations of LCMs and SLAMs that highlight their connections and differences.
Both LCMs and SLAMs assume the multivariate categorical observations $\bo r_i=(r_{i,1},\ldots,r_{i,J})$ are conditionally independent given the latent part.
When modeling the observed $\bo r_i$, an LCM in Figure \ref{fig-graph}(a) adopts a unidimensional latent variable $z_i\in\{1,\ldots,C\}$ while a SLAM in Figure \ref{fig-graph}(b) adopts a $K$-dimensional binary latent vector $\bo a_i = (a_{i,1},\ldots,a_{i,K}) \in \{0,1\}^K$. 
Therefore, an LCM does not necessarily distinguish the $C$ latent classes by definition, 
while a SLAM naturally defines $2^K$ distinct latent classes, each as a pattern detailing the statuses of $K$ fine-grained traits.

Additionally and perhaps more importantly, a SLAM has the key $J\times K$ binary loading matrix $\QQ=(q_{j,k})$, where $q_{j,k}=1$ means observed $r_{i,j}$ depends on the latent $\alpha_{i,k}$ and $q_{j,k}=0$ otherwise. 
Such dependence encoded in $\QQ$ can be equivalently represented as a bipartite graph from the latent to the observed variables as illustrated in Figure \ref{fig-graph}(b).
{As shown in this figure, each observed variable (shaded node) can depend on multiple different latent variables (white nodes). In the literature, the multidimensional IRT models proposed in \cite{bartolucci2007class} and \cite{bacci2016two} also assume multiple discrete latent variables explain multivariate categorical data. Compared with those approaches, SLAMs are more general in that each observable is not restricted to depending on only one latent variable, but rather can depend on multiple ones. This is exactly reflected by the fact each row of $\QQ$ can contain an arbitrary number of ``1''s. In the application to educational assessment, this means each test item can target multiple different latent skills.
Thus the matrix $\QQ$ imposes meaningful and flexible constraints on the parameters to enhance model interpretability.}

\color{black}
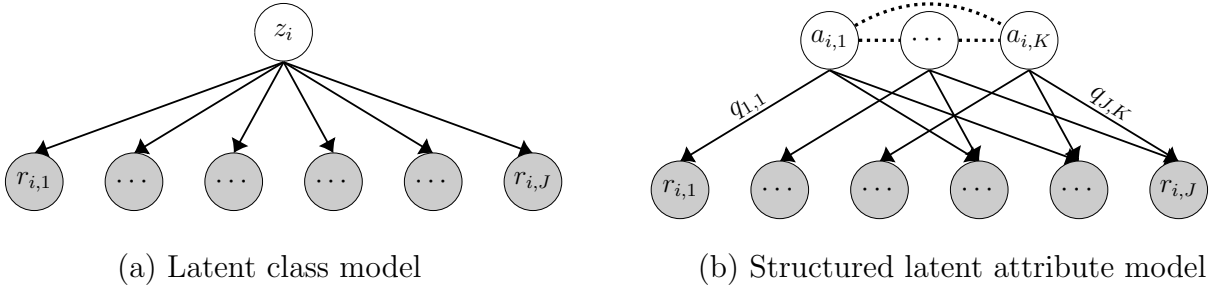
\begin{figure}[h!]
	\centering

\begin{minipage}[c]{0.48\textwidth}
\centering
    	
\resizebox{0.95\textwidth}{!}{%
\begin{tikzpicture}[scale=1.9]

\draw[arr](2,0.96) -- (0,0.23) node [midway,above=-0.12cm,sloped] {};
\draw[arr](2,0.96) -- (0.8,0.23) node [midway,above=-0.12cm,sloped] {};
\draw[arr](2,0.96) -- (1.6,0.23) node [midway,above=-0.12cm,sloped] {};
\draw[arr](2,0.96) -- (2.4,0.23) node [midway,above=-0.12cm,sloped] {};
\draw[arr](2,0.96) -- (3.2,0.23) node [midway,above=-0.12cm,sloped] {};  
\draw[arr](2,0.96) -- (4,0.23) node [midway,above=-0.12cm,sloped] {};  

    \node (h2)[hidden] at (2.0, 1.2) {$z_i$};
    \node (h1)[circle,minimum size=20pt,inner sep=0pt, fill=white] at (1.2, 1.2) {};
    \node (h3)[circle,minimum size=20pt,inner sep=0pt, fill=white] at (2.8, 1.2) {};
    
    \node (v1)[neuron] at (0, 0) {$r_{i,1}$};
    \node (v2)[neuron] at (0.8, 0) {$\cdots$};
    \node (v3)[neuron] at (1.6, 0) {$\cdots$};
    \node (v4)[neuron] at (2.4, 0) {$\cdots$};
    \node (v5)[neuron] at (3.2, 0) {$\cdots$};
    \node (v6)[neuron] at (4, 0) {$r_{i,J}$};

    \fill [dotted,line width=0.5mm,color=white] (h1) -- (h2);
    \fill [dotted,line width=0.5mm,color=white] (h2) -- (h3);

\end{tikzpicture}
}

\end{minipage}
	\hfill
\begin{minipage}[c]{0.48\textwidth}\centering
	
\resizebox{0.95\textwidth}{!}{%
\begin{tikzpicture}[scale=1.9]

\draw[arr](1.2,0.96) -- (0,0.23) node [midway,above=-0.12cm,sloped] {$q_{1,1}$};
\draw[arr](1.2,0.96) -- (2.4,0.23) node [midway,above=-0.12cm,sloped] {};
\draw[arr](1.2,0.96) -- (3.2,0.23) node [midway,above=-0.12cm,sloped] {};  
    
\draw[arr](2,0.96) -- (0.8,0.23) node [midway,above=-0.12cm,sloped] {};
\draw[arr](2,0.96) -- (2.4,0.23) node [midway,above=-0.12cm,sloped] {};
\draw[arr](2,0.96) -- (4,0.23) node [midway,above=-0.12cm,sloped] {};
    
\draw[arr](2.8,0.96) -- (1.6,0.23) node [midway,above=-0.12cm,sloped] {};
\draw[arr](2.8,0.96) -- (3.2,0.23) node [midway,above=-0.12cm,sloped] {};  
\draw[arr](2.8,0.96) -- (4,0.23) node [midway,above=-0.12cm,sloped] {$q_{J,K}$}; 

    \node (v1)[neuron] at (0, 0) {$r_{i,1}$};
    \node (v2)[neuron] at (0.8, 0) {$\cdots$};
    \node (v3)[neuron] at (1.6, 0) {$\cdots$};
    \node (v4)[neuron] at (2.4, 0) {$\cdots$};
    \node (v5)[neuron] at (3.2, 0) {$\cdots$};
    \node (v6)[neuron] at (4, 0) {$r_{i,J}$};

    \node (h1)[hidden] at (1.2, 1.2) {$a_{i,1}$};
    \node (h2)[hidden] at (2.0, 1.2) {$\cdots$};
    \node (h3)[hidden] at (2.8, 1.2) {$a_{i,K}$};

    \draw[dotted,line width=0.5mm] (h1) -- (h2);
    \draw[dotted,line width=0.5mm] (h2) -- (h3);
    \path (h1) edge[dotted,bend left,line width=0.5mm] node [right] {} (h3);

    \end{tikzpicture}
    }	

	\end{minipage}

	\bigskip
	\begin{minipage}[c]{0.45\textwidth}\centering
    (a) Latent class model
	\end{minipage}
    \hfill	
    \begin{minipage}[c]{0.45\textwidth}\centering
    (b) Structured latent attribute model
	\end{minipage}
	
	\caption{Graphical model representations of a traditional latent class model in (a) and a structured latent attribute model in (b).
	All the variables are discrete and the directed solid edges encode conditional dependences.
	In $\QQ=(q_{j,k})_{J\times K}$, the $q_{j,k}=1$ if there is a directed edge from latent attribute $a_{i,k}$ to observed variable $r_{i,j}$.}
	\label{fig-graph}
\end{figure}


The two key features of SLAMs, the \textit{multidimensional discreteness} and \textit{structured dependence}, are motivated by the needs of modern scientific applications, and in turn bring great flexibility and interpretability to real-world modeling practices. 
The latent $\bo a_i = (a_{i,1}, \ldots, a_{i,K})$  summarizes subject $i$'s statuses of multiple latent traits; in educational testing, the skill profile ${\bo a}_i$ provides useful diagnostic feedback by detailing what latent skills each student has/has not mastered; in medical or epidemiological diagnosis, the ${\bo a}_i$ details each patient's presence/absence of certain latent diseases or pathogens.
Such fined-grained profiles form a foundation for designing effective follow-up intervention or treatment.
Furthermore, the structural dependence encoded in the $\QQ$-matrix can represent practitioners' prior knowledge, facilitate dimension reduction, and enhance interpretability. 
In summary, SLAMs enable uncovering hidden fine-grained scientific information, providing model-based clustering of subjects, and facilitating better intervention. These advantages distinguish SLAMs from traditional models such as IRT models or LCMs, and make SLAMs and their variants suitable for a variety of modern applications, including not only education and psychology \citep{chen2015statistical, xu2018jasa, gu2019jmlr}, but also epidemiology \citep{ o2019causes} and biomedicine \citep{ni2020double, chen2020itr}.

\section{Joint MLE and Its Statistical Properties}\label{sec-consist}

\paragraph{Definition of Joint MLE.}
We next formally introduce the joint maximum likelihood estimator.
Under the general setup in \eqref{eq-general}, the log of the joint likelihood of $(\AA,\,\QQ,\,\TT)$ is
\begin{align}\label{eq-loglike}
	\ell^{\text{m}}(\AA,\,\QQ,\,\TT\mid\RR)
	=
	\sum_{i=1}^N \sum_{j=1}^J \left[r_{i,j}\log(f^{\text{m}}(\aa_i,\qq_j,\ttt_j)) + (1-r_{i,j})\log(1 - f^{\text{m}}(\aa_i,\qq_j,\ttt_j))\right],
\end{align}
where the superscript ``m'' denotes a specific model, e.g., a two-parameter or multi-parameter model reviewed in Examples \ref{exp-two}-\ref{exp-mult}.
The joint MLE approach has an important feature that the subjects' latent attributes $\AA=(a_{i,k})$ are incorporated as unknown parameters to estimate. This is different from the marginal MLE which marginalize out the $a_{i,k}$'s and focus on estimating other quantities. 
Indeed, in the applications of SLAMs to cognitive diagnostic modeling \citep{von2019handbook}, inferring the students' latent skill profiles is of great interest as they can provide useful diagnosis of a student's strengths and weaknesses for navigating better follow-up instructions.
However, most statistical developments of SLAMs \citep{chen2015statistical, xu2017, xu2018jasa, gu2019jmlr} focused on the random-effect versions, so their results typically do not apply to the underlying $\AA$.
The important questions of what conditions can guarantee the $\AA$ is consistently estimable and how to estimate it for large-scale data remain unaddressed.
To this end, the joint MLE approach considered in this work directly targets at estimating the unknown quantities $\AA$ and $\QQ$ and provides a natural basis for addressing these questions.

Given the general log-likelihood in \eqref{eq-loglike}, define the joint MLE under a specific SLAM as
\begin{align}\label{eq-prob1}
	(\hat\AA,\,\hat\QQ,\,\hat\TT)^{\text{m}} = 
	&~\argmax_{(\AA,\QQ,\TT)} \ell^{\text{m}}(\RR;\,\AA,\,\QQ,\, \TT)\\ 
	\label{eq-upper}
	&~~\text{subject to fitting a $K$-attribute specified SLAM with $\sum_{k=1}^K\hat q_{j,k} \leq B_j$,}
\end{align}
where $B_j$'s are prespecified upper bounds depending on the model, imposed for theoretical identifiability reasons.
{An interesting study \cite{bonhomme2015grouped} also considered the fixed-effect estimation of discrete latent heterogeneity, motivated by panel data in econometrics. In the regime where the number of subjects $N$ and number of time points $T$ both go to infinity, \cite{bonhomme2015grouped} considered \textit{continuous} data and \textit{unidimensional} discrete heterogeneity. Different from that, in this work when $N$ and $J$ go to infinity, we consider multivariate \textit{categorical} data and \textit{multidimensional} discrete latent features. Therefore, the least squares estimation criterion used in \cite{bonhomme2015grouped} is not applicable here, and we need to seek estimators based on the specific likelihood functions.}

We next make several important remarks about the nuances and differences between estimating two-parameter and multi-parameter SLAMs, in terms of both identifiability and computation.

\begin{remark}[Solve \eqref{eq-prob1} under a two-parameter model]\label{rmk-2p}
\normalfont{
Based on the setup in Example \ref{exp-two}, the two-parameter log-likelihood in \eqref{eq-prob1} can be written in the following explicit form, 
\begin{align}\label{eq-ll-2p}
\ell^{\text{\normalfont{two}}}(\QQ,\,\AA,\,\TT\mid\RR)
=& \sum_{i=1}^N \sum_{j=1}^J \Big[ r_{i,j}\Big(  \prod_{k=1}^K a_{i,k}^{q_{j,k}} \log \theta^+_j + (1-\prod_{k=1}^K a_{i,k}^{q_{j,k}}) \log \theta^-_j \Big)  \\ \notag
&  + (1-r_{i,j}) \Big(  \prod_{k=1}^K a_{i,k}^{q_{j,k}} \log(1- \theta^+_j) + (1-\prod_{k=1}^K a_{i,k}^{q_{j,k}}) \log(1- \theta^-_j) \Big)\Big].
\end{align}
Under the two-parameter likelihood, when solving \eqref{eq-prob1} for $(\hat\AA,\,\hat\QQ,\,\hat\TT)^{\text{\normalfont{two}}}$, we impose a natural constraint $\hat\theta_{j}^+ > \hat\theta_{j}^-$ to ensure identifiability \citep{junker2001cognitive, gu2019dina}. 
A careful inspection of the special combinatorial form under the two-parameter DINA model \eqref{eq-ll-2p} reveals that the upper bound $B_j$ in the optimization problem \eqref{eq-prob1} can be taken as $B_1=\cdots=B_J=B^{\two}=\infty$. That is, there is essentially no need to constrain the number of ``1''s in the estimation of $\QQ$.
To solve \eqref{eq-prob1} under the two-parameter likelihood, we propose a scalable approximate EM-flavor algorithm. This algorithm treats the unknown discrete structures $(\QQ, \AA)$ as missing data to impute in an approximate E (Expectation) step which is based on a few Gibbs samples, and treats continuous parameters $\bo\Theta=\{\bo\theta_j^+, \bo\theta_j^-: 1\leq j\leq J\}$ as model parameters to update in an M (Maximization) step; see Section \ref{sec-est-2p} for details.
}
\end{remark}

\begin{remark}[Solve \eqref{eq-prob1} under a multi-parameter model]\label{rmk-qid}
\normalfont{
Under a multi-parameter SLAM, denote the true binary loading matrix by $\QQ^\true=(q^\true_{j,k})$, and we take the upper bound in \eqref{eq-upper} to be $B_j = B^{\mult}_j={\sum}_{k=1}^K q_{j,k}^\true$.
Under a multi-parameter model \eqref{eq-alleff}, $\QQ$ captures the sparsity structure of the underlying continuous parameters $\bo\mu_j$,
so the constraint ${\sum}_{k=1}^K\hat q_{j,k} \leq B_j = {\sum}_{k=1}^K q_{j,k}^\true$ in \eqref{eq-upper} resembles the $L_0$ constraint on regression coefficients in regression problems for variable selection. 
Theoretically, such a constraint is necessary to ensure $\qq_j$'s are identifiable under a multi-parameter SLAM. 
To see this, consider a toy example with $\qq^\true_j=(1,0)$, then the multi-parameter model with an identity link in Example 2 gives
    \begin{align*}
        \mathbb P(r_{i,j}=1\mid \aa_i,\qq_j^\true,\bo\mu_j) 
        =&~ \mu_{j,0} + \mu_{j,1}q^\true_{j,1}a_{i,1} + \mu_{j,2}q^\true_{j,2}a_{i,2} + \mu_{j,12}(q^\true_{j,1}a_{i,1})(q^\true_{j,2}a_{i,2})\\
        =&~\mu_{j,0} + \mu_{j,1}a_{i,1};
    \end{align*}
    while with an alternative $\tilde\qq_j=(1,1)$ and $\tilde{\bo\mu}_j=(\tilde\mu_{j,0}, \tilde\mu_{j,1}, \tilde\mu_{j,2}, \tilde\mu_{j,12})$ where $\tilde\mu_{j,2}=\tilde\mu_{j,12}=0$, 
    \begin{align*}
        \mathbb P(r_{i,j}=1\mid \aa_i, \tilde\qq_j, \tilde{\bo\mu}_j) 
        =&~ \mu_{j,0} + \tilde\mu_{j,1}\tilde q_{j,1}a_{i,1} + \tilde\mu_{j,2}\tilde q_{j,2}a_{i,2} + \tilde\mu_{j,12}(\tilde q_{j,1}a_{i,1})(\tilde q_{j,2}a_{i,2})\\
        =&~ \tilde\mu_{j,0} + \tilde\mu_{j,1}a_{i,1}
        = \mu_{j,0} + \mu_{j,1}a_{i,1}.
        \quad(\text{if}~~\tilde\mu_{j,0} = \mu_{j,0}~\text{and}~\tilde\mu_{j,1} = \mu_{j,1})
    \end{align*}
    This example illustrates that despite $\qq^\true_j=(1,0) \neq \tilde\qq_j=(1,1)$, the distribution of $r_{i,j}$ given the two are identical, indicating non-identifiability. 
    Therefore theoretically, we need to constrain the number of ``1''s in $\QQ$ for identifiability when the model is multi-parameter.
    
    Although the constraint ${\sum}_{k=1}^K\hat q_{j,k} \leq {\sum}_{k=1}^K q_{j,k}^\true$ is needed for theoretical identifiability under multi-parameter models as stated above, practically, the constrained optimization problem \eqref{eq-upper} can be replaced by an unconstrained one by imposing an appropriate penalty.
Indeed, our estimation method for multi-parameter SLAMs does not assume knowledge of the true values of ${\sum}_{k=1}^K q_{j,k}^\true$, but rather adopts marginal screening and variable selection approaches to directly estimate the entries of $\QQ$ in a second regression stage, following a first stage of approximate estimation of latent attributes in $\AA$; see Section \ref{sec-est-mult} for details.
    }
\end{remark}

\paragraph{Theoretical Properties of the Joint MLE.}
From now on, we consider the model sequence indexed by $(N,J,K)$, where each of  $N,~ J,$ and $K$ can go to infinity.
Thus far we have treated $\bo\Theta$ as a generic notation for continuous parameters in any specific SLAM. 
For technical convenience, we next fix the notation of $\bo\Theta$ as a $J\times 2^K$ matrix $\bo\Theta = (\theta_{j, \aaa})$, where 
\begin{equation}\label{eq-thetaj}
   \theta_{j,\aaa} = \mathbb P(r_{i,j}=1\mid \aa_i=\aaa,\; \qq_j,\;\text{specific model})
\end{equation}
for $j\in[J]$ and $\aaa\in\{0,1\}^K$. The  expressions of $\theta_{j,\aaa}$ under specific two- or multi-parameter models can be easily derived based on Examples \ref{exp-two}--\ref{exp-mult}. 
The following assumptions are made on the true parameters $(\TT^{\true}, \QQ^{\true}, \AA^{\true})$ that generate the data.

\begin{assumption}
	\label{assume-pij}
There exists a finite number $d\geq 2$ such that
	\begin{equation}\label{eq-pijb}
		\frac{1}{J^d} \leq 
	\min_{1\leq j\leq J,\atop \aaa\in\{0,1\}^K} \theta^{\true}_{j,\aaa}
	\leq 
	\max_{1\leq j\leq J,\atop \aaa\in\{0,1\}^K} \theta^{\true}_{j,\aaa}
	\leq 1-\frac{1}{J^d}.	
	\end{equation}
\end{assumption}

\begin{assumption}\label{cond-gap-2p}
For two-parameter SLAMs, suppose $\theta^{+,\true}_{j} > \theta^{-,\true}_{j}$ for each $j$ and that there exists $\{\beta_{J}\}\subseteq(0,\infty)$ such that
	\begin{align}
\label{eq-as-gen2p}
\min_{1\leq j\leq J}~
  \left(\theta^{+,\true}_{j} - \theta^{-,\true}_{j}
        \right)^2 \geq &~ \beta_J.
	\end{align}
For multi-parameter SLAMs, there exists $\{\beta_{J}\}\subseteq(0,\infty)$ such that
	\begin{align}
\label{eq-as-gen2}
\min_{1\leq j\leq J}~
  \left\{ \min_{
              \aaa\circ\qq_j \neq \aaa'\circ\qq_j} 
 \left(\theta^{\true}_{j,\aaa} - \theta^{\true}_{j,\aaa'}\right)^2
        \right\}\geq &~\beta_J,
	\end{align}
where $\aaa\circ\qq_j = (\alpha_{1}q_{j,1}, \ldots, \alpha_{K}q_{j,K})$ denotes element-wise product of binary vectors $\aaa$ and $\qq_j$.
\end{assumption}

\begin{assumption}
	\label{cond-id}
There exist $\{\delta_{J}\}$, $\{p_{N}\}\subseteq(0,\infty)$ and a constant $\epsilon>0$ such that
	\begin{align}
	\label{eq-infq}
		\min_{1\leq k\leq K}\frac{1}{J} \sum_{j=1}^J I(\qq^{\true}_j=\ee_k) 
		\geq &~ \delta_{J};\\
		\label{eq-pmin}
		\min_{\aaa\in\{0,1\}^K}
         \frac{1}{N}
		\sum_{i=1}^N I(\aa^{\true}_i=\aaa) \geq &~ p_{N} \geq \frac{\epsilon}{2^K}.
	\end{align}
	Also assume $\sum_{k=1}^K q_{j,k}^{\true} \leq K_0$ for a constant $K_0$.
\end{assumption}

Note that by writing all the lower bounds in the above assumptions as depending on a subscript $J$ or $N$, we indeed allow them to go to zero as $J$ and $N$ go to infinity.
This type of assumptions distinguish the current theoretical investigation from all the previous works on SLAMs \citep[e.g.,][]{chen2015statistical, xu2018jasa, gu2019jmlr}. 
As to be shown in the following theorems, as long as the rate at which these $\beta_J$, $\delta_J$, and $p_N$ go to zero satisfy some mild requirements, consistency of joint MLE can be ensured.

More detailed discussions on the assumptions are in order. 
Assumption \ref{assume-pij} is a very mild condition on the Bernoulli parameters $\theta_{j,\aaa}$'s. 
Assumption \ref{cond-gap-2p} lower bounds the gap of Bernoulli parameters for different latent classes, under the two- and multi-parameter SLAM, respectively.
Such a gap $\beta_J$ measures how separated the latent classes are and hence quantifies how strong the signals are.
This assumption has its counterpart in the finite-$J$ regime; e.g., \cite{xu2018jasa} and \cite{gu2019dina, gu2020partial} imposed $\beta_J>0$ when studying identifiability. Instead, here we allow $\beta_J\to 0$ and establish estimability and consistency.
Assumption \ref{cond-id} is about the discrete structures $\QQ$ and $\AA$, where \eqref{eq-infq} resembles a requirement that ``$\QQ$ should contain an identity submatrix $I_K$'' in the studies with finite $J$ \citep{chen2015statistical, xu2018jasa}.
Here with $J\to\infty$, a finite number of submatrices $I_K$ in $\QQ$ may not suffice for estimability and consistency, and \eqref{eq-infq} requires $\QQ$ to contain an increasing number of $I_K$'s as $J$ grows.
In the literature, \cite{wang2015nonpa} made a similar assumption on $\QQ$ when establishing consistency of a nonparametric estimator for cognitive diagnostic models, and \cite{chenlizhang2020} also imposed a similar requirement on the loading matrix when studying continuous latent factor models.
{Note that this requirement \eqref{eq-infq} implies the matrix $\QQ$ does not contain any all-zero column.
Theoretically, if $\QQ$ contains some all-zero column, then the model is not identifiable. This is because dropping this all-zero column of $\QQ$ and reducing the number of latent attributes by one will give the same distribution of the observables.}
As for another requirement \eqref{eq-pmin} in Assumption \ref{cond-id}, it implies the $2^K$ latent patterns do not exhibit too uneven proportions in the sample. A resemblance for this requirement in random-effect SLAMs is $p_{\aaa}>0$ for all $\aaa\in\{0,1\}^K$, where $p_{\aaa}$ denotes the population proportion of latent pattern $\aaa$.

Denote $M = (NJ)^{-1}\sum_{i=1}^N\sum_{j=1}^J \MP(r_{i,j}=1\mid \text{true model})$, the average positive response rate in the sample.
The following main theorem establishes the consistency and bounds the rate of convergence of joint MLE in recovering the latent structure.

\begin{theorem}[Consistency of joint MLE under either two- or multi-parameter model]
\label{thm-joint-both}
Consider either a two-parameter or a multi-parameter SLAM with $(\hat\QQ,\hat\AA)$ obtained from solving \eqref{eq-prob1}. 
When $N,J\to\infty$, suppose 
$\sqrt{J}=O\left(\sqrt{M}N^{1-c}\right)$ for some small constant $c\in(0,1)$ and $K=o(MJ\log J)$.
Under Assumptions \ref{assume-pij}, \ref{cond-gap-2p}, and \ref{cond-id}, the following two conclusions hold.
\begin{itemize}
\item[(a)] There is
\begin{align}
\notag
\frac{1}{NJ}
\sum_{i=1}^N\sum_{j=1}^J \Big(\MP(r_{i,j}=1\mid \QQ^{\true},\AA^{\true},\TT^{\true}) - \MP(r_{i,j}=1\mid \hat \QQ,\hat\AA,\TT^{\true}) \Big)^2
= o_P\left(\frac{\gamma_{J}}{\beta_J}\right),	
\end{align}
where for a small positive constant $\epsilon>0$,
\begin{equation}\label{eq-defgamma}
	\gamma_{J} = \frac{(\log J)^{1+\epsilon}}{\sqrt{J}}\cdot\sqrt{M\log (2^K)}.
\end{equation}

\item[(b)] 
Up to a permutation of the $K$ latent attributes, there is
	\begin{align}\label{eq-errqa}
	\frac{1}{J}
	\sum_{j=1}^J I(\qq_j^{\true} \neq \hat\qq_j)
	= o_P\left(\frac{\gamma_{J}}{\beta_J \cdot p_N}\right),
	\qquad
	\frac{1}{N}\sum_{i=1}^N I(\hat\aa_i\neq \aa_i^{\true})
    = o_P\left(\frac{\gamma_{J}}{\beta_J \cdot \delta_J}\right).
\end{align}
\end{itemize}
\end{theorem}

\begin{remark}\label{rmk-continuous}
For large-scale continuous latent factor analysis, \cite{chenlizhang2019joint, chenlizhang2020} exploited the \textit{low-rank-matrix} structure to establish consistency and bound convergence rate of MLE.
	In a continuous latent factor model, the low-rankness usually exactly captures the intrinsic characteristic of the model; for example, the latent structure is summarized as an inner product term $\Theta_{N\times K} A_{J\times K}^\top$ in \cite{chenlizhang2020} (where $\Theta_{N\times K}$ collect the continuous person-parameters and $A_{J\times K}$ collect the continuous item-parameters), which is a matrix with low rank $K$.
	However, for discrete latent variable models, especially the complicated SLAMs considered here, the low-rankness is often a too rough and sometimes imprecise summary of the latent structure. This is because discrete latent structure ($\AA$ and $\QQ$ here) would induce an \textit{unobserved partition} of data underlying a probabilistic model, which is not the case when latent variables are continuous and hence requires different analysis.
\end{remark}

\begin{remark}
The proof of part (a) of Theorem \ref{thm-joint-both} uses a similar technique as the profile likelihood approach  in the network community detection literature \citep[see, e.g.][]{choi2012sbm, zhao2020id}.
A proof technique of a similar spirit is useful here because the existence of discrete latent variables allows reformulating the maximum likelihood problem \eqref{eq-prob1} as performing certain model-based clustering. Indeed, each vector $\qq_j$ categorizes the attribute patterns into distinct clusters locally for each observed variable $j$, in different ways under different model assumptions in Examples \ref{exp-two} and \ref{exp-mult}.
Notably, also apparent from Examples \ref{exp-two} and \ref{exp-mult} is that the model setup of a SLAM is fundamentally different from a stochastic block model for network data. 
The unobserved partition in SLAMs are more subtle to deal with than other simpler discrete latent variable models (including the network community models) due to the parameter constraints imposed by the $\QQ$-matrix. 
The overall proof procedure used to establish Theorem \ref{thm-joint-both} needs to take into account such unique parameter constraints. 
\end{remark}

The $\gamma_{J}$ in part (a) of Theorem \ref{thm-joint-both} bounds the error of recovering the average positive response probability under the estimated $(\hat\QQ,~\hat\AA)$. Part (b) further separately bounds the errors of the estimators for $\hat\QQ$ and $\hat\AA$, respectively.
The derived rates in Theorem \ref{thm-joint-both}(b) imply that the sequences $\{\beta_J\}$, $\{p_N\}$, and $\{\delta_J\}$ are allowed to go to zero while still guaranteeing consistency, as long as $\gamma_{J}/(\beta_J\cdot p_N)\to 0$ and $\gamma_{J}/(\beta_J\cdot \delta_J)\to 0$.
Theorem \ref{thm-joint-both}(b) not only ensures the asymptotic consistency of joint MLE, but also offers insight into the accuracy of estimating $\QQ$ and $\AA$ with finite samples and finite $J$.
In particular, if $\beta_J$ and $\delta_J$ are constants and $K$ is finite, then the finite sample error bounds in \eqref{eq-errqa} become $O((\log J)^{1+\epsilon} \cdot J^{-1/2})$.

\section{Scalable Estimation Algorithms}\label{sec-comp}

This section presents algorithms for computing the joint MLE for two-parameter and multi-parameter SLAMs.
Recall that a two-parameter SLAM can be viewed as a submodel for a multi-parameter one. The succinct form of two-parameter models allows for developing a scalable approximate estimation approach, and the next Section \ref{sec-est-2p} proposes an algorithm specifically tailored for two-parameter models.
Then Section \ref{sec-est-mult} builds on this algorithm and further provides an estimation approach for the more general multi-parameter models.

\vspace{-1mm}
\subsection{Estimation under the Two-Parameter Model}\label{sec-est-2p}
EM algorithms \citep{dempster1977}  are popular methods for latent variable model estimation.
For SLAMs, a traditional EM algorithm for computing the marginal MLE under a random-effect model assumes each $\aa_i$ follows a categorical distribution with $|\{0,1\}^K|=2^K$ components. In this setup, the E step updates the probabilities of each $\aa_i$ being each possible pattern in $\{0,1\}^K$. 
The cardinality of this space $|\{0,1\}^K| = 2^K$ grows exponentially with $K$, so evaluating all the $\aa_i$'s and $\qq_j$'s probabilities of being all the possible configurations has complexity $O((N+J) 2^K)$ in each EM iteration. 
This incurs high computational cost for moderate to large $K$. 
{On the other hand, here we consider the joint MLE for fixed-effect SLAMs and treat subjects' latent attributes $a_{i,k}$'s and also $q_{j,k}$'s as parameters. This formulation requires different estimation procedures from the traditional EM for computing the marginal MLE.
We next propose a new EM-flavor algorithm with a stochastic component well suited to the considered scenario. The new algorithm directly targets at estimating the individual $a_{i,k}$'s and $q_{j,k}$'s, and further uses a stochastic step in order to scale to high-dimensional data.}
In particular, our new algorithm draws a few Gibbs samples of the entries of the discrete $\QQ$ and $\AA$ in an approximate E step to achieve scalability.

The details of the algorithm are as follows.
The entries of $\QQ$, $\AA$ are treated as missing data to be imputed in an approximate E step, and the continuous $\bo\Theta=(\ttt^+, \ttt^-)$ are treated as model parameters to be updated in an M step.
In the approximate E step,
we propose to take an approximation by drawing a few (denote the number by $C$) Gibbs samples of entries of $\AA = (a_{i,k})$ (along the direction of updating subjects' patterns), and then draw $C$ Gibbs samples of entries of $\QQ = (q_{j,k})$ (along the direction of updating variables' loadings).
Under the two-parameter log-likelihood in \eqref{eq-ll-2p}, given the current iterates of parameters $(\ttt^+,\ttt^-)$, the conditional distributions of each $a_{i,k}$ and $q_{j,k}$ from which we draw the Gibbs samples are
\begin{align*}
    \mathbb P(a_{i,k} = 1\mid -)
    &=\sigma\left(-\sum_{j=1}^J q_{j,k} \prod\limits_{1\leq m\leq K\atop m\neq k} a_{i,m}^{q_{j,m}} 
    \left[r_{i,j} \log\left(\frac{\theta^+_j}{ \theta^-_j}\right) + (1-r_{i,j})\log\left(\frac{1-\theta^+_j}{1-\theta^-_j}\right)\right]\right);\\[2mm]
    \mathbb P(q_{j,k} = 1\mid -)
    &=\sigma\left(\sum_{i=1}^N (1-a_{i,k}) \prod_{1\leq m\leq K\atop m\neq k} a_{i,m}^{q_{j,m}}
    \left[r_{i,j} \log\left(\frac{\theta^+_j}{ \theta^-_j}\right) + (1-r_{i,j})\log\left(\frac{1-\theta^+_j}{1-\theta^-_j}\right) \right]\right),
\end{align*}
where 
 $\sigma(x) = \exp(x)/(1 + \exp(x))$ denotes the sigmoid function.
In approximate E step in the $t$-th iteration, after drawing $C$ Gibbs samples $\AA^{(t), 1}, \ldots, \AA^{(t), C}$ and $\QQ^{(t), 1}, \ldots, \QQ^{(t), C}$, we take a stochastic approximation of $\AA$ in the following manner,
\begin{align}\label{eq-algo-sa}
	\AA^{\text{ave},\,(t)} \leftarrow 
	\left(1-t^{-1}\right)\AA^{\text{ave},\,(t-1)}
	+ 
	t^{-1}
	\sum_{c=1}^C \AA^{(t), c}/C,
\end{align}
where $\AA^{\text{ave},(t-1)}$ denotes the $\AA$-matrix averaged from all the previous iterations up to iteration $t-1$. 
The update \eqref{eq-algo-sa} uses a similar idea to the stochastic approximation EM (SAEM) algorithm in \cite{delyon1999saem}.
For $\QQ$,  we define $\QQ^{(t)} = I(\sum_{c=1}^C \QQ^{(t), c}/C > 1/2)$; that is, the average $\QQ$ obtained from the $C$ Gibbs samples is rounded element-wisely to the nearest integer (0 or 1) to give $\QQ^{(t)}$.
Then in the M step, fixing the current $\AA^{\text{ave},\,(t)}$ and $\QQ^{(t)}$, we can update the item parameters $\ttt^+$ and $\ttt^-$ in closed forms under the two-parameter model.
We call such an algorithm EM with Alternating Direction Gibbs EM (ADG-EM) as each E step iteratively draws Gibbs samples of discrete latent structures $\AA$ and $\QQ$.
Preliminary simulations show that drawing $C<10$ Gibbs samples in each E step usually suffices for good performance.
The steps of ADG-EM are summarized in Algorithm \ref{algo-screen}. 
The Supplementary Material includes simulation studies assessing the convergence behavior of this algorithm.

{We make a remark on the stochastic approximation step of the proposed algorithm. As briefly mentioned before, computation can be challenging for data with large $N$, $J$, and $K$, because $\QQ$ and $\AA$ will be huge matrices with complex dependencies. But if we think of entries of $\QQ$ and $\AA$ individually in a Bayesian fashion, then each entry follows a Bernoulli distribution \textit{a posteriori} and indeed has an analytic posterior that is easy to sample from. With this thinking, our stochastic approximation procedure relies on a few Gibbs steps to achieve scalability. Such a procedure is specifically motivated by the multidimensional binary nature of the latent structures.}

Our Algorithm \ref{algo-screen} applies the stochastic approximation to updating $\AA$ but not to that of $\QQ$ in each iteration; that is, the update ``$\QQ= I(\QQ^{\text{sum}}/C>1/2)$ element-wisely'' in Algorithm \ref{algo-screen} does not depend on the iteration number $t$, in contrast to \eqref{eq-algo-sa}. 
We find through simulations that this algorithm has good estimation accuracy in various cases including when $N$ and $J$ are both very large. But one could similarly apply the stochastic approximation to both $\QQ$ and $\AA$; we present this modified version as Algorithm \ref{algo-adg-saem} in the Supplementary Material.

\begin{algorithm}[h!]
\caption{ADG-EM: Alternating Direction Gibbs EM for estimating $\QQ$ and $\AA$}\label{algo-screen}

\SetKwInOut{Input}{input}
\SetKwInOut{Output}{Output}

\KwData{Response matrix $\RR=(r_{i,j})_{N\times J}\in\{0,1\}^{N\times J}$ and number of attributes $K$.}
 Initialize $\AA=(a_{i,k})_{N\times K}\in\{0,1\}^{N\times K}$ and  $\QQ=(q_{j,k})_{J\times K}\in\{0,1\}^{J\times K}$.\\
 Initialize parameters $\ttt^+$ and $\ttt^-$. 
 Set $t=1$, ~$\AA^{\text{ave}}=\zero$.\\
 \While{not converged}{
 
   \lFor{$(i,j)\in[N]\times [J]$}{
      $
      \psi_{i,j} \leftarrow 
      r_{i,j}\log[\theta^+_j / \theta^-_j] + 
      (1-r_{i,j})\log[(1-\theta^+_j) / (1-\theta^-_j)]
      $
  }
  
  $\AA^{\text{sum}}\leftarrow\zero$, \quad $\QQ^{\text{sum}}\leftarrow\zero$.\\
\tcp{Approximate E Step:~draw $C$ Gibbs samples of entries of $\AA$ and $\QQ$}
\For{$r\in[C]$}{
    \For{$(i,k)\in[N]\times[K]$}{
    Draw $a_{i,k}\sim\text{Bernoulli}\Big(\sigma\Big(-\sum_{j} q_{j,k} \prod_{m\neq k} a_{i,m}^{q_{j,m}} \psi_{i,j} \Big)\Big)$
    }
     $\AA^{\text{sum}} \leftarrow \AA^{\text{sum}}+\AA$;    
    }
    
  $\AA^{\text{ave}} \leftarrow
   (1-t^{-1})\AA^{\text{ave}} +
   t^{-1}(\AA^{\text{sum}}/C);\quad
  t\leftarrow t+1.$
  
  \For{$r\in[C]$}{
    \For{$(j,k)\in[J]\times[K]$}{
    Draw $q_{j,k}\sim\text{Bernoulli}\Big(\sigma\Big(\sum_{i} (1-a_{i,k}) \prod_{m\neq k} a_{i,m}^{q_{j,m}} \psi_{i,j} \Big)\Big)$
    }
     $\QQ^{\text{sum}} \leftarrow \QQ^{\text{sum}}+\QQ$;
}
  
  $\QQ= I(\QQ^{\text{sum}}/C>1/2)$ element-wisely;\quad
  $\mathbf I^{\text{ave}} =\Big( \prod_k \{a_{i,k}^{\text{ave}}\}^{q_{j,k}} \Big)_{N\times J}$;
  \color{black}
  
  \tcp{M Step:~update model parameters}
   \For{$j\in [J]$}{
     $\theta_j^+ \leftarrow (\sum_{i}r_{i,j} I^{\text{ave}}_{i,j})/(\sum_{i}I_{i,j}^{\text{ave}})$,\qquad 
     $\theta_j^- \leftarrow (\sum_{i}r_{i,j} (1-I^{\text{ave}}_{i,j}))/(\sum_{i}(1-I_{i,j}^{\text{ave}}))$;
   }
 }

$\hat \AA = I(\AA^{\text{ave}} > 1/2)$ element-wisely.

\Output{$\hat\QQ$ and $\hat\AA$.}
\end{algorithm}

In terms of computational complexity, Algorithm \ref{algo-screen} has $O((N+J)K)$ complexity in each iterative step thanks to the approximation based on a small number ($C<10$) of Gibbs samples, in contrast to the $O((N+J)2^K)$ complexity of the regularized EM algorithms in \cite{chen2015statistical} and \cite{xu2018jasa} that evaluate the probabilities of all the $2^K$ configurations of the binary latent patterns. This reduction to linear complexity in $K$ greatly reduces the computational cost of estimating a SLAM for large-scale data.
To our knowledge, this is among the first estimation algorithms for SLAMs or cognitive diagnostic models that have linear complexity in $K$ and enjoy good estimation accuracy; see the simulation studies in Section \ref{sec-simu} for details of performance.

\subsection{Estimation under the Multi-Parameter Model}
\label{sec-est-mult}

The multi-parameter model in Example \ref{exp-mult} involves potentially many more parameters than the two-parameter model, since all the main effects and interaction effects of latent attributes possibly enter the likelihood.
This complicated form poses a greater challenge to computation, especially for large-scale and high-dimensional scenarios considered here.
Fortunately, the two-parameter DINA model is a submodel for multi-parameter SLAMs in an interesting way such that under a same $\QQ$-matrix, the main term in the former exactly captures the highest-order interaction term of the active attributes in the latter (see the discussion after Example \ref{exp-mult}). Therefore, when the key interest is in recovering the discrete latent structures in $\QQ$ and $\AA$, such a relationship inspires the following question: can one maximize the two-parameter likelihood to obtain any meaningful approximate estimator when data indeed come from a multi-parameter model?

On the practical side, the two-parameter DINA is indeed a very popular model employed by practitioners and likely the most widely used model in analyzing diagnostic assessment data in education \citep[e.g., see][]{chen2015statistical, culpepper2015, chen2018bayesian}, though the multi-parameter models are more general and flexible alternatives \citep{HensonTemplin09, dela2011}.
Such practices are mainly due to the computational simplicity and nice interpretability of the two-parameter model, yet the risk of over-simplification exists.
Motivated by the computational need and the scientific practice stated above, we next first study the property of the \textit{oversimplified joint MLE}, obtained from maximizing the two-parameter likelihood when the true data-generating model is instead multi-parameter. Later, we will show that such theoretical property inspires the development of a scalable two-step estimation procedure for multi-parameter models.

\paragraph{Property of the Oversimplified Joint MLE under the Two-parameter Likelihood.}
Next, we consider the situation when a multi-parameter SLAM is {oversimplified} to a two-parameter SLAM. 
We first provide conditions that guarantee a {oversimplified} joint MLE is consistent in estimating part of model structure. This provides a basis for subsequent second-stage estimation. 
We also establish that under certain stronger conditions, the {oversimplified} two-parameter joint MLE directly give consistent estimation of rows of $\QQ$ and $\AA$.
Together, these theoretical results will inspire the development of valid and efficient estimation methods for multi-parameter SLAMs in the later half of this section.

We introduce some notation. Under a multi-parameter SLAM, denote $P^{\true}_{i,j} = \MP(r_{i,j}=1\mid \text{true model})$. Given any $(\QQ, \AA)$, we define the two-parameter approximation by 
\begin{align}\label{eq-pijd2}
    P^{2\approx}_{i,j}(\QQ, \AA) = 
    \begin{cases}
		\dfrac{\sum_{m=1}^N I(\aa_m \succeq \qq_j) P_{m,j}^{\true} }{\sum_{m=1}^N I(\aa_m \succeq \qq_j)}, & \text{if }\aa_i \succeq \qq_j;
		\\[5mm]
		\dfrac{\sum_{m=1}^N I(\aa_m \nsucceq \qq_j) P_{m,j}^{\true} }{\sum_{m=1}^N I(\aa_m \nsucceq \qq_j)}, & \text{if }\aa_i \nsucceq \qq_j.
	\end{cases}
\end{align}
The $P^{2\approx}_{i,j}(\QQ, \AA)$ is determined by an arbitrary specification of the discrete latent structure $(\QQ, \AA)$ and also the true continuous parameters $P^\true_{i,j}$ (which further depends on $(\QQ^\true, \AA^\true)$ and $\bo\Theta^\true$).
As implied by the definition in \eqref{eq-pijd2}, the $P^{2\approx}_{i,j}(\QQ, \AA)$ is indeed a two-parameter approximation, because for each item $j$, the set of probabilities $\{P^{2\approx}_{i,j}(\QQ, \AA):\, 1\leq i\leq N\}$ only take two possible values, depending on whether or not $\aa_i \succeq \qq_j$.

We first provide conditions sufficient for consistency of \textit{part of} the discrete latent structures given by a {oversimplified} MLE. These conditions would imply a two-stage estimation procedure to be described in Section \ref{sec-est-mult}.
Denote $D(p \| q) = p\log(p/q) + (1-p)\log\{(1-p)/(1-q)\}$, the Kullback-Leibler divergence of a Bernoulli distribution with parameter $p$ from that with parameter $q$. Define the following function of $(\QQ,\AA)$,
\begin{align}\label{eq-fj}
f_j(\QQ, \AA)
=\sum_{i=1}^N D\Big(P^{\true}_{i,j} \Big\| P^{2\approx}_{i,j}(\QQ,\AA) \Big).	
\end{align}
To interpret, for item $j$ the $f_j(\QQ, \AA)$ characterizes the KL divergences from the true parameters $P^{\true}_{i,j}$ to the two-parameter approximation induced by the discrete structure $(\QQ, \AA)$.
We first consider the following assumption to replace the previous Assumption \ref{cond-gap-2p} on the true parameters under the multi-parameter model.
For two numbers $a$ and $b$, denote the maximum of them by $a\vee b$. 
Recall that the $\mathcal E_0$ defined in Assumption \ref{as-weak} is the set of variables that depend on some single latent attribute.
Define $\mathcal E_0 = \{j\in[J]:\,\qq_j=\ee_k~\text{for some}~k\in[K]\}$.
Under Assumption \ref{as-weak}, we have the following theorem.

\begin{assumption}\label{as-weak}
Define $\mathcal E_0 = \{j\in[J]:\,\qq_j=\ee_k~\text{for some}~k\in[K]\}$.
	The true data-generating multi-parameter SLAM satisfies
	\begin{align}\label{eq-less}
	\min_{j\in\mathcal E_0} \left(\theta^{\true}_{j,\one_K} - \theta^{\true}_{j,\zero_K}\right)^2 \geq \zeta_J; \qquad
	\sum_{j\notin\mathcal E_0} f_j(\ZZ^{\true})
	    =
	    \min_{\ZZ=(\QQ,\AA)}
		\sum_{j\notin\mathcal E_0} f_j(\ZZ)+
	    o(NJ\cdot\eta_J),
	\end{align}
	for some $\{\zeta_J\}\subseteq(0,1)$ and some bounded sequence $\{\eta_J\}\subseteq[0,\infty)$.
\end{assumption}
In the special case where $\eta_J=0$, Eq.~\eqref{eq-less} implies that the oracle two-parameter approximation is the best possible two-parameter approximation in the sense of minimizing the KL-divergence.
In general cases when $\{\eta_J\}\subseteq[0,\infty)$ is a bounded sequence, Eq.~\eqref{eq-less} weakens to requiring the oracle two-parameter approximation to be close to the best.
This \eqref{eq-less} in Assumption \ref{as-weak} imposes a quite mild requirement on the data-generating true parameters.

\begin{theorem}\label{thm-mis-weak}
	Suppose the data $\RR^{\true}$ come from a multi-parameter SLAM but the estimators $\hat\ZZ^{2\approx} = (\hat\QQ^{2\approx},\hat\AA^{2\approx})$ are obtained through maximizing the oversimplified two-parameter likelihood \eqref{eq-prob1}. 
    Suppose Assumptions \ref{assume-pij}, \ref{cond-id}, and \ref{as-weak} hold.
    With $\sqrt{J}=O\left(\sqrt{M}N^{1-c}\right)$ for a small $c>0$ and $\gamma_J$ defined in \eqref{eq-defgamma}, there is
	\begin{align*}
	\frac{1}{J}\sum_{j\in\mathcal E_0} I(\hat\qq^{2\approx}_j\neq\qq_j^{\true}) 
	= o_P\left(\frac{\eta_J \vee \gamma_J}{\zeta_J\cdot p_N}\right),
	\qquad
	\frac{1}{N}\sum_{i=1}^N I(\hat\aa^{2\approx}_i\neq \aa_i^{\true})
    = o_P\left(\frac{\eta_J \vee \gamma_J}{\zeta_J\cdot \delta_J}\right),
\end{align*}
up to a permutation of the $K$ attributes.
The joint MLE under an oversimplified two-parameter submodel is  consistent in recovering $\AA^{\true}$ and the single-attribute rows in $\QQ^{\true}$.
\end{theorem}

Theorem \ref{thm-mis-weak} has the following useful practical implication. After a first step of maximizing the {oversimplified} two-parameter likelihood to estimate $\AA^{\true}$ and the single-attribute rows in $\QQ^{\true}$, a ``regression'' type second step can be used to further estimate the remaining multi-attribute rows in $\QQ^{\true}$ based on the first stage estimator $\hat\AA$.
In Section \ref{sec-est-mult}, we provide a practical estimation procedure following this rationale. 

In practice, when the true parameters are more similar to the two-parameter submodel than Assumption \ref{as-weak}, the {oversimplified} joint MLE can even directly gives the consistency of all row vectors of $\QQ$ and $\AA$. 
The following assumption and theorem formalize this intuition.

\begin{assumption}[True Parameters More Similar to a Two-Parameter Model]\label{as-strong}
	As $N,~J\to\infty$, the true data-generating multi-parameter SLAM satisfies
	\begin{align}
	\label{eq-strong1}
        &\min_{j\in[J]} \min_{\aaa\succeq\qq_j^{\true}\atop \aaa'\nsucceq\qq_j^{\true}} \left(\theta^{\true}_{j,\aaa} - \theta^{\true}_{j,\aaa'}\right)^2 \geq \Delta_J; \qquad
		\sum_{j\notin\mathcal E_0} f_j(\QQ^\true, \AA^\true)
		= o(NJ\cdot\eta'_J),
	\end{align} 
	for some $\{\Delta_J\},\,\{\eta'_J\}\subseteq(0,\infty)$, where $f_j(\QQ^\true, \AA^\true)$ is as defined in \eqref{eq-fj}.
\end{assumption}

\begin{theorem}[True Parameters More Similar to a Two-Parameter Model]\label{thm-mis-strong}
	Suppose the data $\RR^{\true}$ come from a multi-parameter SLAM but the estimators $\hat\ZZ^{2\approx} = (\hat\QQ^{2\approx},\hat\AA^{2\approx})$ are obtained through maximizing the oversimplified two-parameter likelihood \eqref{eq-prob1}. 
	Under Assumptions \ref{assume-pij}, \ref{cond-id}, and \ref{as-strong}, as $N,\,J\to\infty$, with $\sqrt{J}=O\left(\sqrt{M}N^{1-c}\right)$ for a small $c>0$ and $\gamma_J$ defined in \eqref{eq-defgamma}, 
	\begin{align*}
	\frac{1}{J}\sum_{j=1}^J I( \hat\qq_j^{2\approx}\neq\qq_j^{\true}) 
	= o_P\left(\frac{\gamma_J\vee\eta'_J}{\Delta_J\cdot p_N}\right),
	\qquad
	\frac{1}{N}\sum_{i=1}^N I(\hat\aa_i^{2\approx}\neq \aa_i^{\true})
    = o_P\left(\frac{\gamma_J\vee\eta'_J}{\Delta_J\cdot \delta_J}\right),
\end{align*}
up to a permutation of the $K$ latent attributes.
In this case, the joint MLE under a oversimplified two-parameter submodel is consistent in recovering rows of $\QQ^{\true}$ and $\AA^{\true}$.
\end{theorem}

The implication of Theorem \ref{thm-mis-strong} is that when the true parameters are  similar enough to a two-parameter model, directly maximizing the oversimplified two-parameter likelihood suffices in recovering all the discrete latent structures $\QQ$ and $\AA$.
Thanks to Theorem \ref{thm-mis-weak} and Theorem \ref{thm-mis-strong}, the scalable estimation algorithm for the two-parameter model proposed in Section \ref{sec-est-2p} can serve as a useful approximation for computing joint MLE under a multi-parameter model.
In particular, the different scenarios characterized by  Assumption \ref{as-weak} (referred to as multi-parameter model with \textit{weaker two-parameter signal} from now on) and Assumption \ref{as-strong} (referred to as multi-parameter model with \textit{stronger two-parameter signal})  inspire two ways of performing estimation. 
Since the conditions in Theorem \ref{thm-mis-weak} are weaker than those in Theorem \ref{thm-mis-strong}, we next focus on the more general case of {weaker two-parameter signal} and present a two-stage estimation procedure.
We also provide the one-stage estimation results corresponding to the stronger-two-parameter-signal case in the Supplementary Material.

\paragraph{Two-stage Estimation for Multi-parameter SLAMs Corresponding to Theorem \ref{thm-mis-weak}.}
When the multi-parameter model satisfies Assumption \ref{as-weak}, Theorem \ref{thm-mis-weak} offers a useful insight that directly maximizing the {oversimplified} two-parameter likelihood can lead to consistent estimators of $\AA$ and those single-attribute row vectors in $\QQ$.
Such theoretical guarantee about $\hat\AA$ via a {oversimplified} MLE inspires the following two-stage estimation procedure. 
After using Algorithm \ref{algo-screen} to obtain $\hat\QQ$ and $\hat\AA$, we fix $\hat\AA$ as some surrogate ``covariates'' in order to re-estimate the matrix $\QQ$ through a second regression step.
Specifically, a multi-parameter SLAM in Example \ref{exp-mult} has the following reparametrization,
\begin{align}\label{eq-repar}
	\theta_{j,\hat\aa_i}
	= f\left(\sum_{S\subseteq \{1,\ldots,K\}} \mu_{j,S}\prod_{k\in S}\hat a_{i,k} \right),
	\quad \mu_{j,S}\neq 0~\text{only if}~
	\qq_{j, S} = \one,
\end{align} 
where $\mu_{j,S}$ is the coefficient for the interaction effect of the attributes in $S$, and $\qq_{j,S} = (q_{j,k}:\, k\in S)$.
Therefore, the sparsity structure of vector 
$\bo \mu_j=(\mu_{j,S};\,S\subseteq\{0,1\}^K)$ in the reparametrization \eqref{eq-repar} encode the information of $\qq_j$.
Now if $\aa_i$'s are treated as known instead of latent, 
for each $j$ we can use a penalized logistic regression to find the nonzero regression coefficients $\mu_{j,S}$'s. 
Then those nonzero $\mu_{j,S}$'s define the set $\mathcal K_j$ and hence determine the vector $\qq_j$. This is the basic rationale for our second \textit{regression stage}.

More specifically, in this second {regression stage}, for each item $j$, the parameter vector $\bo\mu_j=(\mu_{j,S};\,S\subseteq[K])$ involves all the possible interaction effects of the $K$ binary attributes. So $\bo\mu_j$ has dimension $2^K$, which can be huge given a moderate number of latent attributes. This is in analogy to the high-dimensional regression problem for a generalized linear model with link function $f^{-1}$.
When $2^K$ is huge, we recommend using the independence screening approach \citep{fan2008sure} to select candidate interactions of the attributes and then performing the variable selection only on the set of candidate interactions of attributes.
The \textit{all-effect marginal screening} method is as follows.
For an arbitrary subset $S\subseteq[K]$ of latent attributes, viewing the interaction term ${\prod_{k\in S}}\, \hat a_{i,k}$
as a ``feature'', we define its maximum marginal likelihood estimator $\hat{\bo\mu}_{j,S}^M=(\hat\mu_{j,S_0}^M, ~\hat\mu_{j,S}^M)$ based on the logistic regression as 
\begin{align}\notag
	\hat{\bo\mu}_{j,S}^M=
	 \argmin_{\bo\mu_{j,S}^M}
	 \frac{1}{N}\sum_{i=1}^N \left[r_{i,j}
	 \left(\mu_{j,S_0}^M + \mu_{j,S}^M {\prod_{k\in S}}\hat a_{i,k}\right) -
	 \log\left\{\exp\left(\mu_{j,S_0}^M + \mu_{j,S}^M {\prod_{k\in S}}\hat a_{i,k}\right)+1\right\}
	 \right];
\end{align}
here using logistic regression is appropriate for marginal screening because the responses $\{r_{1,j}, \ldots, r_{N, j}\}$ are binary.
Then we select the following set $\hat{\mathcal M}_{j,\,\text{Scre}}$ of candidate interactions,
$\hat{\mathcal M}_{j,\,\text{Scre}} = \left\{S\subseteq\{0,1\}^K:\, |\hat\mu_{j,S}^M|>\tau_{\text{thres}}\right\}$,	
where $\tau_{\text{thres}}>0$ is a prespecified threshold.
An even faster screening method is the \textit{main-effect marginal screening}, which  only screens the marginal main effects of the $K$ attributes for each item. That is, for $\tau'_{\text{thres}}>0$, define
\begin{align}\label{eq-kmain}
	&\hat{\mathcal K}^{\text{main}}_{j} 
	=
	 \left\{k\in\{1,\ldots,K\}:\, |\hat\mu_{j,\{k\}}^M|>\tau'_{\text{thres}}\right\},
	 \quad \hat{\mathcal M}_{j,\,\text{mScre}}
	=
	\text{power set of }\hat{\mathcal K}^{\text{main}}_{j}.
\end{align}
\begin{remark}
In practice, one can bypass the issue of the selection of the threshold $\tau'_{\text{thres}}$ in \eqref{eq-kmain} in the following way. That is, we can arrange the absolute values of the $K$ marginal main effects $|\hat\mu_{j,k}^M|$ from the largest to the smallest, denoted by $|\hat\mu_{j,(1)}^M|>|\hat\mu_{j,(2)}^M|>\cdots>|\hat\mu_{j,(K)}^M|$. From this ranking, we then select the first $k'$ attributes as the candidate ones for which the gap $|\hat\mu_{j,(k')}^M| - |\hat\mu_{j,(k'+1)}^M|$ is the largest.
In the simulation studies, we find that main-effect marginal screening coupled with this selection strategy usually suffices for good performance.
\end{remark}

Finally, for each item $j$, given the set of candidate terms $\hat{\mathcal M}_{j,\,\text{mScre}}$ (or $\hat{\mathcal M}_{j,\,\text{Scre}}$), we use a $L_1$-penalized logistic regression treating these candidate terms as predictors to arrive at a final set of selected terms $\hat{\mathcal M}_{j,\,\text{pen}}$, which is a subset of the power set of $[K]$. 
The tuning parameter of the $L_1$ penalty is chosen by 5-fold cross validation.
Based on this, the vector $\hat\qq_j$ can be determined. 
The following example illustrates the two-stage estimation procedure.

\begin{example}[Estimating a Multi-Parameter Model]\label{exp-weak}
We generate data with $(N,J,K)=(2400,1200,3)$ under the multi-parameter model with $f(\cdot)$ being the identity link (GDINA, \cite{dela2011}).
The true $\QQ$ has half of the row vectors loading on some single attribute, one fourth loading on two attributes, and the remaining one fourth loading on all three attributes; these are visualized in Figure \ref{fig-weak}.
Define $\mathcal K_j = \{k\in[K]: q_{j,k}=1\}$ to be the set of active attributes for variable $j$ and specify the $\mu$-parameters in \eqref{eq-alleff} as 
\begin{align*}
    \theta_{j,\zero_K} =&~ \mu_{j,\varnothing} = 0.2,~~
    \theta_{j,\one_K} = \sum_{S\subseteq \mt K_j}\mu_{j,S} = 0.8;~~
	\mu_{j,S} = \frac{\theta_{j,\one_K} - \theta_{j,\zero_K}}{2^{|\mt K_j|}-1} ~~\text{for any}~~  S\subseteq\mathcal K_j,~S\neq\varnothing,
\end{align*}
where $|\mathcal K_j|=1,2,3$. This setting of the item parameters are the same as the simulation settings in \cite{xu2018jasa}, that is, for each item all the main-effect and interaction-effect parameters are equal.
The results for $(N,J,K)=(2400,1200,3)$ are presented in the upper panel of Figure \ref{fig-weak}.
In this scenario, the first-stage $\hat\QQ^{1\text{\normalfont{st}}}$ differs from $\QQ^{\true}$ by 39 entries, out of the $J\times K=3600$ entries.
The first stage $\hat\AA$ exactly equals $\AA_0$.
Treating $\hat\AA$ as known and fixed in the second stage estimation leads to a second-stage estimator $\hat\QQ^{2\text{\normalfont{st}}}$ which exactly equals $\QQ^{\true}$.
In the bottom panel of Figure \ref{fig-weak}, we show the estimation results for a simulated dataset with $(N,J,K)=(3000,2000,10)$ and the findings are similar.

\begin{figure}[h!]
\centering
\includegraphics[width=0.7\linewidth]{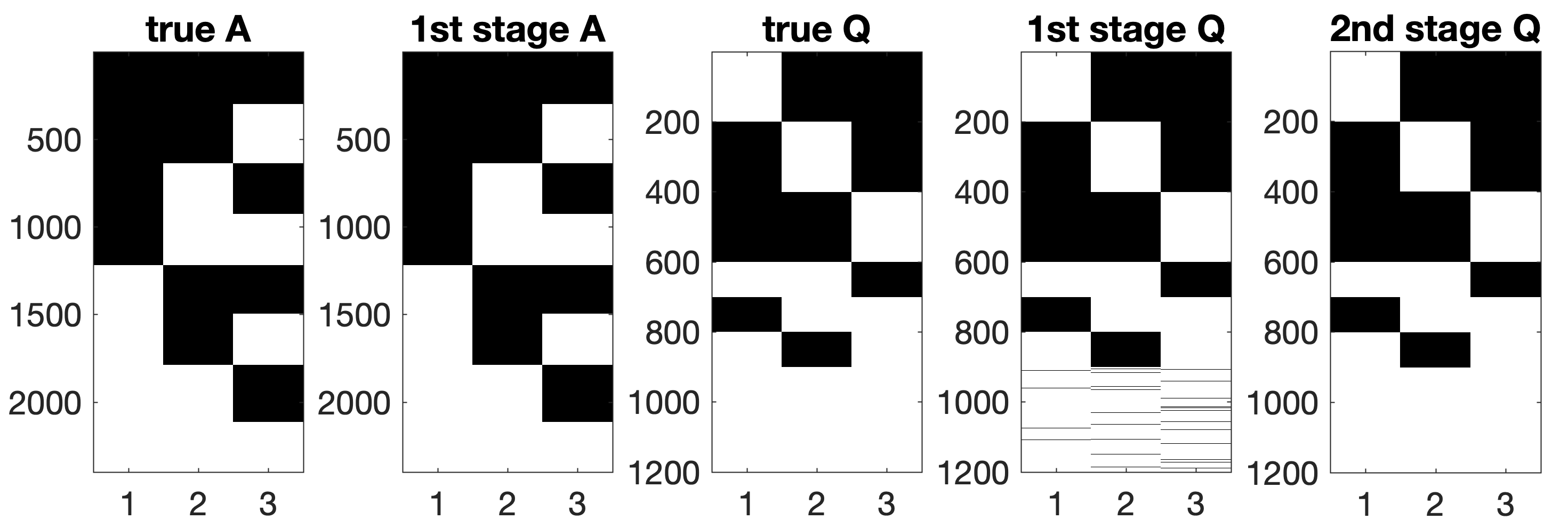}

\includegraphics[width=0.7\linewidth]{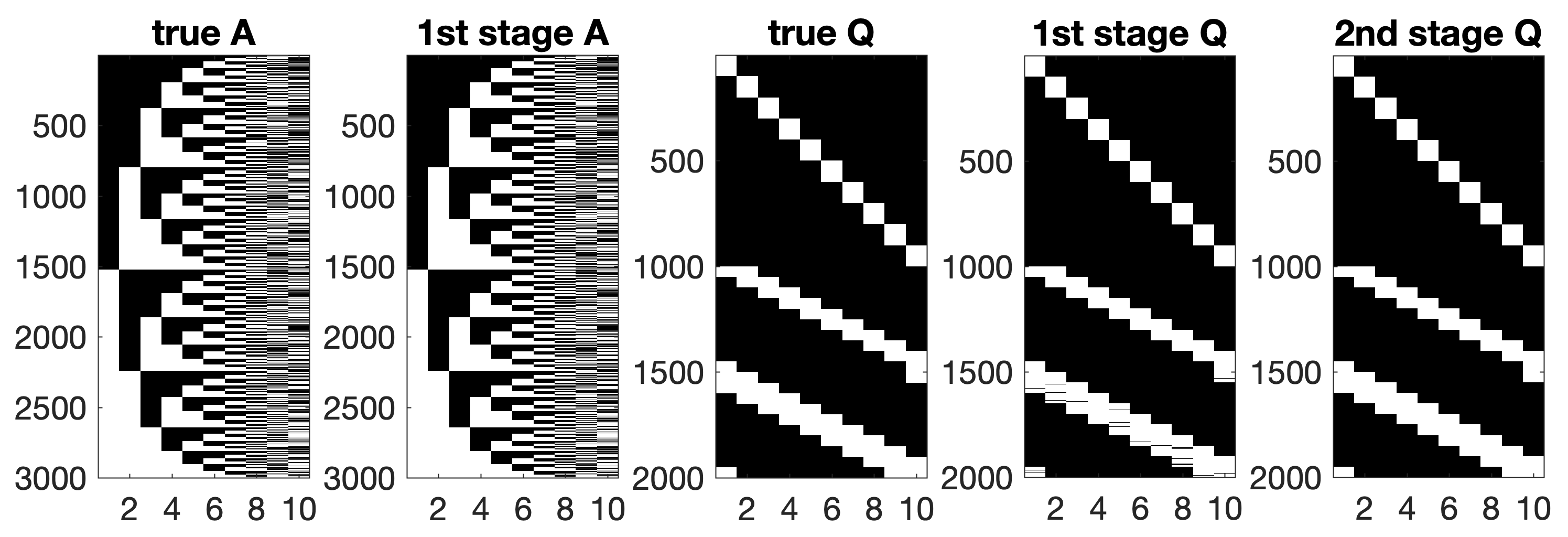}

\caption{
\textbf{Weaker} two-parameter signal (Assumption \ref{as-weak}) with $\mu_{j,\mathcal K_j}=(\theta_{j,\one_K} - \theta_{j,\zero_K})/(2^{|\mt K_j|}-1)$.
\textbf{Upper for $(N,J,K)=(2400,1200,3)$}: the first-stage $\hat\QQ^{1\text{st}}$ differs from $\QQ^{\true}$ by 39 entries, out of the $J\times K=3600$ entries.
\textbf{Bottom for $(N,J,K)=(3000,2000,10)$}:  the first-stage $\hat\QQ^{1\text{st}}$ differs from $\QQ^{\true}$ by 61 entries, out of the $J\times K=20000$ entries.}
 \label{fig-weak}
\end{figure}

\end{example}

\section{Simulation Studies}\label{sec-simu}

\paragraph{Simulations under the Two-Parameter Model.}
In this simulation study, we generate data under the two-parameter DINA model and examine Algorithm \ref{algo-screen}'s performance under  $(N,\,J)= (100,1000)$, $(1000,\,1000)$, or $(2000,\,2000)$, and $K=7, 10, 15$ ($2^7=128$, $2^{10}=1024$, $2^{15}=32768$).
In each simulation setting, the true $\QQ$ vertically stacks $J/(2K)$ copies of $I_K$, $J/(4K)$ copies of $\QQ^{(2)}_{\text{\normalfont{block}}}=(q^{(2)}_{j,k})$, and another $J/(4K)$ copies of  $\QQ^{(3)}_{\text{\normalfont{block}}}=(q^{(3)}_{j,k})$; here the entries of ``1''s are in the locations $q^{(2)}_{k,k}=1$ for $k\in[K]$, $q^{(2)}_{k,k+1}=1$ for $k\in\{1,\ldots,K-1\}$ and $q^{(2)}_{K,1}=1$; and $q^{(3)}_{k,k}=1$ for $k\in[K]$, $q^{(3)}_{k,k+1}=1$ for $k\in\{1,\ldots,K-1\}$ and $q^{(3)}_{K,1}=1$, and $q^{(3)}_{k,k+2}=1$ for $k\in\{1,\ldots,K-2\}$ and $q^{(3)}_{K-1,1}=q^{(3)}_{K,2}=1$. %

The true parameters are set to $1-\theta_j^+=\theta_j^-=0.2$ for each $j$. In each setting, 200 independent replications are carried out. The estimation accuracies are presented in Table \ref{tab-acc}. The column labeled as ``$\hat\AA=\AA^{\true}$'' records the number of replications out of 200 where the algorithm exactly recovers the entire matrix $\AA$; column ``$\hat\aa_i=\aa_i^{\true}$'' records the mean accuracy of recovering the $N$ row vectors of $\AA$ across the replications; column ``$\hat a_{i,k}=a_{i,k}^{\true}$'' records the mean accuracy of recovering the $NK$ individual entries of $\AA$.
The columns $\hat\QQ=\QQ^{\true}$, $\hat\qq_j=\qq_j^{\true}$, $\hat q_{j,k}=q_{j,k}^{\true}$ record similar measures for $\QQ$.
When $N=J=2000$, for all the considered $K$,  both matrix $\QQ$ and matrix $\AA$ are exactly recovered in each replication.

\begin{table}[h!]
  \centering
  \resizebox{\textwidth}{!}{%
  \begin{tabular}{lrrrccrcc}
    \toprule
    \multirow{2}{*}{\centering $2^K$}  & \multirow{2}{*}{\centering $J$} & \multirow{2}{*}{\centering $N$} & \multicolumn{3}{c}{$\hat\AA$} & \multicolumn{3}{c}{$\hat\QQ$} \\
    \cmidrule(lr){7-9}
    & & 
      & $\hat\AA=\AA^{\true}$ & $\hat\aa_i=\aa_i^{\true}$ & $\hat a_{i,k}=a_{i,k}^{\true}$
      & $\hat\QQ=\QQ^{\true}$ & $\hat\qq_j=\qq_j^{\true}$ & $\hat q_{j,k}=q_{j,k}^{\true}$ \\
    \midrule
 \multirow{3}{*}{\centering $2^{7}$} 
    &  $100$ &  $1000$
        &  0/200   & $0.910$ & $0.986$
        &  200/200 & $1.000$ & $1.000$
        \\
    &  $1000$ &  $1000$
        &  188/200 & $0.940$ & $0.970$
        &  188/200 & $0.941$ & $0.975$
        \\
    &  $2000$ &  $2000$
        &  200/200 & $1.000$ & $1.000$
        &  200/200 & $1.000$ & $1.000$
         \\
    \midrule
    \multirow{3}{*}{\centering $2^{10}$} 
    &  $100$ &  $1000$
        &  0/200   & $0.678$ & $0.949$
        &  185/200 & $0.965$ & $0.987$\\
    &  $1000$ &  $1000$
        &  189/200 & $0.955$ & $0.980$
        &  191/200 & $0.956$ & $0.983$
        \\
    &  $2000$ &  $2000$
        &  200/200 & $1.000$ & $1.000$
        &  200/200 & $1.000$ & $1.000$
         \\
    \midrule
    \multirow{3}{*}{\centering $2^{15}$} 
    &  $200$ &  $2000$
        &  0/200   & $0.709$ & $0.956$
        &  166/200 & $0.935$ & $0.980$\\
    &  $1000$ &  $1000$
        & 138/200 & $0.985$ & $0.993$
        & 194/200 & $0.985$ & $0.995$
        \\
    &  $2000$ &  $2000$
        & 200/200 & $1.000$ & $1.000$
        & 200/200 & $1.000$ & $1.000$
         \\
    \bottomrule
  \end{tabular}
 }
\caption{Two-parameter model (DINA) estimation results.}
\label{tab-acc}
\end{table}

\vspace{-5mm}
\paragraph{Simulations under the Multi-Parameter Model.}
We generate data under parameter settings similar to Example \ref{exp-weak} for various $N$, $J$, and $K$. For $K=7$ or $K=10$, we vary $(N,J)$  in $(100,1000)$, $(1000,1000)$ and $(2000,2000)$; for $K=15$, we vary $(N,J)$ in $(200,2000)$, $(1000,1000)$ and $(2000,2000)$.
In each of the considered scenarios, 200 simulation replications are carried out.
The estimation results are shown in Table \ref{tab-gdina-weak} and Table \ref{tab-gdina-strong}, respectively. 
The true parameter settings behind Table \ref{tab-gdina-weak} correspond to the weaker multi-parameter Assumption \ref{as-weak}, and they are the same as the simulation settings in \cite{xu2018jasa}.
Table \ref{tab-gdina-weak} shows that the first stage estimation yields very high accuracy of estimating rows in $\AA$ (perfect recovery in the considered scenarios), which provides a good basis for proceeding with the second stage of re-estimating rows in $\QQ$. Indeed, the second-stage estimator  $\hat\QQ^{(2)}$ based on the penalized regression approach introduced in Section \ref{sec-est-2p} shows desirable improvement over the first-stage estimator $\hat\QQ^{(1)}$. 
The true parameter settings behind Table \ref{tab-gdina-strong} correspond to the stronger multi-parameter Assumption \ref{as-strong}. As $(N,J)$ increase from $(1000,1000)$ to $(2000,2000)$, the {oversimplified} two-parameter MLE improves to almost perfect recovery of the discrete latent structures. This corroborates Theorem \ref{thm-mis-strong} that when the true parameters underlying a multi-parameter model are more similar to a two-parameter model, the {oversimplified} MLE from one-stage estimation can itself leads to consistency.

In practice, when fitting a SLAM to real data, if it is not clear whether a two-parameter model or a multi-parameter one is more suitable, we recommend performing two-stage estimation as described in Section \ref{sec-est-mult} to improve the estimation accuracy of the $\QQ$-matrix, as empirically shown in Table \ref{tab-gdina-weak}.
Then after performing the two-stage estimation procedure, one can apply some information criterion such as BIC to compare the first-stage estimator under the two-parameter model and second-stage estimator under the multi-parameter model in order to reach a final decision. 
In summary, our simulation studies show that across all the considered scenarios including the challenging case with $2^{15}=32768$, the proposed estimators have good accuracy of recovering the $\qq_j$'s and $\aa_i$'s.

\begin{table}[h!]
  \centering
  \resizebox{\textwidth}{!}{%
  \begin{tabular}{lcccccccc}
    \toprule
    \multirow{2}{*}{\centering $2^K$}  & \multirow{2}{*}{\centering $J$} & \multirow{2}{*}{\centering $N$} & \multicolumn{2}{c}{$\hat\AA$} & \multicolumn{2}{c}{first-stage $\hat\QQ^{(1)}$} 
        & \multicolumn{2}{c}{second-stage $\hat\QQ^{(2)}$} \\
    \cmidrule(lr){4-5}\cmidrule(lr){6-7}\cmidrule(lr){8-9}
    & &  & $\hat\aa_i=\aa_i^{\true}$ & $\hat a_{i,k}=a_{i,k}^{\true}$ 
         & $\hat\qq^{(1)}_j=\qq_j^{\true}$ & $\hat q^{(1)}_{j,k}=q_{j,k}^{\true}$
         & $\hat\qq^{(2)}_j=\qq_j^{\true}$ & $\hat q^{(2)}_{j,k}=q_{j,k}^{\true}$
       \\
    \midrule
    \multirow{3}{*}{\centering $2^{7}$} 
    &  $100$ &  $1000$
        & $0.858$ & $0.978$
        & $0.971$ & $0.996$ 
        & $0.991$ & $0.998$
        \\
    &  $1000$ &  $1000$
        & $1.000$ & $1.000$
        & $0.916$ & $0.988$ 
        & $0.989$ & $0.998$
        \\
    &  $2000$ &  $2000$
        & $1.000$ & $1.000$
        & $0.955$ & $0.994$  
        & $1.000$ & $1.000$
         \\
    \midrule
    \multirow{3}{*}{\centering $2^{10}$} 
        &  $100$ &  $1000$
        &  $0.560$ & $0.941$ 
        &  $0.943$ & $0.993$  
        &  $0.971$ & $0.995$
        \\
    &  $1000$ &  $1000$
        & $1.000$ &  $1.000$
        & $0.910$ & $0.991$  
        & $0.986$ & $0.998$
         \\
    &  $2000$ &  $2000$
        & $1.000$ & $1.000$
        & $0.950$ & $0.995$  
        & $1.000$ & $1.000$
         \\
    \midrule
    \multirow{3}{*}{\centering $2^{15}$} 
    &  $200$ &  $2000$
        & $0.546$ & $0.942$
        & $0.960$ & $0.996$ 
        & $1.000$ & $1.000$
        \\
    &  $1000$ &  $1000$
        & $1.000$ &  $1.000$
        & $0.904$ & $0.994$  
        & $0.983$ & $0.999$
         \\
    &  $2000$ &  $2000$
        & $1.000$ &  $1.000$
        & $0.941$ & $0.996$  
        & $1.000$ & $1.000$
         \\
    \bottomrule
  \end{tabular}
 }
 \caption{Two-stage estimation for multi-parameter model (GDINA) under the \textit{weaker} two-parameter signal with $\mu_{j,\mathcal K_j}=\left(\theta_{j,\one_K}-\theta_{j,\zero_K}\right)/(2^{|\mathcal K_j|}-1)$ where $|\mathcal K_j|=1,\,2,\,3$. 
  }
  \label{tab-gdina-weak}
\end{table}

\begin{table}[h!]
  \centering
  \resizebox{\textwidth}{!}{%
  \begin{tabular}{lrrrccrcc}
    \toprule
    \multirow{2}{*}{\centering $2^K$}  & \multirow{2}{*}{\centering $J$} & \multirow{2}{*}{\centering $N$} & \multicolumn{3}{c}{$\hat\AA$} & \multicolumn{3}{c}{$\hat\QQ$} \\
    \cmidrule(lr){4-6}\cmidrule(lr){7-9}
    & &
    & $\hat\AA=\AA^{\true}$ & $\hat\aa_i=\aa_i^{\true}$ & $\hat a_{i,k}=a_{i,k}^{\true}$ 
    & $\hat\QQ=\QQ^{\true}$ & $\hat\qq_j=\qq_j^{\true}$ & $\hat q_{j,k}=q_{j,k}^{\true}$
       \\
    \midrule
    \multirow{3}{*}{\centering $2^{7}$} 
    &  $100$ &  $1000$
        &  0/200   & $0.888$ & $0.983$
        &  192/200 & $1.000$ & $1.000$\\
    &  $1000$ &  $1000$
        &    200/200 & $1.000$ & $1.000$
        &    115/200 & $0.999$ & $1.000$
                \\
    &  $2000$ &  $2000$
        &  200/200 & $1.000$ & $1.000$
        &  173/200  & $1.000$ & $1.000$
         \\
    \midrule
    \multirow{3}{*}{\centering $2^{10}$} 
    &  $100$ &  $1000$
        &  0/200   & $0.638$ & $0.952$
        &  160/200 & $0.982$ & $0.997$\\
    &  $1000$ &  $1000$
        &  199/200 & $1.000$ & $1.000$
        &  114/200 & $1.000$ & $1.000$
         \\
    &  $2000$ &  $2000$
        & 200/200 & $1.000$ & $1.000$
        & 173/200  & $1.000$ & $1.000$
         \\
    \midrule
    \multirow{3}{*}{\centering $2^{15}$} 
    &  $200$ &  $2000$
        &  0/200  & $0.689$ & $0.968$
        &  171/200  & $0.973$ & $0.994$
        \\
    &  $1000$ &  $1000$
        &  148/200 & $1.000$ & $1.000$
        &  40/200  & $0.998$ & $1.000$\\
    &  $2000$ &  $2000$
        & 200/200  & $1.000$ & $1.000$
        & 170/200  & $1.000$ & $1.000$
         \\
    \bottomrule
  \end{tabular}
  }
 \caption{One-stage estimation for multi-parameter model (GDINA) under the \textit{stronger} two-parameter signal with $\mu_{j,\mathcal K_j}=\left(\theta_{j,\one_K}-\theta_{j,\zero_K}\right)/2$.
  }
  \label{tab-gdina-strong}
\end{table}

\section{Real Data Analysis}\label{sec-real}
We apply the proposed estimation method to real data from an educational assessment, the Trends in International Mathematics and Science Study (TIMSS). This dataset is a subset of the TIMSS 2011 Austrian data for analyzing students' abilities in mathematical sub-competences and is available in the \verb|R| package \verb|CDM|. 
It includes responses of $N=1010$ Austrian fourth grade students and $J=47$ items. Nine ($K=9$) attributes were specified in \cite{george2015cdm}:
(DA) \textit{Data and Applying}, 
(DK) \textit{Data and Knowing}, 
(DR) \textit{Data and Reasoning},
(GA) \textit{Geometry and Applying},
(GK) \textit{Geometry and Knowing}, 
(GR) \textit{Geometry and Reasoning},
(NA) \textit{Numbers and Applying}, 
(NK) \textit{Numbers and Knowing},
(NR) \textit{Numbers and Reasoning};
a provisional $\QQ$-matrix $\QQ^{\text{orig}}$ of size $47\times 9$ was also provided.

One structure specific to such large scale assessments is that only a subset of all items in the entire study is presented to each of students \citep{george2015cdm}. This results in many missing values in the $N\times J$ data matrix, and the considered dataset has a missing rate $51.73\%$. The joint MLE approach can be easily extended to handle the missing data under the ignorable missingness assumption.
Under such an assumption, it indeed suffices to replace the log-likelihood function over the $\{r_{i,j}:\,(i,j)\in[N]\times[J]\}$ by that over $\{r_{i,j}:\,(i,j)\in\Omega\}$, where $\Omega\subseteq[N]\times[J]$ is the set of indices in $\RR$ corresponding to those observed entries.  
In particular, the original log-likelihood function \eqref{eq-ll-2p} under the two-parameter model should be replaced by the following objective function,
\begin{align}
\notag
\ell^{\Omega,\,\text{two}}(\QQ,\,\AA,\,\TT\mid\RR)
=& \sum_{(i,j)\in\Omega}\Big[ r_{i,j}\Big(  \prod_{k}a_{i,k}^{q_{j,k}} \log \theta^+_j + (1-\prod_{k}a_{i,k}^{q_{j,k}}) \log \theta^-_j \Big) \Big]  \\ \notag
&  + (1-r_{i,j}) \Big(  \prod_{k}a_{i,k}^{q_{j,k}} \log(1- \theta^+_j) + (1-\prod_{k}a_{i,k}^{q_{j,k}}) \log(1- \theta^-_j) \Big)\Big].
\end{align}
With missing values in $\RR$, the previous ADG-EM Algorithm \ref{algo-screen} can be replaced by Algorithm \ref{algo-screen-miss} 
presented in the Supplementary Material.

The original $\QQ$-matrix provided in the TIMSS dataset $\QQ^{\text{orig}}$ has each item measuring only one attribute. This $\QQ^{\text{orig}}$ gives the interpretation of the attributes and encodes the domain knowledge about the test items. Therefore, we use $\QQ^{\text{orig}}$ to initialize the proposed algorithm. 
Moreover, we fix $J_{\text{anchor}}=K=9$ ``anchor'' items' row vectors in $\QQ^{\text{orig}}$ along the iterations of the algorithm. The anchor items are chosen such that  their corresponding row vectors form an identity submatrix $I_K$ of the $\QQ$-matrix. By this we hope to fix the interpretation of the $K$ columns as the $K$ provided attributes.
The two-parameter DINA model is often used to model and analyze data from educational assessments.
In the data analysis, we first perform estimation under the two-parameter DINA model and then also proceed with the second-stage estimation as described in Section \ref{sec-est-mult} to estimate $\QQ$-matrix under a multi-parameter GDINA model. But the two-parameter model gives a smaller BIC value and indicates a better fit. So next we only discuss the results given by the two-parameter model fitting.

In the resulting estimator $\QQ^{\text{est}}$, there are ten rows that have more than one nonzero entries, which are presented in Table \ref{tab-timss} together with their item number and item label.
First, for each of these ten items, the estimated $\qq$-vector always measures the attribute originally specified in $\QQ^{\text{orig}}$ (the dark orange entry of ``$\textcolor{orange!80!black}{\one}$'' in each of the ten rows in Table \ref{tab-timss}). 
This implies that the meaning of the original attributes are preserved in our estimation.
In addition, Table \ref{tab-timss} reveals extra information that some items depend on certain additional attributes besides the originally specified one (the dark blue entries of ``$\darkblue\one$'' in Table \ref{tab-timss}). 
For example, items M031379, M031380, M051001 originally are designed to measure attribute (NR) \textit{Number and Reasoning}, but the estimated $\QQ^{\text{est}}$ implies they also depend on the attribute (NA) \textit{Number and Applying}. 
In particular, the third item M051001 ``Soccer tournament'' asks:  
\textit{in a soccer tournament, teams get:
3 points for a win,
1 point for a tie,
0 points for a loss.
Zedland has 11 points.
What is the smallest number of games Zedland could have played?}
This is a difficult question for fourth graders and targets complicated skills in the content domain ``Number''; its difficulty is reflected in our estimation result that this item's estimated $\qq$-vector depends on all of the three attributes about ``Number'': (NA), (NK), and (NR). 
Table \ref{tab-timss} shows that items generally seem to have some clustered dependence on attributes falling in the same cognitive domain or the same content domain: attributes (NA), (NK), (NR) in the content domain ``Number'' are often measured together, and attributes (GA) and (NA) in the cognitive domain ``{Applying}'' are often measured together.

\begin{table}[h!]
  \centering
  \resizebox{\textwidth}{!}{%
  \begin{tabular}{llccccccccc}
  \toprule
  \multirow{2}{*}{\centering Item No.} &
  \multirow{2}{*}{\centering Item Label} & 
  \multicolumn{9}{c}{Attributes}\\
   \cmidrule(lr){3-11}
 & &
DA &
DK & 
DR &
GA &
GK & 
GR &
NA & 
NK &
NR\\
  \midrule
  M031379 & Trading sports cards & 0 & 0 & 0 & 0 & 0 & 0 & $\darkblue\one$ & 0 & $\textcolor{orange!80!black}{\one}$\\
  M031380 & Trading cartoon cards & 0 & 0 & 0 & 0 & 0 & 0 & $\darkblue\one$ & 0 & $\textcolor{orange!80!black}{\one}$\\
  M051001 & Soccer tournament & 0 & $\darkblue\one$ & 0 & 0 & 0 & 0 & $\darkblue\one$ & $\darkblue\one$ & $\textcolor{orange!80!black}{\one}$\\
  M051015 &  Complete Jay's shape & 0 & 0 & 0 & $\textcolor{orange!80!black}{\one}$ & 0 & 0 & $\darkblue\one$ & 0 & 0\\
  M051123 & Lines of symmetry complex figure &  0  &   0  &   0  &   $\darkblue\one$  &   $\textcolor{orange!80!black}{\one}$  &   0  &   0 &    $\darkblue\one$   &  0\\
  M041098 & How many cans must Sean buy & 0   &  0  &   0    &   $\darkblue\one$  &   0    & 0   &  $\textcolor{orange!80!black}{\one}$   &  0   & 0\\
  M041104 & Number between 5 and 6 & 0  &   0 &    0  &   $\darkblue\one$  &   $\darkblue\one$  &   0  &   0  &   $\textcolor{orange!80!black}{\one}$   &  0\\ 
  M041299 & Fraction of the cake eaten & 0  &   0  &   0  &   $\darkblue\one$   &  0 &    0  &  $\darkblue\one$  &     $\textcolor{orange!80!black}{\one}$   &  0\\ 
  M041143 & Identify shapes in the picture & 0  &   0    & 0 &    $\textcolor{orange!80!black}{\one}$  &   $\darkblue\one$   &  0    & 0 &    0  &   0\\
  M051006 & Cost of ice cream &  $\darkblue\one$  &  0    &  0  &    $\darkblue\one$  &   0   &  0  &   $\textcolor{orange!80!black}{\one}$  &   0   &  0\\ 
  \bottomrule
 \end{tabular}
 } 
 \caption{The ten multi-attribute rows in the estimated $\QQ^{\text{est}}$ for the TIMSS 2011 Austrian data. The entries ``$\darkblue\one$''s in dark blue are those that are estimated to be 1 but originally are 0 in $\QQ^{\text{orig}}$; the entries ``$\textcolor{orange!80!black}{\one}$''s in dark orange are those that are originally 1 in $\QQ^{\text{orig}}$.}
 \label{tab-timss}
\end{table}

\begin{figure}[h!]
	\centering
	\includegraphics[width=0.55\textwidth]{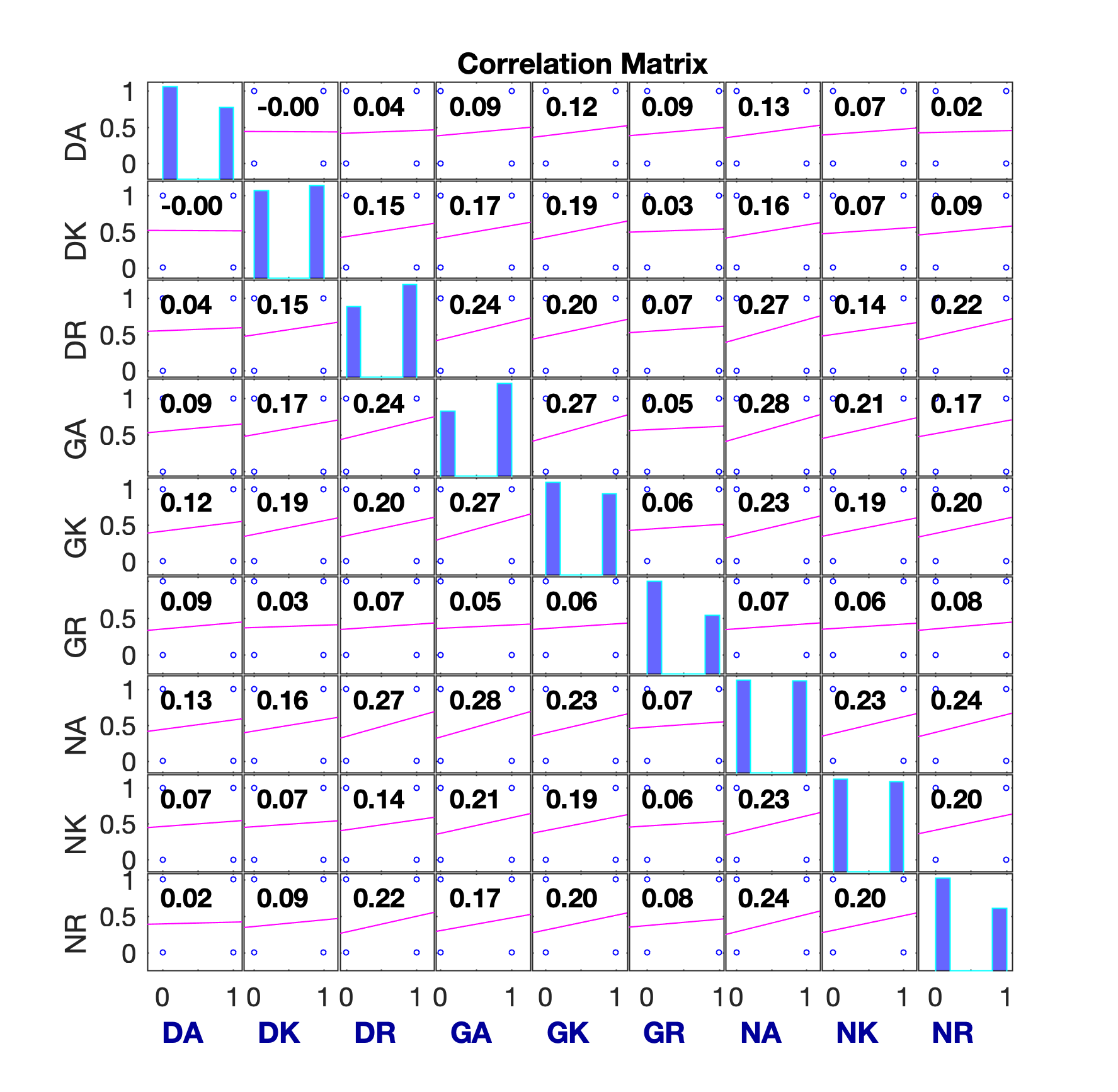}
	\caption{Correlation plots of the $K=9$ latent attributes based on the estimated $\AA^{\text{est}}$.}
	\label{fig-corr}
\end{figure}

We also examine the estimated $N\times K$ matrix $\AA^{\text{est}}$ for the $N=1010$ students. 
Based on $\AA^{\text{est}}$, marginally,  students master skills regarding \textit{Data} (average mastery 50.63$\%$) better than \textit{Geometry} (average mastery 48.12$\%$)  and \textit{Number} (average mastery  46.50$\%$); and they master skills regarding \textit{Applying} (average mastery  50.83$\%$) and \textit{Knowing} (average mastery  49.17$\%$) better than \textit{Reasoning} (average mastery  45.25$\%$).
The ``average mastery'' above is calculated as follows: for \textit{Data}, the average mastery is taken to be the average of the three columns of $\AA^{\text{est}}$ corresponding to DA, DK, DR.
Figure \ref{fig-corr} further shows the pairwise correlations between the nine attributes based on $\AA^{\text{est}}$. It can be seen that the attributes falling in the same content domain \textit{Number} do show relatively high correlations, where the three pairwise correlations between NA, NK, NR are 0.23, 0.24, 0.20. The correlation between GA and NA is 0.28, also  high. This aligns with our earlier observation that the estimated row vectors in $\QQ^{\text{est}}$ also tend to measure these attributes together.

\section{Discussion}\label{sec-disc}
This paper investigates the joint MLE approach to large-scale structured latent attribute analysis from both the theoretical and methodological perspectives. 
We provide theoretical guarantees for the estimability and consistency of the latent structures in the regime where all of the number of individuals, the number of observed variables, and the number of latent attributes can grow large. 
The obtained estimation error bounds not only guarantee asymptotic consistency of estimating both the variable loading vectors and subject latent profiles, but also offer insights into their estimation accuracies with finite samples. 
These consistency results also give practical implications for designing the $\QQ$-matrix in cognitive diagnostic applications.
For computation, we develop a scalable approximate algorithm to find the joint MLE of two-parameter SLAMs and also propose an effective two-stage estimation procedure for multi-parameter SLAMs.
Simulation studies and real data analysis demonstrate the usefulness of the proposed estimation approaches.

The developments in this work also open up several possibilities for future research. On the methodological side, based on the established results on consistency and finite sample error bounds for estimating latent structures, an interesting future task is performing statistical inference on SLAMs with a large number of test items and high-dimensional latent attributes. 
On the computational front, it would be interesting to relate or generalize the idea of the proposed estimation algorithm to other discrete optimization problems; it is also desirable to investigate the algorithm's theoretical properties in the future. 

This paper focuses on the discrete latent attribute modeling framework, and we include a particular study of misspecifying a multi-parameter SLAM to a two-parameter submodel motivated by computational needs and scientific practices. Besides such possible oversimplification, there could be other types of misspecifications, such as potentially misspecifying the continuous latent variables to be discrete.
As for this, we point out that this work does not intend to replace continuous latent factor modeling with the discrete counterpart, but rather to complement the former in suitable applications. 
In the future, it would be interesting to study consequences of the potential misspecification of continuous latent variables to discrete ones to elucidate their differences and connections.

On a final note, many specific models belonging to the SLAM family were initially proposed in the literature of cognitive diagnostic modeling. But this modeling framework's unique advantages of capturing fine-grained latent information and providing model-based clustering allow for applications far beyond this discipline. 
For example, similar modeling approaches have recently been employed in psychiatric evaluation \citep{dela2018}, disease epidemiology diagnosis \citep{o2019causes}, electronic health records \citep{ni2020double}, and precision medicine \citep{chen2020itr}. Just like the continuous latent factor analysis is nowadays widely used \citep{fan2021robust, chenlizhang2020, bing2020adaptive} beyond its initial application in psychometrics,
we believe the multidimensional discrete latent trait modeling also has great future promise in broader fields and warrants further statistical developments. By introducing and analyzing a principled joint MLE approach here, we hope this work contributes a step towards that farreaching goal.

\bigskip\noindent
\textbf{Supplementary Material.} 
The Supplementary Material contains all the technical proofs of the theoretical results and also includes additional discussion on computation.

\bigskip\noindent
\textbf{Acknowledgements.}  This research was supported by NSF CAREER SES-1846747, DMS-1712717, SES-1659328, and also by NIH NIEHS R01ES027498, R01ES028804. This research has also received funding from the European Research Council under the European Union's Horizon 2020 research and innovation program (grant agreement No 856506). The authors thank the Editor, Associate Editor, and reviewer for helpful and constructive comments.

\singlespacing
\bibliographystyle{apalike}
\bibliography{ref_jmle}

\spacingset{1.45} 

\newpage
\setcounter{section}{0}
\renewcommand{\thealgocf}{S.\arabic{algocf}}
\renewcommand{\thesection}{S.\arabic{section}}  
\renewcommand{\theequation}{S.\arabic{equation}}

\begin{center}
{\Large\textbf{Supplement to ``A Joint MLE Approach to Large-Scale Structured Latent Attribute Analysis''}}
\end{center}

\bigskip
\bigskip
\bigskip

This Supplementary Material is organized as follows.
Section \ref{sec-add} presents some additional simulation examples.
Section \ref{sec-pfs} gives the proofs of Theorems \ref{thm-joint-both}--\ref{thm-mis-strong} presented in the main text together with the proofs of several technical lemmas. 
Section \ref{sec-addalgo} presents additional algorithms, including one for estimation in the missing data scenario.

\section{Additional Simulation Studies}\label{sec-add}
In Subsection \ref{subsec-conv}, we present two specific simulation examples illustrating the convergence behavior and accuracy of the proposed ADG-EM algorithm, Algorithm \ref{algo-screen}. Then  Subsection \ref{subsec-conv2} further examines the algorithm's convergence behavior through a replicated simulation study.
In Subsection \ref{subsec-1stage}, we provide a simulation example illustrating the performance of the one-stage approximate estimator for multi-parameter SLAMs.

\subsection{Convergence behavior of the ADG-EM algorithm}\label{subsec-conv}

\normalfont{
We present two specific examples to visualize the intermediate results of Algorithm \ref{algo-screen}.
These two examples are both in the setting $(N,J,K)=(1000,1000,7)$ with parameters  $30\%=\theta_j^-=1-\theta_j^+$ for all $j\in[J]$, and all the latent attribute patterns are evenly distributed.
The data-generating $\QQ$-matrix vertically stacks $J/(2K)$ copies of submatrix $I_K$ and an appropriate number of another $K\times K$ submatrix $\QQ^2_{\text{\normalfont{block}}}=(q^{(2)}_{j,k})$, where $q^{(2)}_{k,k}=1$ for $k\in[K]$, $q^{(2)}_{k,k+1}=1$ for $k\in\{1,\ldots,K-1\}$ and $q^{(2)}_{K,1}=1$.

In the first example, we use ``randomly perturbed initialization'' for $(\QQ,\AA)$. 
Figure \ref{fig-bifac-perturb} shows the results of Algorithm \ref{algo-screen} together with its intermediate results along the first 4 iterations of the stochastic EM steps. The 6 plots in the first row of Figure \ref{fig-bifac-perturb} show the reconstruction of the data matrix $\RR$, and the 6 plots in the second row of Figure \ref{fig-bifac-perturb} show the estimation of  $\QQ$.
 Specifically, after the $t$-th iteration, based on the $\hat \QQ^{\text{\normalfont{iter.} }t}$, the $\hat \RR^{\text{\normalfont{iter.} }t}$ is reconstructed with the $(i,j)$th entry defined as
\begin{equation}
	\label{eq-recon}
	I\left(\hat\theta_j^+\cdot\xi_{\hat\qq_j,\hat{\bo a}_i} 
	+ \hat\theta_j^- \cdot(1-\xi_{\hat\qq_j,\hat{\bo a}_i})
	> \frac{1}{2} \right),\quad (i,j)\in[N]\times [J].
\end{equation}
which is the integer (0 or 1) nearest to the posterior mean of $(i,j)$th entry of $\RR$.
The ground truth for $\RR$ is just the $N\times J$ ideal response matrix in the noiseless case $\RR^{\ideal} = (r^{\ideal}_{i,j})$, where $r^{\ideal}_{i,j} = \xi(\qq_j,\bo a_i) = \prod_{k=1}^K a_{i,k}^{q_{j,k}}$.
Along the first 3 stochastic EM iterations, the matrix $\QQ$ change 2246, 275, 11 entries, respectively. Then  from the 4th iteration until the stopping criterion is reached, we observe that all the entries of $\QQ$ remain the same during the sampling in the E step. In the last several iterations the item parameters $(\ttt^+,\ttt^-)$ continued to change slightly and converge. 
Let $(r^{\text{\normalfont{observe}}}_{i,j})$ and $(r^{\text{\normalfont{recons}}}_{i,j})$ denote the observed noisy data matrix and the reconstructed data matrix in the end of the algorithm, respectively. Corresponding to the trial in Figure \ref{fig-bifac-perturb}, there is
 \begin{align*}
 	\frac{1}{NJ}\sum_{(i,j)\in[N]\times[J]} 
 	I(r^{\ideal}_{i,j}\neq r^{\text{\normalfont{observe}}}_{i,j}) = 0.2995,\quad
 	\frac{1}{NJ}\sum_{(i,j)\in[N]\times[J]} 
 	I(r^{\ideal}_{i,j}\neq r^{\text{\normalfont{recons}}}_{i,j}) = 5.21\times 10^{-5}.
 \end{align*}
 In the above display, the $0.2995$ reflects the noise rate in the observed data matrix corresponding to $1-\theta_j^+=\theta_j^-=0.3$ for each $j\in[J]$; and the $5.21\times 10^{-5}$ represents the error rate of reconstructing the $N\times J$ ideal response matrix, which is far smaller than the initial noise rate by several magnitudes. Indeed, there is no discernible difference between $\RR_{\text{iter. }4}$ and $\RR^{\ideal}$ based on the two rightmost plots in the first row of Figure \ref{fig-bifac-perturb}.
}

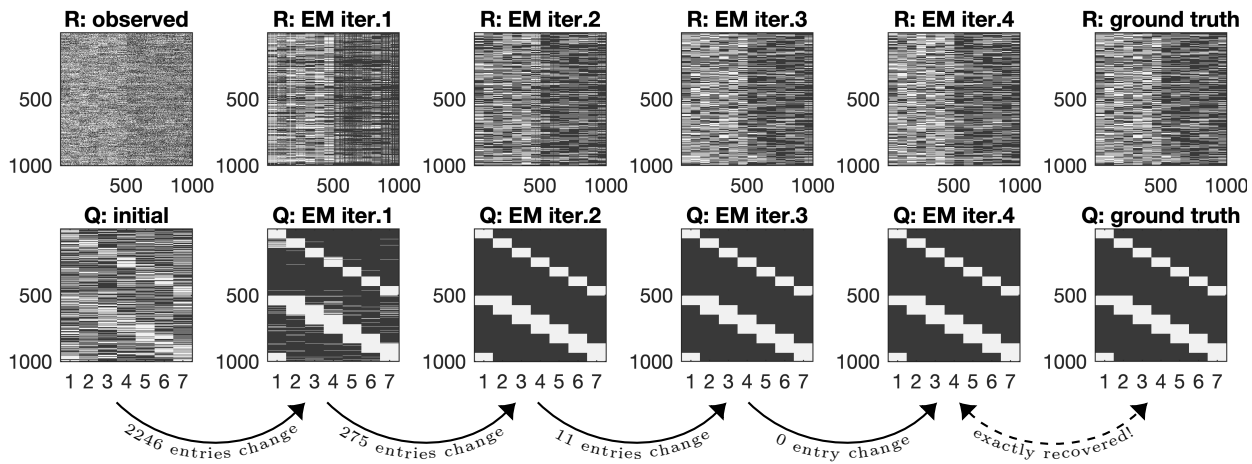
\begin{figure}[h!]
	\centering

    \begin{minipage}[b]{\textwidth}
	\centering
	    \includegraphics[width=\textwidth]{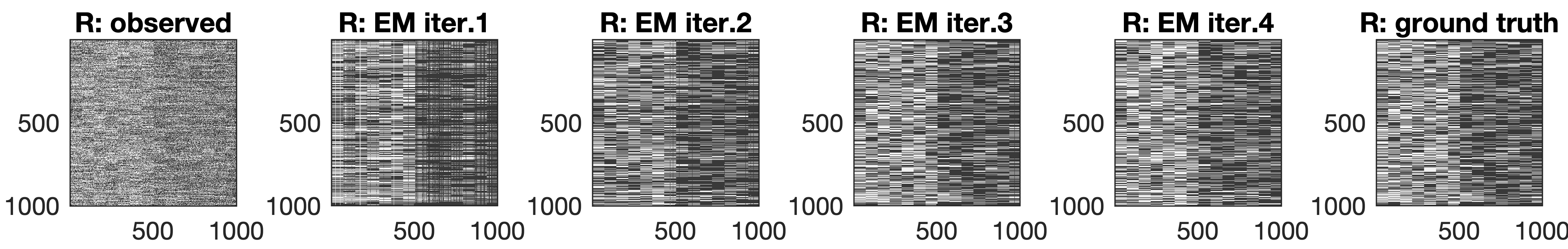}
	\end{minipage}

	\begin{minipage}[b]{1\textwidth}
	    \centering
	    \includegraphics[width=\textwidth]{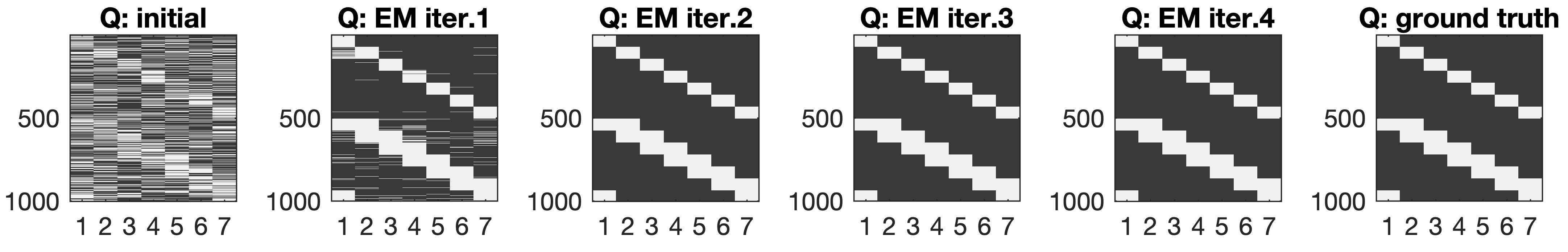}
	\end{minipage}
	
	\vspace{-3mm}
	\begin{tikzpicture}
	\centering
	    \node (hold1) at (0, 0) {};

		\node (Qini) at (1.2, 0) {};
		
		\node (Q1) at (4, 0) {};
		
		\node (Q2) at (6.8, 0) {};

        \node (Q3) at (9.6, 0) {};
        
        \node (Q4) at (12.4, 0) {};
        
        \node (Qtrue) at (15.2, 0) {};
        
        \node (hold2) at (18, 0) {};

        \def\myshift#1{\raisebox{-2.5ex}}
		\draw [->,thick,postaction={decorate,decoration={text along path,text align=center,text={|\tiny\myshift|2246 entries change}}}] (Qini) to [bend right=45]  (Q1);
				
		\draw [->,thick,postaction={decorate,decoration={text along path,text align=center,text={|\tiny\myshift|275 entries change}}}] (Q1) to [bend right=45]  (Q2);
		
		\draw [->,thick,postaction={decorate,decoration={text along path,text align=center,text={|\tiny\myshift|11 entries change}}}] (Q2) to [bend right=45]  (Q3);
		
		\draw [->,thick,postaction={decorate,decoration={text along path,text align=center,text={|\tiny\myshift|0 entry change}}}] (Q3) to [bend right=45]  (Q4);
		
		\draw [<->,dashed,thick,postaction={decorate,decoration={text along path,text align=center,text={|\tiny\myshift|exactly recovered!}}}] (Q4) to [bend right=45]  (Qtrue);
	\end{tikzpicture}
\caption{Estimation with \textit{randomly perturbed initialization}. Color {white} represents value ``1'' and color {black} represents value ``0''. Only 3 stochastic EM iterations suffice for perfect estimation of the structural matrix $\QQ$.
}
\label{fig-bifac-perturb}
\end{figure}

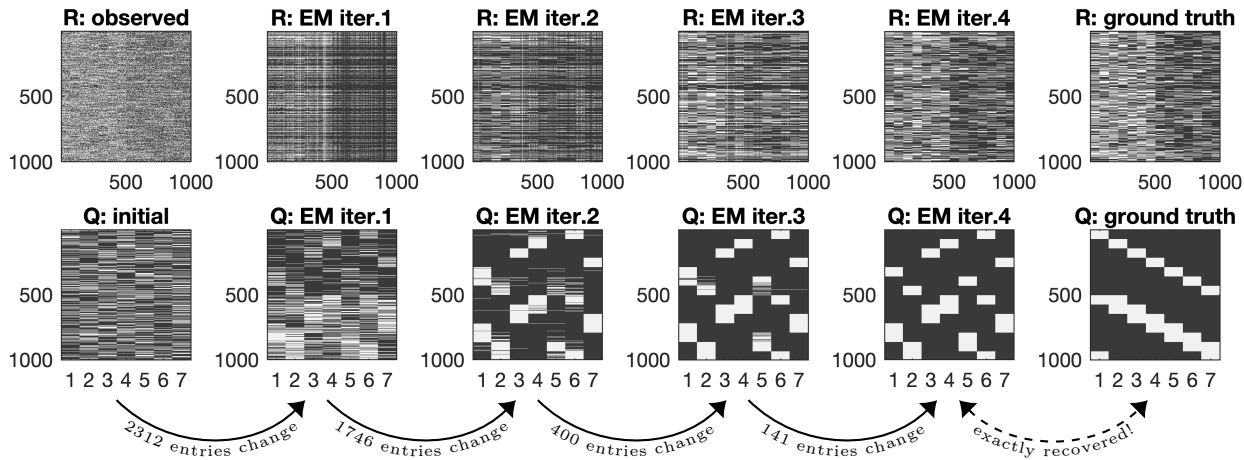
\begin{figure}[h!]
	\centering

    \begin{minipage}[b]{\textwidth}
	\centering
	    \includegraphics[width=\textwidth]{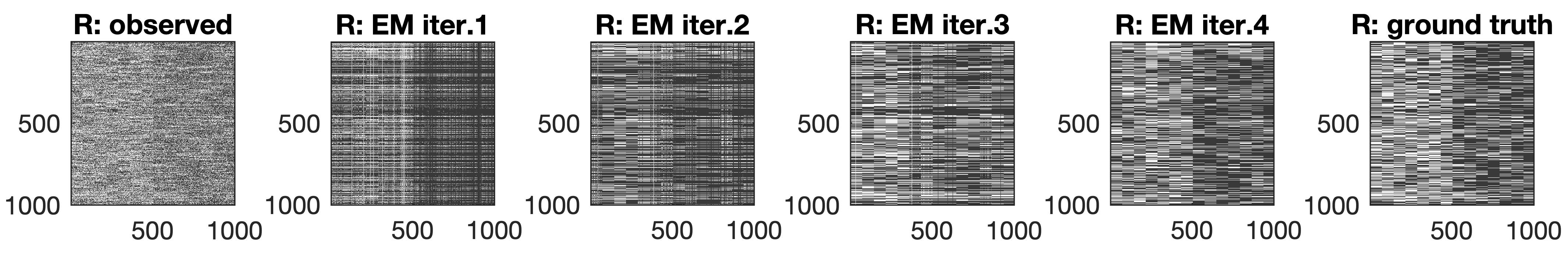}
	\end{minipage}

	\begin{minipage}[b]{1\textwidth}
	    \centering
	    \includegraphics[width=\textwidth]{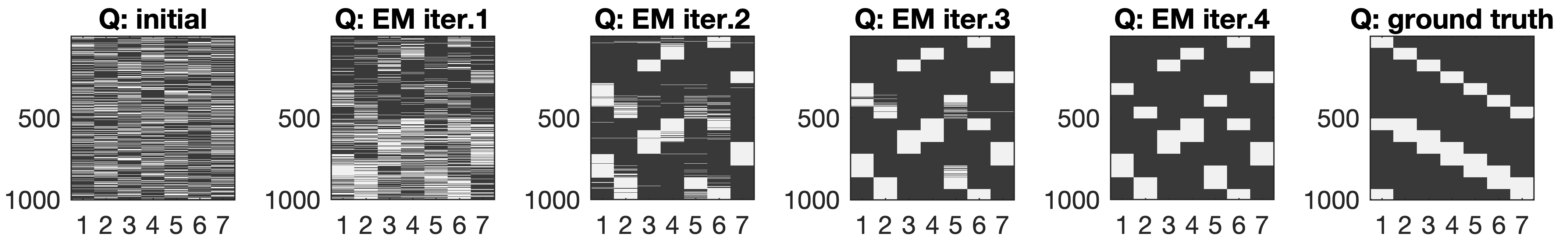}
	\end{minipage}
	
	\vspace{-3mm}
	\begin{tikzpicture}
	\centering
	    \node (hold1) at (0, 0) {};

		\node (Qini) at (1.2, 0) {};
		
		\node (Q1) at (4, 0) {};
		
		\node (Q2) at (6.8, 0) {};

        \node (Q3) at (9.6, 0) {};
        
        \node (Q4) at (12.4, 0) {};
        
        \node (Qtrue) at (15.2, 0) {};
        
        \node (hold2) at (18, 0) {};

        \def\myshift#1{\raisebox{-2.5ex}}
		\draw [->,thick,postaction={decorate,decoration={text along path,text align=center,text={|\tiny\myshift|2312 entries change}}}] (Qini) to [bend right=45]  (Q1);
				
		\draw [->,thick,postaction={decorate,decoration={text along path,text align=center,text={|\tiny\myshift|1746 entries change}}}] (Q1) to [bend right=45]  (Q2);
		
		\draw [->,thick,postaction={decorate,decoration={text along path,text align=center,text={|\tiny\myshift|400 entries change}}}] (Q2) to [bend right=45]  (Q3);
		
		\draw [->,thick,postaction={decorate,decoration={text along path,text align=center,text={|\tiny\myshift|141 entries change}}}] (Q3) to [bend right=45]  (Q4);
		
		\draw [<->,dashed,thick,postaction={decorate,decoration={text along path,text align=center,text={|\tiny\myshift|exactly recovered!}}}] (Q4) to [bend right=45]  (Qtrue);
	\end{tikzpicture}
\caption{Estimation with \textit{entirely random initialization}. Color {white} represents value ``1'' and color {black} represents value ``0''. 
Only 4 stochastic EM iterations of the proposed ADG-EM Algorithm \ref{algo-screen} suffice for almost perfect decomposition and reconstruction. The stochastic $\QQ$ after 4 iterations is identical to the true $\QQ$ after a column permutation. }
\label{fig-bifac-K7}
\end{figure}

In the second visualization example, we use ``entirely random initialization'' to obtain the  $(\QQ_{\text{ini}},\AA_{\text{ini}})$ as input to Algorithm \ref{algo-screen}. 
Figure \ref{fig-bifac-K7} shows the results of Algorithm \ref{algo-screen} together with its intermediate results along the first 4 iterations of the stochastic EM steps.
Along the first 4 stochastic EM iterations, the matrix $\QQ$ changed 2312, 1746, 400, 141 entries, respectively. Then  from the 5th iteration until the stopping criterion is reached, all the entries of $\QQ$ remain the same. 
With this entirely random initialization mechanism, the finally obtained $\hat\QQ$ only differs from $\QQ_{\text{true}}$ by a column permutation. This  permutation of the latent attributes is the inevitable and trivial ambiguity associated with estimating a $\QQ$-matrix  \citep{chen2015statistical}.
The proposed ADG-EM algorithm also succeeds in this scenario.
For Figure \ref{fig-bifac-K7}, the reconstruction result for the data matrix $\RR$ with noise rate $30\%$ is $7.20\times 10^{-5}$.
This high reconstruction accuracy shows that estimating $\QQ$ up to a column permutation does not compromise reconstructing $\RR$ at all.

\subsection{Simulations Examining the Convergence of Algorithm \ref{algo-screen}}
\label{subsec-conv2}
In this simulation study, we further examine the convergence behavior of Algorithm \ref{algo-screen}, following up the previous Subsection \ref{subsec-conv}.
We still simulate data in the scenario $(N,J,K)=(1000,1000,7)$ considered in Subsection \ref{subsec-conv}, with $1-\theta_j^+=\theta_j^-=30\%$ for all $j\in[J]$ and the $2^7=128$ latent profiles are approximately evenly distributed. 
The ground truth $1000\times 7$ matrix $\QQ$ is visualized in the bottom-right plot in Figure \ref{fig-bifac-perturb}, with color white representing value ``1'' and color black representing value ``0''. 
For each of 200 simulated datasets, we apply our ADG-EM Algorithm \ref{algo-screen} alone to estimate $\QQ$ and reconstruct the ideal case $\RR$ using expression \eqref{eq-recon}.
The initializations $\{\QQ_{\text{ini}}\}$'s are obtained from randomly perturbing about one third entries in the true $\QQ$ in each run. Instead of specifying a stopping criterion based on the convergence of the objective function, in the current experiment we just run exactly  10 stochastic EM iterations in Algorithm \ref{algo-screen}; we record the number of entry-differences between the estimated $\QQ$ and the true $\QQ_{\text{true}}$ along each EM iteration, and present the corresponding boxplot in Figure \ref{fig-box}(b). 
In addition, we record the number of entry-differences between $\QQ_{\text{true}}$ and the initial value $\QQ_{\text{ini}}$, which is given as input to the algorithm, and present the boxplot based on 200 runs in Figure \ref{fig-box}(a). 

The two boxplots in Figure \ref{fig-box} show the  convergence performance and estimation accuracy of the proposed ADG-EM algorithm. 
Out of the $1000\times 7=7000$ entries in $\QQ$, although the initialization of $\QQ$ differs from the true one by more than 2000 entries on average, after just one stochastic EM iteration, the number of entry-differences between $\QQ_{\text{iter.}~1}$ and $\QQ_{\text{true}}$  decreases to less than $300$ entries in most cases. After just 3 stochastic EM iterations, for a vast majority of the 200 datasets, the $\QQ_{\text{true}}$ is perfectly recovered and remains unchanged in further iterations of the algorithm. Indeed, after 10 iterations, for each of the 200 simulated datasets, the $\QQ_{\text{true}}$ is exactly recovered. 

\begin{figure}[h!]
\centering
	\begin{minipage}[c]{0.35\textwidth}
	\centering
	  \includegraphics[width=\textwidth]{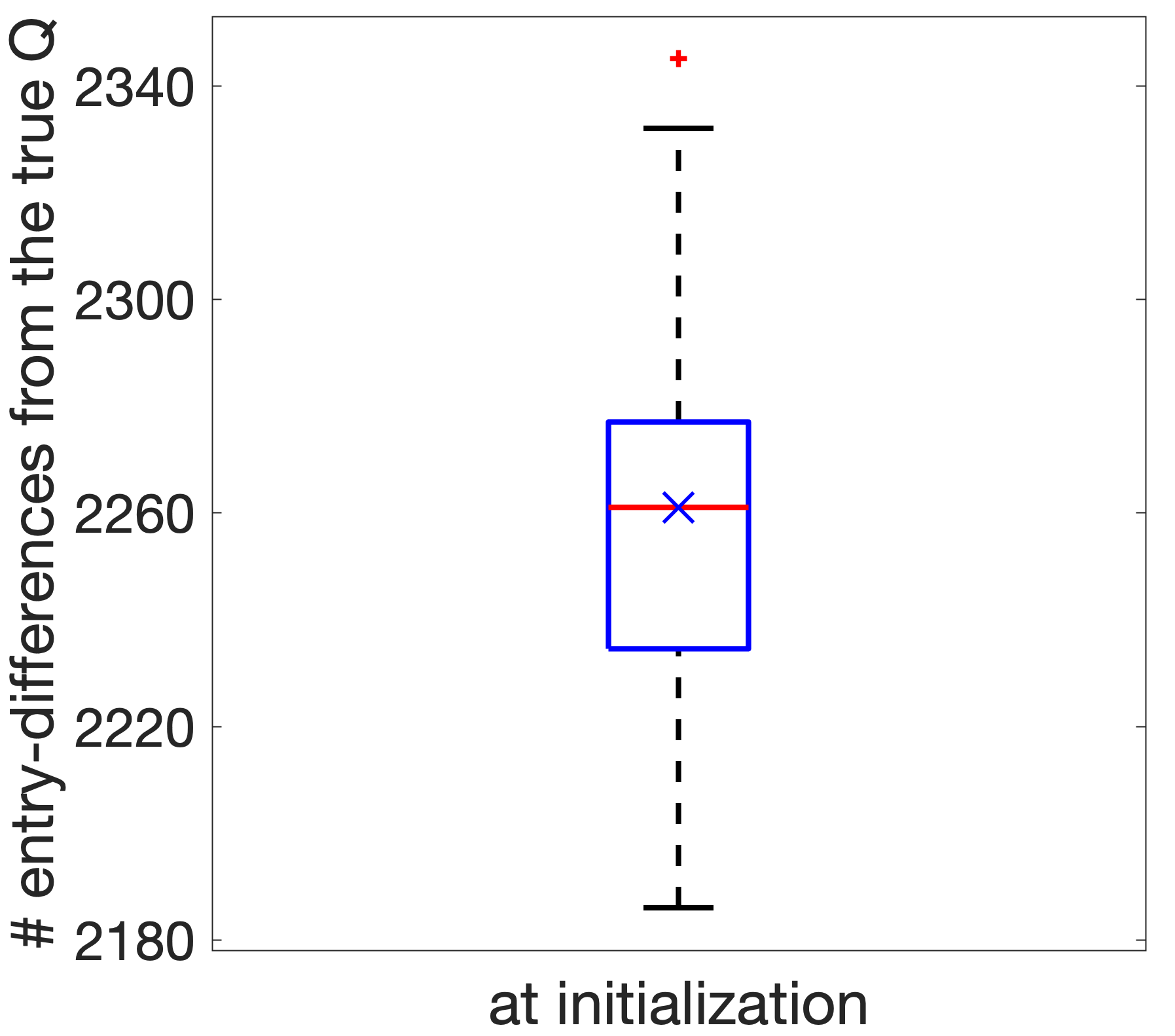}
	\end{minipage}
	\hfill
	\begin{minipage}[c]{0.62\textwidth}
	\centering
\includegraphics[width=\textwidth]{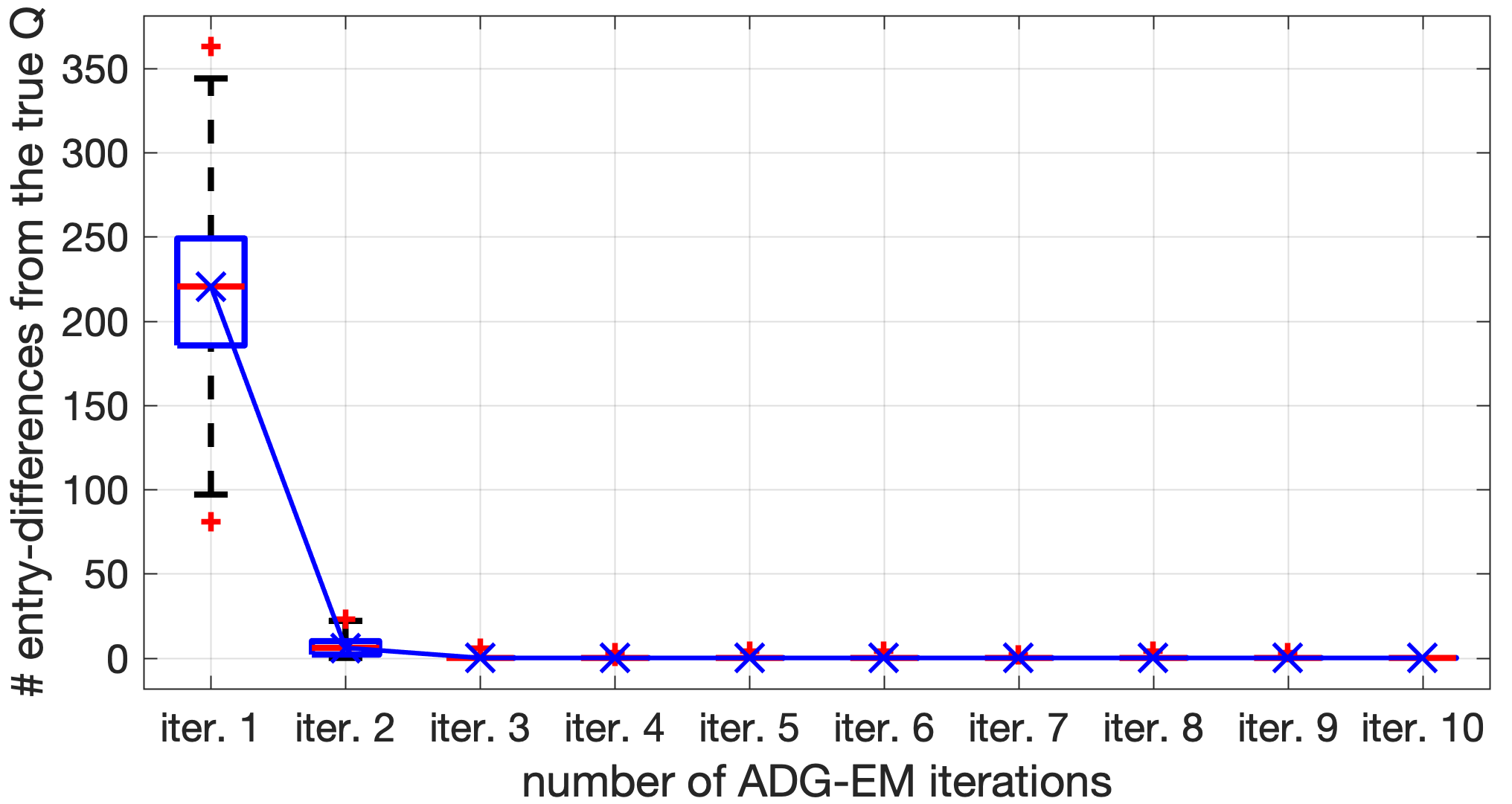}
	\end{minipage}

	\begin{minipage}[b]{0.33\textwidth}
		\centering
		(a) $\#$ entry-differences between $\QQ_{\text{ini}}$ and $\QQ_{\true}$
	\end{minipage}\hfill
	\begin{minipage}[b]{0.63\textwidth}
		\centering
		(b) $\#$ entry-differences between $\QQ_{\text{iter}}$ and $\QQ_{\true}$
	\end{minipage}

\caption{Algorithm \ref{algo-screen}'s convergence behavior. (a): boxplot of entry-differences between the initialization $\QQ_{\text{ini}}$ and the true $\QQ_{\text{true}}$, with size $1000\times 7$; (b) entry-differences between $\QQ_{\text{iter}}$  and $\QQ_{\true}$ along the first ten iterations. Results are based on 200 simulations. }
	\label{fig-box}
\end{figure}


\subsection{One-stage estimation corresponding to Theorem \ref{thm-mis-strong}}\label{subsec-1stage}
When the data-generating multi-parameter model satisfies the stronger Assumption \ref{as-strong}, Theorem \ref{thm-mis-strong} indicates that directly maximizing the misspecificed two-parameter likelihood suffices for estimating $(\QQ,\AA)$ consistently.
We present the following example to illustrate the behavior of the proposed method in this scenario of the existence of ``stronger two-parameter signal''.
Similar to Example \ref{exp-weak} in the main text, we still generate data with $(N,J,K)=(2400,1200,3)$ under the multi-parameter GDINA model proposed in \cite{dela2011}.
The true $\QQ$ still takes the same form as that in the previous Subsection \ref{subsec-conv}. The difference is on the specification of item parameters. 
Here we set $\theta_{j,\zero_K}=0.2$ and $\theta_{j,\one_K}=0.8$, and set the $\mu_{j,\mathcal K_j}$ corresponding to the highest order of interaction among the required attributes to be $\frac{1}{2}(\mu_{j,\one_K} - \mu_{j,\zero_K})$. And we set all the remaining interaction-effect and main-effect parameters to be equal. That is, the $\mu$-parameters in \eqref{eq-alleff} are
\begin{align*}
    \theta_{j,\zero_K} =&~ \mu_{j,\varnothing} = 0.2,\qquad
    \theta_{j,\one_K} = \sum_{S\subseteq \mt K_j}\mu_{j,S} = 0.8;\\
    \mu_{j,\mt K_j} =&~ \frac{\theta_{j,\one_K} - \mu_{j,\zero_K}}{2},
    \quad
	\mu_{j,S} = \frac{\theta_{j,\one_K} - \theta_{j,\zero_K}}{2} \cdot \frac{1}{2^{|\mt K_j|}-2} ~~\text{for any}~~  S\subsetneqq\mathcal K_j,~S\neq\varnothing.
\end{align*}
Figure \ref{fig-strong} presents the estimation results. It shows that in this example, the misspecified MLE has perfect performance on recovering both $\QQ^0$ and $\AA^0$ exactly.

\begin{figure}[h!]
\centering
\includegraphics[width=0.85\linewidth]{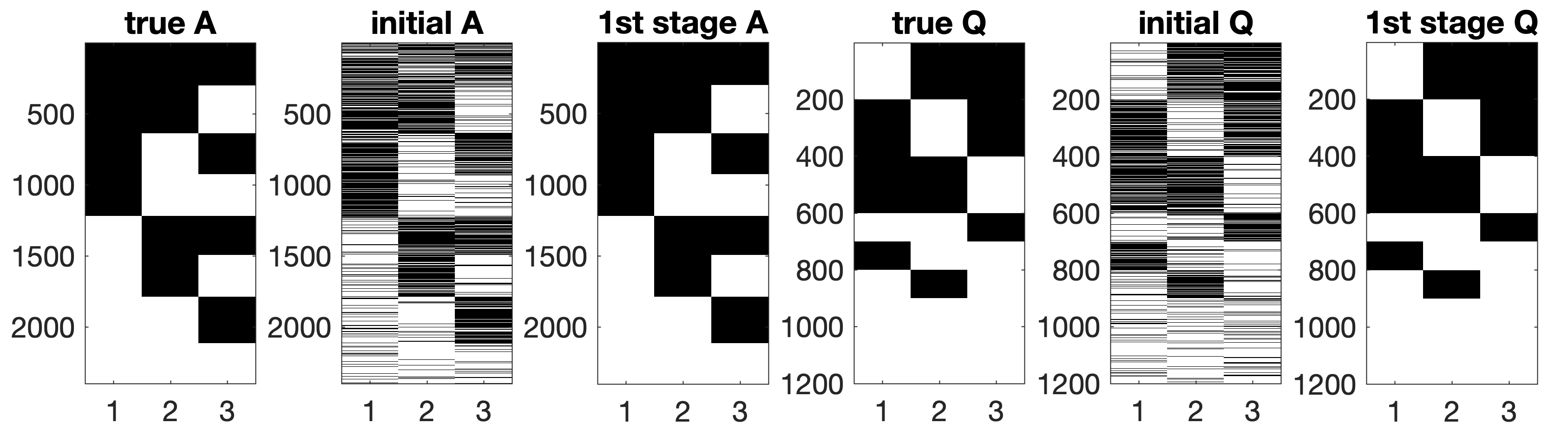}

\includegraphics[width=0.85\linewidth]{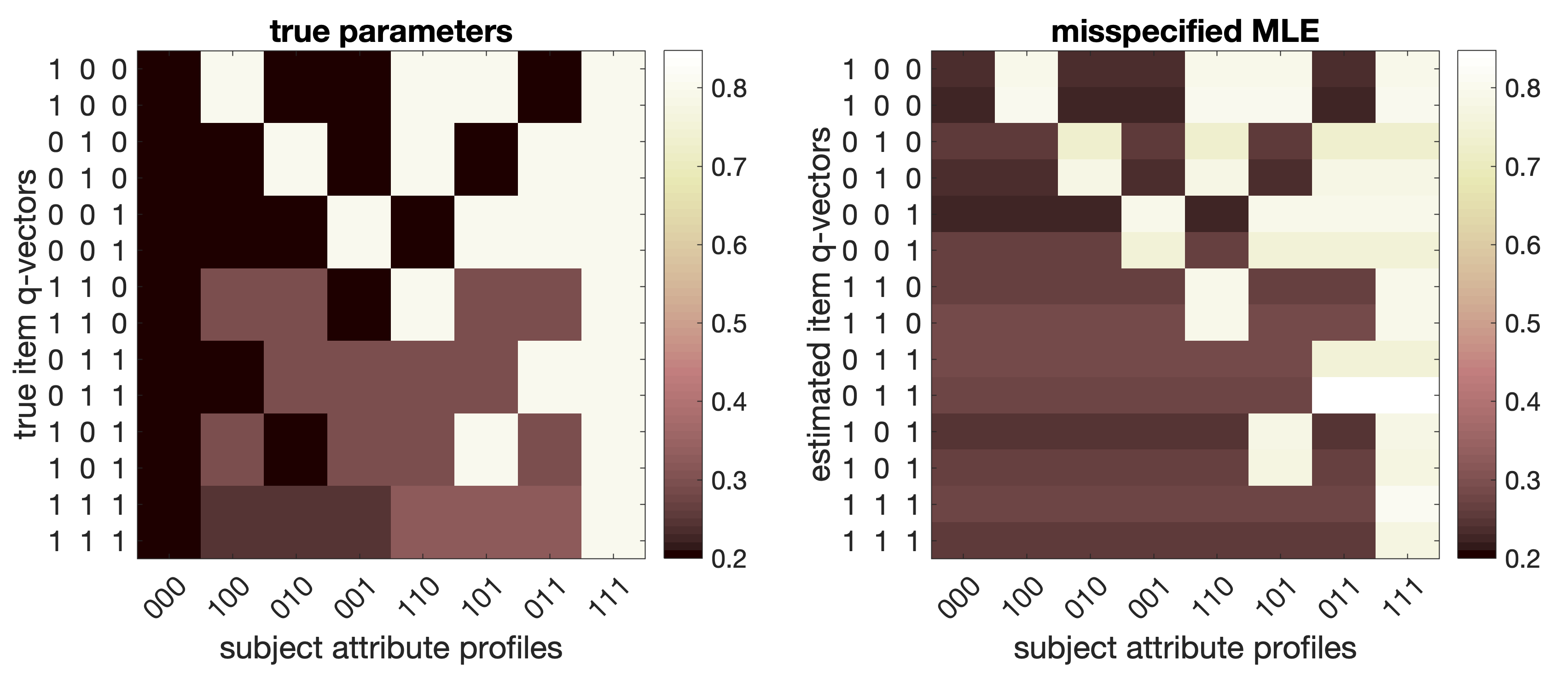}
\caption{\textit{Stronger} two-parameter signal (Assumption \ref{as-strong}) with $\mu_{j,\mathcal K_j}=\left(\theta_{j,\one_K}-\theta_{j,\zero_K}\right)/2$. 
\textbf{Upper row}: $\QQ^0,\QQ^{1\text{st}}_{\text{\normalfont{ini}}},\hat\QQ$ and $\AA^0,\AA_{\text{\normalfont{ini}}},\hat\AA$, where white denotes ``1'' and black denotes ``0''.
 \textbf{Bottom left}: true item parameters $\TT$ for 14 items out of the $J=1200$ items used to generate data under the GDINA model. \textbf{Bottom right}: estimated item parameters $\hat\TT$ corresponding to the same 14 items from the misspecified two-parameter MLE.}
 \label{fig-strong}
\end{figure}

\section{Technical Proofs}\label{sec-pfs}
We introduce some useful notation to facilitate the proofs.
Given a specific modeling assumption such as the two-parameter model or the multi-parameter model, the two binary matrices $\QQ$ and $\AA$ define the ideal response structure $\phi(\qq_j,\aa_i)$ under the two-parameter SLAM or $\xi(\qq_j,\aa_i)$ under the multi-parameter SLAM, as introduced in the main text.
If viewing $\qq_j$ as fixed and varying $\aa_i\in\{0,1\}^K$, item $j$ induces a ``local'' latent class model which categorizes the $2^K$ latent attribute patterns into several classes based on the ideal responses. In particular, under a two-parameter SLAM such as the DINA model, the number of local latent classes induced by each item is always 2, while that under a multi-parameter SLAM such as the GDINA model is $L_j:=2^{K_j}$, where $K_j=\sum_{k=1}^K q_{j,k}$ denotes the number of attributes measured by item $j$.
Therefore, we use a general notation $\ZZ=(z_{i,j})$ to refer to the collection of the latent class structures across all the items $j\in\{1,\ldots,J\}$, where $z_{i,j}$ denotes the latent class membership of individual $i$ for item $j$. 
Then under a SLAM with $K$ latent attributes, the index $i$ for $z_{i,j}$ can vary across all the $2^K$ latent classes; we also write $L=2^K$ for brevity.
Using this notation, we can denote by $\theta_{j,z_{i,j}}$ the item parameter of item $j$ corresponding to the local latent class that individual $i$ belongs to. 
For notational simplicity, we sometimes slightly abuse the notation and write $\theta_{j,z_{i,j}}$ simply as $\theta_{j,z_{i}}$. 
Denote the true parameters that generate the data by $(\TT^0,\,\QQ^0,\,\AA^0) = (\TT^{\true}, \QQ^{\true}, \AA^{\true})$ for notational simplicity.
Define
\begin{align}\label{eq-def-pij}
	P_{i,j}=&~\mathbb P(r_{i,j}=1)=\theta^0_{j,z_i^0}.
\end{align}
Denote
\begin{align*}
	\ell(\mathbf Z, \TT\mid \mathbf R) 
	= \sum_{i=1}^N \sum_{j=1}^J \Big\{ r_{i,j}\log(\theta_{j,z_i}) + (1-r_{i,j})\log(1-\theta_{j,z_i}) \Big\}.
\end{align*}
Denote the expectation of the above $\ell(\mathbf Z, \, \TT\mid \mathbf R)$ by 
\begin{align}\label{eq-ell}
	\bar\ell(\mathbf Z, \TT) 
	= \ME [\ell(\mathbf Z, \, \TT\mid \mathbf R)]
	=&~ \sum_{i=1}^N \sum_{j=1}^J \Big\{ P_{i,j}\log(\theta_{j,z_i}) + (1-P_{i,j})\log(1-\theta_{j,z_i}) \Big\}.
\end{align}
Then there is $\bar\ell(\mathbf Z, \TT) = \mathbb E[ \ell(\RR;\,\mathbf Z, \, \TT)]$, where the expectation is taken with respect to the distribution of $\RR$.

Given arbitrary $(\QQ,\,\AA)$, denote 
\begin{align}\label{eq_proll}
\ell(\RR;\, \QQ,\,\AA) =&~ \sup_{\TT} \ell(\RR;\,\mathbf Z, \, \TT) = \ell(\RR;\, \QQ,\,\AA, \hat\TT^{(\QQ,\,\AA)}),\\ \notag
\quad
\bar\ell(\QQ,\,\AA) =&~ \sup_{\TT} \bar\ell(\mathbf Z, \, \TT) = \bar\ell(\QQ,\,\AA, \bar\TT^{(\QQ,\,\AA)}),
\end{align}
where $\hat\TT^{(\QQ,\,\AA)} = \arg\max_{\TT} \ell(\RR;\,\mathbf Z, \, \TT)$ and $\bar \TT^{(\QQ,\,\AA)} = \arg\max_{\TT}\bar\ell(\mathbf Z,\TT)$. Then under any realization of $\mathbf Z$, the following holds for any latent class $a\in\{1,\ldots,2^K\}$,
\begin{equation}\label{eq-zmle}
	\hat\theta^{(z)}_{j,a} = \frac{\sum_{i}Z_{i,a} r_{i,j}}{\sum_{i}Z_{i,a}},\quad
	\bar\theta^{(z)}_{j,a} = \frac{\sum_{i}Z_{i,a} P_{i,j}}{\sum_{i}Z_{i,a}}.
\end{equation}
Note that $(\hat\QQ,\,\hat\AA) = \argmax_{\QQ,\,\AA} \ell(\RR;\,\QQ,\,\AA,\hat\TT^{(\QQ,\,\AA)})= \argmax_{\QQ,\,\AA} \ell(\RR;\,\QQ,\,\AA)$, where $\hat\TT^{(\QQ,\,\AA)}$ maximizes the profile likelihood $\ell(\RR;\,\mathbf Z, \, \TT)$ given a particular realization $(\QQ,\,\AA)$.
Denote by $I(\cdot)$ the binary indicator function which equals one if the argument inside is true and equals zero otherwise.
In the following, Section \ref{sec-pf1} and Section \ref{sec-pf2} contain the proofs of Theorem \ref{thm-joint-both}, respectively. 
Section \ref{sec-pf34} includes the proofs of Theorems \ref{thm-mis-weak}-\ref{thm-mis-strong}. Section \ref{sec-lemma} gives the proofs of some technical lemmas used in the main proofs.

\subsection{Proof of Theorem \ref{thm-joint-both} for Two-Parameter Models}
\label{sec-pf1}
 We first outline the main steps of the proof as follows and then proceed one by one.
 
\noindent
\textbf{Outline of the proof of part (a).}

\noindent\textbf{Step 1}: Express $\ell(\RR;\, \ZZ) - \bar \ell(\ZZ)$ in terms of $\sum_{j}\sum_{a}n_{j,a}D(\hat \theta_{j,a}\| \bar \theta_{j,a}) + X - \mathbb E(X)$, where $X$ is a random variable depending on $\RR$ and $\bar\TT^{(\ZZ)}$ under $\ZZ$, and $n_{j,a} = \sum_{i=1}^N I(\text{subject $i$ is in class $a$})$.

\noindent\textbf{Step 2}: Bound the first term  $\sum_{j}\sum_{a}n_{j,a}D(\hat \theta_{j,a}\| \bar \theta_{j,a})$ in the above display uniformly over all possible $\ZZ$.

\noindent\textbf{Step 3}: Bound the second term $X-\ME(X)$. Combine this and Step 2  to obtain a bound for $\sup_{\ZZ} |\ell(\RR;\, \ZZ) - \bar \ell(\ZZ)|$.

\noindent\textbf{Step 4}: (Denote the true latent class memberships by $\ZZ^0$ and the those maximizing the likelihood by $\hat \ZZ$.) Establish $\bar \ell(\ZZ^0) \geq \bar\ell(\ZZ)$ for all $\ZZ$. Use triangle inequality to upper-bound the non-negative quantity $\bar \ell(\ZZ^0) - \bar \ell(\hat \ZZ)$.
\begin{align*}
	0\leq \bar \ell(\ZZ^0) - \bar \ell(\hat \ZZ)
     \leq [\bar \ell(\ZZ^0) - \ell(\RR;\,\ZZ^0)] + 
     [\ell(\RR;\,\ZZ^0) - \ell(\RR;\,\hat\ZZ)]
     + [\ell(\RR;\,\hat\ZZ) - \bar \ell(\hat \ZZ)]
\end{align*}
Since in the above display the middle group of terms $[\ell(\RR;\,\ZZ^0) - \ell(\RR;\,\hat\ZZ)]\leq 0$, we have $0\leq \bar \ell(\ZZ^0) - \bar \ell(\hat \ZZ) \leq 2 \sup_{\ZZ} |\ell(\RR;\, \ZZ) - \bar \ell(\ZZ)|$.

\bigskip
\noindent
\textbf{Outline of the proof of part (b).}

\noindent\textbf{Step 5}: Based on the result obtained in Step 4, obtain the consistency of estimating part of the single-attribute row vectors in $\mathbf Q$ and part of row vectors in $\mathbf A$ under Assumption \ref{cond-id} using an identifiability argument.

\noindent\textbf{Step 6}: Obtain  the consistency of estimating all the row vectors of $\mathbf A$.

\noindent\textbf{Step 7}: Obtain  the consistency of estimating all the row vectors of $\mathbf Q$.

\bigskip
\noindent
\textbf{Proof of Part (a) of Theorem \ref{thm-joint-both}:}
The proof techniques of this part are similar in spirit to the maximum profile likelihood technique in \cite{choi2012sbm} for stochastic block models.
We next proceed step by step as outlined before.

\bigskip

\noindent\textbf{Step 1.} 
Recall $D(p \| q) = p\log(p/q) + (1-p)\log((1-p)/(1-q))$ denotes the Kullback-Leibler divergence of a Bernoulli distribution with parameter $p$ from that with parameter $q$. 
In this step we prove a lemma as follows. The proofs of all the technical lemmas are deferred to Section \ref{sec-lemma}.
 \begin{lemma}
	\label{lem-express}
	Let $(r_{i,j}; \, 1\leq i\leq N, 1\leq j\leq J)$ denote independent Bernoulli trials with parameters $(P_{i,j}; \, 1\leq i\leq N, 1\leq j\leq J)$. Under a general latent class model, given an arbitrary $\ZZ$, there is 
	\begin{align}\label{eq-lemma}
		&~\sup_{\TT} \ell(\RR;\,\mathbf Z, \, \TT) - \sup_{\TT} \mathbb E[\ell(\RR;\,\mathbf Z, \, \TT)] \\ \notag
	=&~  \sum_{j=1}^J \sum_{a=1}^{L_j} n_{j,a}  D(\hat\theta_{j,a} \| \bar \theta_{j,a})  
    + \sum_{i=1}^N\sum_{j=1}^J (r_{i,j} - P_{i,j})\log\Big( \frac{\bar\theta_{j,z_i}}{1-\bar\theta_{j,z_i}} \Big) \\ \notag
    = &~  \sum_{j=1}^J \sum_{a=1}^{L_j}  n_{j,a}  D(\hat\theta_{j,a} \| \bar \theta_{j,a})  
    + X-\ME X,
	\end{align}
where 
$$
X= \sum_{i=1}^N\sum_{j=1}^J r_{i,j}\log\left( \frac{\bar\theta_{j,z_i}}{1-\bar\theta_{j,z_i}} \right)
$$ 
is a random variable depending on $\ZZ$, and $L_j$ denotes the number of ``local'' distinct latent classes induced by  $\qq_j$ for item $j$.
\end{lemma}

\begin{proof}
Please see Page \pageref{pf-lem1}.
\end{proof}

\bigskip
\noindent\textbf{Step 2.} In this step we prove the following lemma.
\begin{lemma}
	\label{lem-kl-2p}
	Under a two-parameter SLAM,
	the following event happens with probability at least $1-\delta$,
\begin{equation*}
\max_{\ZZ}\left\{\sum_{j=1}^J\sum_{a=0,1} n_{j,a} D(\hat\theta^{\,\ZZ}_{j,a}  \| \bar \theta^{\,\ZZ}_{j,a}) \right\}
<
 N\log(2^K) + 2J\log\Big(\frac{N}{2} + 1\Big) - \log\delta.
\end{equation*}
\end{lemma}

\begin{proof}
Please see Page \pageref{pf-lem2}.
\end{proof}

\bigskip
\noindent\textbf{Step 3.} In this step we bound $|X-\mathbb E[X]|$, with $X$ defined in \eqref{eq-defx}. 
Introduce notation $X_{i,j}=r_{i,j}\log (\bar\theta_{j,z_i}/(1-\bar\theta_{j,z_i}))$, then $X = \sum_{i=1}^N \sum_{j=1}^J X_{i,j}$. Under Assumption \ref{assume-pij}, 
there is $|X_{i,j}|\leq d\log J$. 
Then we have 
{$\sum_{i=1}^N \sum_{j=1}^J \ME[X_{i,j}^2] = \sum_{i=1}^N \sum_{j=1}^J \MP(r_{i,j}=1) X^2_{i,j} = \sum_{i}\sum_{j} P_{i,j} X_{i,j}^2 \leq MNJ (d\log J)^2$.}
Applying the Bernstein's inequality to the sum of independent bounded random variables, we have the following holds for any fixed $\ZZ$,
\begin{align*}
	\mathbb P (|X-\mathbb E[X]| \geq \epsilon)
\leq &~2\exp\left\{ -\frac{(1/2)\epsilon^2}{\sum_i\sum_j \ME[X_{i,j}^2]+ (2/3)\epsilon\log J } \right\}\\
\leq &~ 2\exp\left\{ -\frac{(1/2)\epsilon^2}{d^2 MNJ(\log J)^2 + (2/3)\epsilon\log J} \right\}.
\end{align*}

We next prove the following proposition.
\begin{proposition}
	\label{prop-scale1}
	Under the following scaling for some small positive constant $c>0$,
\begin{equation}\label{eq-scale}
\sqrt{J}\cdot 2^K = o (N^{1-c}),
\end{equation}
we have $$\frac{1}{NJ}\max_{\ZZ} |\ell(\RR;\,\ZZ) - \ME \ell(\RR;\,\ZZ)| 
	= o_P\left(\frac{\sqrt{M\log 2^K}}{\sqrt{J}} (\log J)^{1+\epsilon}\right).
	$$
\end{proposition}

\begin{proof}[Proof of Proposition \ref{prop-scale1}]
Combining the results of Step 2 and Step 3, since that there are $(2^K)^N$ possible assignments of $\ZZ$, we apply the union bound to obtain
\begin{align}\label{eq-deltanj}
	&~\MP (\max_{\ZZ} |\ell(\RR;\,\ZZ) - \ME \ell(\RR;\,\ZZ)| \geq 2\epsilon \delta_{NJ})\\ \notag
\leq &~ (2^K)^N \MP \left[
\left\{\sum_{j=1}^J\sum_{a=0,1} n_{j,a} D(\hat\theta_{j,a} \| \bar \theta_{j,a})\geq \epsilon \delta_{NJ}\right\}
\cup
\left\{|X - \ME[X]| \geq \epsilon \delta_{NJ}\right\}
\right]\\ \notag
\leq &~  \exp\Big\{N\log(2^K) + J2^K\log\Big(\frac{N}{2^K} + 1\Big) - \epsilon \delta_{NJ}\Big\} \\ \notag
&~  + 2\exp\Big\{N\log 2^K -\frac{\epsilon^2 \delta_{NJ}}{2d^2(MNJ/\delta_{NJ})(\log J)^2 + (4/3)\epsilon\log J) } \Big\}.
\end{align}
In order for the second term on the right hand side of the above display to go to zero, the following of $\delta_{NJ}$ would suffice,
\begin{equation}
	\label{eq-scale}
	\delta_{NJ} \succsim N\sqrt{MJ\log 2^K} \log J.
\end{equation}
We take $\delta_{NJ} = N\sqrt{MJ\log 2^K} (\log J)^{1+\epsilon}$ for a small positive constant $\epsilon$. Further, under this $\delta_{NJ}$, in order for the first term on the right hand side of \eqref{eq-deltanj} to go to zero, 
Then the right hand side of \eqref{eq-deltanj} goes to zero as $N, J$ go large. 
Then the scaling $\sqrt{J} = O(\sqrt{M} N^{1-c})$ and $K=o(MJ\log J)$ described in the theorem
 yields $\MP (\max_{\ZZ} |\ell(\RR;\,\ZZ) - \ME \ell(\RR;\,\ZZ)| \geq 2\epsilon \delta_{NJ}) = o(1)$, which implies
\begin{equation}
	\label{eq-llrate}
	\frac{1}{NJ}\max_{\ZZ} |\ell(\RR;\,\ZZ) - \ME \ell(\RR;\,\ZZ)| 
	=o_P\left(\frac{\sqrt{M\log 2^K}}{\sqrt{J}} (\log J)^{1+\epsilon}\right).
\end{equation}
This proves Proposition \ref{prop-scale1}.
\end{proof}

\bigskip
\noindent\textbf{Step 4.}
Denote the true class assignments by $\ZZ^0$. We first establish 
\begin{equation}
	\label{eq-tr}
	\bar \ell(\ZZ^0) \geq \bar\ell(\ZZ), \quad\text{for all }\ZZ.
\end{equation}
First note that $\theta^0_{j,z_i^0} = P_{i,j}$, and 
$$
\bar\theta_{j,z_i^0} = \frac{\sum_{m=1}^N Z_{m, z_i^0}^0 P_{m,j}}{\sum_{m=1}^N Z^0_{m, z_i^0}} 
= \frac{\sum_{m=1}^N Z_{m, z_i^0}^0 P_{i,j}}{\sum_{m=1}^N Z^0_{m, z_i^0}} = P_{i,j}.
$$
The difference $\bar \ell(\ZZ^0) - \bar\ell(\ZZ)
$ can be written as
\begin{align*}
	\bar \ell(\ZZ^0) - \bar\ell(\ZZ)
=&~ \sum_{i=1}^N \sum_{j=1}^J \left[ P_{i,j} \log\left(\frac{\bar\theta^0_{j,z_i^0}}{\bar\theta^{\ZZ}_{j,z_i}}\right) + (1-P_{i,j})\log\left(\frac{1-\bar\theta^0_{j,z_i^0} }{1- \bar\theta^{\ZZ}_{j,z_i}}\right)\right] \\
=&~ \sum_{i=1}^N \sum_{j=1}^J \left[ P_{i,j} \log\left(\frac{P_{i,j}}{\bar\theta^{\ZZ}_{j,z_i}}\right) + (1-P_{i,j})\log\left(\frac{1-P_{i,j} }{1- \bar\theta^{\ZZ}_{j,z_i}}\right)\right]\\
=&~ \sum_{i=1}^N \sum_{j=1}^J D(P_{i,j} \| \bar\theta^{\ZZ}_{j,z_i})\geq 0,
\end{align*}
therefore establishing \eqref{eq-tr}. Since the above holds for every $\ZZ$, it also holds for the maximum likelihood estimator $\hat\ZZ$. We further upper bound $\bar \ell(\ZZ^0) - \bar\ell(\ZZ)$ from above as follows,
\begin{align*}
	0\leq \bar \ell(\ZZ^0) - \bar \ell(\hat \ZZ)
     \leq [\bar \ell(\ZZ^0) - \ell(\RR;\,\ZZ^0)] + 
     \underbrace{[\ell(\RR;\,\ZZ^0) - \ell(\RR;\,\hat\ZZ)]}_{\leq 0}
     + [\ell(\RR;\,\hat\ZZ) - \bar \ell(\hat \ZZ)],
\end{align*}
where $[\ell(\RR;\,\ZZ^0) - \ell(\RR;\,\hat\ZZ)] \leq 0$ results from the definition of $\hat\ZZ$ as the MLE, that is $\ZZ$ maximizes the $\ell(\RR;\,\ZZ,\hat\TT^{\ZZ})$. Therefore
\begin{align*}
	0\leq \bar \ell(\ZZ^0) - \bar \ell(\hat \ZZ)
	 \leq &~ [\bar \ell(\ZZ^0) - \ell(\RR;\,\ZZ^0)] +  [\ell(\RR;\,\hat\ZZ) - \bar \ell(\hat \ZZ)]
	 \\
	 \leq &~ 2\sup_{\ZZ} |\bar \ell(\ZZ) - \ell(\RR;\,\ZZ)|
	 = 2 \sup | \ell(\RR;\,\ZZ) - \ME  \ell(\RR;\,\ZZ)| \\
	 =&~ N\sqrt{JM\log 2^K} (\log J)^{1+\epsilon}.
\end{align*}
So we obtained $\bar \ell(\ZZ^0) - \bar \ell(\hat \ZZ) = o_P(NJ\cdot \gamma_{J})$.

In the following, we sometimes 
denote $\xi(\qq_j,\aa_i)$ by $\xi_{i,j}$ 
for notational convenience. 
Under a two-parameter SLAM,
\begin{align}\label{eq-lower}
	o_P(NJ\cdot \gamma_{J}) = &~ \bar \ell(\ZZ^0) - \bar \ell(\hat \ZZ) \\ \notag
    =&~\sum_{i=1}^N\sum_{j=1}^J \left[P_{i,j}\log\left(\frac{\theta^{\ZZ^0}_{j,z_i^0}}{\bar\theta^{\hat\ZZ}_{j,\hat z_i}}\right)
    + (1-P_{i,j})\log\left(\frac{1-\theta^{\ZZ^0}_{j,z_i^0}}{1-\bar\theta^{\hat\ZZ}_{j,\hat z_i}}\right) \right]\\ \notag
= &~ \sum_{i=1}^N\sum_{j=1}^J D(P_{i,j}\| \bar\theta^{\hat\ZZ}_{j,\hat z_i}) \geq
\sum_{i=1}^N\sum_{j=1}^J 
2\left(P_{i,j} -\bar\theta^{\hat\ZZ}_{j,\hat z_i}\right)^2\\ \notag
= &~ \sum_{i=1}^N\sum_{j=1}^J 
2\left( \MP(r_{i,j}=1\mid \QQ^{\true},\,\AA^{\true},\,\TT^{\true}) - \MP(r_{i,j}=1\mid \hat\QQ,\,\hat\AA,\,\TT^{\true}) \right)^2.
\end{align}
Thus far, part (a) of Theorem \ref{thm-joint-both} for two-parameter SLAMs is proved.

\bigskip
\noindent
\textbf{Proof of Part (b) of Theorem \ref{thm-joint-both} for two-parameter models:}

\noindent
\textbf{Step 5.}
Under a two-parameter model, given $\hat\ZZ$-induced $\{\hat\xi_{i,j}\}$, define
\begin{align*}
	N^j_{ab} = \sum_{i=1}^N I(\xi^0_{i,j}=a)I(\hat\xi_{i,j}=b),\quad (a,b)\in\{0,1\}^2,
\end{align*}
then there is
\begin{align}\notag
	\bar\theta_{j,+}^{\hat\ZZ} 
	=&~
	\frac{\sum_{i=1}^N I(\hat\xi_{i,j}=1)P_{i,j}}{\sum_{i=1}^N I(\hat\xi_{i,j}=1)} \\
	=&~
	\frac{\sum_{i=1}^N I(\hat\xi_{i,j}=1)\cdot[I(\xi^0_{i,j}=1)\theta^0_{j,+}+I(\xi^0_{i,j}=0)\theta^0_{j,-}]}{\sum_{i=1}^N I(\hat\xi_{i,j}=1)}
	=
	\frac{N^j_{11}\theta^0_{j,+}+N^j_{01}\theta^0_{j,-}}{N^j_{11}+N^j_{01}}\\
	\notag
\bar\theta_{j,-}^{\hat\ZZ} 
	=&~
	\frac{\sum_{i=1}^N I(\hat\xi_{i,j}=0)\cdot[I(\xi^0_{i,j}=1)\theta^0_{j,+}+I(\xi^0_{i,j}=0)\theta^0_{j,-}]}{\sum_{i=1}^N I(\hat\xi_{i,j}=0)}
	=
	\frac{N^j_{10}\theta^0_{j,+}+N^j_{00}\theta^0_{j,-}}{N^j_{10}+N^j_{00}}.
\end{align}
Under $\hat\ZZ$, we impose a natural constraint $\bar\theta_{j,+}^{\hat\ZZ} > \bar\theta_{j,-}^{\hat\ZZ}$ for identifiability purpose, then the above representation together with $\theta^0_{j,+}>\theta^0_{j,-}$ yields
\begin{align}\label{eq-estid}
    \frac{N^j_{11}\theta^0_{j,+}+N^j_{01}\theta^0_{j,-}}{N^j_{11}+N^j_{01}}
    >&~
	\frac{N^j_{10}\theta^0_{j,+}+N^j_{00}\theta^0_{j,-}}{N^j_{10}+N^j_{00}},
	\quad \Longrightarrow \quad
	N_{11}^j N_{00}^j  > N_{10}^j N_{01}^j.
\end{align}
Under Assumption \ref{cond-id},  the Pinsker's inequality \citep{csiszar2011info} between Kullback-Leibler divergence and total variation distance gives
\begin{align*}
	D(\theta_{j,+} \| \theta_{j,-}) 
=&~ \theta_{j,+}\log(\theta_{j,+}/\theta_{j,-}) + (1-\theta_{j,+})\log((1-\theta_{j,+})/(1-\theta_{j,-})) \\
\geq &~ \frac{1}{2} (|\theta_{j,+}-\theta_{j,+}| + |(1-\theta_{j,+})-(1-\theta_{j,-})|)^2 {
= 2(\theta_{j,+}-\theta_{j,-})^2
\geq 2\beta_J.
}
\end{align*}

Therefore,
\begin{align*}
& D(\theta^0_{j,-} \| \bar\theta^{\hat\ZZ}_{j,-})
\geq 
2\left[\frac{N^j_{10} \cdot (\theta_{j,+}^0 - \theta_{j,-}^0)}{N^j_{10}+N^j_{00}}
\right]^2,
\quad
	D(\theta^0_{j,+} \| \bar\theta^{\hat\ZZ}_{j,-})
\geq 
2\left[\frac{N^j_{00} \cdot (\theta_{j,+}^0 - \theta_{j,-}^0)}{N^j_{10}+N^j_{00}}
\right]^2,
\\
& D(\theta^0_{j,+} \| \bar\theta^{\hat\ZZ}_{j,+})
\geq 
2\left[\frac{N^j_{01} \cdot (\theta_{j,+}^0 - \theta_{j,-}^0)}{N^j_{11}+N^j_{01}}\right]^2,
\quad
D(\theta^0_{j,-} \| \bar\theta^{\hat\ZZ}_{j,+})
\geq 
2\left[\frac{N^j_{11} \cdot (\theta_{j,+}^0 - \theta_{j,-}^0)}{N^j_{11}+N^j_{01}}\right]^2.
\end{align*}
We also have the following representation
\begin{align}\notag
	\theta^{0}_{j,z_i^0} 
	=&~ \xi^0_{i,j}\theta^{0}_{j,+} + (1-\xi^0_{i,j})\theta^{0}_{j,-}, 
\\ \label{eq-theta-xi}
\bar\theta^{\hat\ZZ}_{j,\hat z_i} 
=&~
		\hat\xi_{i,j} \bar\theta_{j,+}^{\hat\ZZ} 
		+
		 (1-\hat\xi_{i,j})\bar\theta_{j,-}^{\hat\ZZ}
=
		\hat\xi_{i,j}\frac{N^j_{11}\theta^0_{j,+}+N^j_{01}\theta^0_{j,-}}{N^j_{11}+N^j_{01}} 
		+
		 (1-\hat\xi_{i,j})\frac{N^j_{10}\theta^0_{j,+}+N^j_{00}\theta^0_{j,-}}{N^j_{10}+N^j_{00}}. 
\color{black}
\end{align}
Therefore,
\begin{align}\notag
&~ \bar\ell(\ZZ^0) - \bar\ell(\hat\ZZ) 
=\sum_{i=1}^N\sum_{j=1}^J D(P_{i,j}\| \bar\theta^{\hat\ZZ}_{j,\hat z_i})  \\ \notag
= &~ \sum_{j=1}^J
\Big[
N^j_{00}\cdot D(\theta^0_{j,-} \| \bar\theta^{\hat\ZZ}_{j,-})
+
N^j_{10}\cdot D(\theta^0_{j,+} \| \bar\theta^{\hat\ZZ}_{j,-}) +
N^j_{01}\cdot D(\theta^0_{j,-} \| \bar\theta^{\hat\ZZ}_{j,+})
+
N^j_{11}\cdot D(\theta^0_{j,+} \| \bar\theta^{\hat\ZZ}_{j,+})
\Big]
\\ \notag
\geq &~
\sum_{j=1}^J \left[
\left\{\frac{N^j_{00} (N^j_{10})^2 + N^j_{10} (N^j_{00})^2 }{(N^j_{10}+N^j_{00})^2}
+ \frac{N^j_{01} (N^j_{11})^2 + N^j_{11} (N^j_{01})^2 }{(N^j_{10}+N^j_{00})^2}
\right\}
\cdot 2(\theta^0_{j,+}-\theta^0_{j,-})^2
\right]\\ \notag
\geq &~
\sum_{j=1}^J \left[
\left\{
\frac{N^j_{10}N^j_{00}}{N^j_{10}+N^j_{00}}
+ \frac{N^j_{11}N^j_{01}}{N^j_{11}+N^j_{01}}\right\}
\cdot 2(\theta^0_{j,+}-\theta^0_{j,-})^2
\right]\\ \notag
\geq &~
2\beta_{J} \cdot \sum_{j=1}^J
\left(\frac{N^j_{10}N^j_{00}}{N^j_{10}+N^j_{00}}
+ \frac{N^j_{11}N^j_{01}}{N^j_{11}+N^j_{01}}
\right) \\ \notag
\geq &~
\beta_{J} \cdot \sum_{j=1}^J (\min\{N^j_{10},N^j_{00}\} + \min\{N^j_{11},N^j_{01}\})
\\ \label{eq-jall}
= &~
\beta_{J} \cdot 
\biggr[\sum_{j:\, N^j_{10}<N^j_{00},\atop N^j_{01}<N^j_{11}} (N^j_{10} + N^j_{01})
+
\sum_{j:\, N^j_{10}<N^j_{00},\atop N^j_{01}>N^j_{11}} (N^j_{10} + N^j_{11})
+
\sum_{j:\, N^j_{10}>N^j_{00},\atop N^j_{01}<N^j_{11}} (N^j_{00} + N^j_{01})\biggr],
\end{align}
where the last equality holds because for each $j\in[J]$, the two events $N^j_{10}>N^j_{00}$ and $N^j_{01}>N^j_{11}$ can not happen simultaneously due to the previously established $N_{11}^j N_{00}^j  > N_{10}^j N_{01}^j$. 
We need the following lemma.
\begin{lemma}\label{lem-j012}
Recall $\min_{\aaa\in\{0,1\}^K} \frac{1}{N}\sum_{i=1}^N I(\aa^0_i=\aaa) \geq p_{N} \geq \epsilon/2^K$ in Assumption \ref{cond-id}. 	Define the following random sets depending on the maximum likelihood estimator $\hat\ZZ$,
\begin{align*}
	&\mathcal J_0 = \{j\in[J]:\, N^j_{10}<N^j_{00},\, N^j_{01}<N^j_{11}\}; \\
	&\mathcal J_1 = \{j\in[J]:\, N^j_{10}<N^j_{00},\, N^j_{01}>N^j_{11}\}; \\
    &\mathcal J_2 = \{j\in[J]:\, N^j_{10}>N^j_{00},\, N^j_{01}<N^j_{11}\},
\end{align*}
then $|\mathcal J_1| = o_P(J\cdot\gamma_{J}/\beta_J)$, $|\mathcal J_2| = o_P(J\cdot\gamma_{J}/\beta_J)$, and hence $1 - |\mathcal J_0|/J =o_P(\gamma_{J}/\beta_J).$
\end{lemma}

\begin{proof}
Please see Page \pageref{pf-lem3}.
\end{proof}

Thanks to Lemma \ref{lem-j012}, under Assumption \ref{cond-gap-2p} we now have
\begin{align*}
	o_P(\gamma_{J})\geq 
	 \frac{1}{NJ} \left( \bar\ell(\ZZ^0) - \bar \ell(\hat \ZZ) \right)
	\geq
\sum_{j\in\mathcal J_0} \frac{N^j_{10} + N^j_{01}}{NJ}
=\beta_{J} \cdot \frac{1}{NJ}
\sum_{j\in\mathcal J_0}\sum_{i=1}^N I(\hat\xi_{i,j}\neq \xi_{i,j}^0).
\end{align*}

\bigskip
The following lemma would be useful.
\begin{lemma}\label{lem-qa}
Let $\mathbf Q$ be a $K\times K$ matrix of binary entries, and $\aa^1,\aa^2,\ldots,\aa^{2^K}$ be $K$-dimensional vectors of binary entries. Let $\xi(\qq_j,\aa_i) = \xi^{A}(\qq_j,\aa_i)$ or $\xi(\qq_j,\aa_i) =\xi^{O}(\qq_j,\aa_i)$.
\begin{itemize}
\item[(a)] If $\xi(I_K, \{0,1\}^K) = \xi(\mathbf Q, \{\aa^1,\aa^2,\ldots,\aa^{2^K}\})$, then 
	\begin{equation}\label{eq-equiv}
		\mathbf Q \sim I_K, \quad
		\{\aa^1,\aa^2,\ldots,\aa^{2^K}\} = \{0,1\}^K.
	\end{equation}
\item[(b)] If $\xi(\qq_j, \{0,1\}^K) = \xi(\tilde\qq_j, \{0,1\}^K)$, then 
	\begin{equation}\label{eq-lemqa-2}
		\qq_j\sim\tilde \qq_j.
	\end{equation}	
\end{itemize}
\end{lemma}

\begin{proof}
Please see Page \pageref{pf-lem4}.
\end{proof}

We continue the proof of Step 5. Based on the conclusion of Step 4, we have
\begin{align}\notag
	&~o_P\left(\frac{\gamma_{J}}{\beta_J}\right)
    \geq
\frac{1}{NJ}\sum_{j\in\mathcal J_0}\sum_{i=1}^N I(\hat\xi_{i,j}\neq \xi_{i,j}^0)  = 
\frac{1}{NJ}\sum_{j\in\mathcal J_0}\sum_{i=1}^N I\left( \xi(\qq^0_j,\aa^0_i) \neq \xi(\hat\qq_j,\hat\aa_i)\right).
\end{align}
Next we focus on obtaining a lower bound of the above right hand side under \eqref{eq-infq} in Assumption \ref{cond-id}. We need to introduce some notation.
Consider the set of items $j\in \mathcal J_0$. For each latent attribute $k$, denote by $j_{k}^1$ the smallest integer $j$ such that item $j$ has a $\qq$-vector $\ee_k$, and denote by $j_{k}^2$ the second smallest integer $j$ such that $\qq_j=\ee_k$, etc. 
For each positive integer $m$, denote
\begin{equation}
	\mathcal B^{m} = \{j_{1}^m, j_{2}^m, \ldots, j_K^m\}.
\end{equation}
For each $k\in\{1,\ldots,K\}$, denote
\begin{align}\label{eq-jmin}
	J_{\min} = \min_{1\leq k\leq K}|\{\mathcal J_0:\, \qq^0_j=\ee_k\}|
\end{align}
Then we have that 
\begin{align*}
	\mathcal B^{m} \cap \mathcal B^{l} = \varnothing\text{~~ for any~~}m\neq l,\quad
	\mathcal B^{J,\ee}=
	\bigcup_{m=1}^{J_{\min}} \mathcal B^m \subseteq \mathcal J_0.
\end{align*}
Then the item set $\mathcal B^{J,\ee}$ is a set of some single-attribute items in $\mathcal J_0$. We also denote the remaining items before item $J$ by $\mathcal B^{J,\text{mult}} = [J] \setminus \mathcal B^{J,\ee}$.

Now consider the set of subjects $i\in\{1,2,\ldots N\}$.
For each possible latent attribute pattern $\aaa\in\{0,1\}^K$, denote by $i^1_{\aaa}$ the smallest integer $i$ such that the $i$th subject's latent attribute profile equals $\aaa$, and denote by $i^s_{\aaa}$ the $s$th smallest integer $i$ such that the $i$th subject's latent attribute profile equals $\aaa$. For each positive integer $s$, denote 
\begin{equation}
	\mathcal N^s = \{i^s_{\aaa}:\,\aaa\in\{0,1\}^K\}.
\end{equation}
For each $\aaa\in\{0,1\}^K$, denote
\begin{align}\label{eq-nmin}
	N_{\min} = \min_{\aaa\in\{0,1\}^K}
	     |\{1\leq i\leq N:\, \aa^0_i=\aaa\}|.
\end{align}
Then 
\begin{align*}
	\mathcal N^{s} \cap \mathcal N^{t} = \varnothing\text{~~ for any~~}s\neq t,\quad
	\bigcup_{s=1}^{N_{\min}} \mathcal N^s \subseteq \{1,2,\ldots, N\}.
\end{align*}
Since there is $\mathcal B^m=\{j_1^m,\ldots,j_K^m\}$ and $\mathcal N^s = \{i^s_{\aaa}:\,\aaa\in\{0,1\}^K\}$, 
for notational convenience, we denote
\begin{align}\label{eq-abbre}
    \{\hat \qq_{j_1^m},\ldots,\hat \qq_{j_K^m}\}:=\hat\qq_{\mathcal B^m},\quad
	\{\hat\aa_{i^s_{\aaa}}: \aaa\in\{0,1\}^K\}
	:=
	\hat\aa_{\mathcal N^s}.
\end{align}
So $\hat\qq_{\mathcal B^m}$ represents the set of estimated $\qq$-vectors corresponding to items in the set $\mathcal B^m$, and $\hat\aa_{\mathcal N^s}$ represents the set of estimated attribute patterns corresponding to subjects in the set $\mathcal N^s$. Similarly, we define $\qq^0_{\mathcal B^m}$ and $\aa^0_{\mathcal N^s}$. By definition, there is $\qq^0_{\mathcal B^m}=I_K$ and $\aa^0_{\mathcal N^s}=\{0,1\}^K$.
Now we have
\begin{align*}
	&~ \sum_{i=1}^N \sum_{j\in\mathcal J_0} I\left( \xi(\qq^0_j,\aa^0_i) \neq \xi(\hat\qq_j,\hat\aa_i)\right)
	\geq \sum_{i=1}^N  \sum_{m=1}^{J_{\min}} \sum_{j\in\mathcal B^m}
	        I\left( \xi(\qq^0_j,\aa^0_i) \neq \xi(\hat\qq_j,\hat\aa_i)\right) \\
	\geq &~ \sum_{i=1}^N  
	\sum_{m=1}^{J_{\min}}
	        \sum_{k=1}^K I\left( \xi(\ee_k,\aa^0_i) \neq \xi(\hat\qq_{j_k^m},\hat\aa_i)\right) \quad(\text{since~}\qq_{j^m_k}^0=\ee_k \text{~by definition})\\
	\geq &~  \sum_{s=1}^{N_{\min}} 
	         \sum_{i\in \mathcal N_{\aaa}}     
	         \sum_{m=1}^{J_{\min}}
	        \sum_{k=1}^K 
	        I\left( \xi(\ee_k,\aa_i^0) \neq \xi(\hat\qq_{j_k^m},\hat\aa_i)\right) \\
	\geq &~  \sum_{s=1}^{N_{\min}} 
	         \sum_{\aaa\in\{0,1\}^K}     
	         \sum_{m=1}^{J_{\min}}
	        \sum_{k=1}^K 
	        I\left( \xi(\ee_k,\aaa) \neq \xi(\hat\qq_{j_k^m},\hat\aa_{i^s_{\aaa}})\right) 
	        \quad(\text{since~}\aa_{i^s_{\aaa}}^0=\aaa \text{~by definition})\\
	\geq &~ \sum_{s=1}^{N_{\min}}\sum_{m=1}^{J_{\min}}  I\Big( 
	\xi(\{\ee_k:k\in[K]\}, ~\{\aaa:\aaa\in\{0,1\}^K\}) \neq \\ 
	&\qquad\qquad~~~
	\xi(\{\hat\qq_{j_k^m}:k\in[K]\}, ~\{\hat\aa_{i^s_{\aaa}}: \aaa\in\{0,1\}^K\})
	\Big)\\
	\geq &~ \sum_{s=1}^{N_{\min}}\sum_{m=1}^{J_{\min}}  \frac{1}{2}\left(I(\{\hat\qq_{j_k^m}:k\in[K]\} \nsim I_K)
	 +
	    I(\{\hat\aa_{i^s_{\aaa}}: \aaa\in\{0,1\}^K\} \neq\{0,1\}^K)\right)\\
	    &~ (\text{by Lemma \ref{lem-qa}})\\
   = &~  \frac{N_{\min}}{2}\sum_{m=1}^{J_{\min}} I(\hat \qq_{\mathcal B^m} \nsim \qq^0_{\mathcal B^m})
       +
       \frac{J_{\min}}{2}\sum_{s=1}^{N_{\min}} I(\aaa_{\mathcal N^s} \neq \aaa^0_{\mathcal N^s}),
\end{align*}
so with the $\gamma_{J}$ defined in Theorem \ref{thm-joint-both} we have 
\begin{align}
	\notag
	o_P\left(\frac{\gamma_{J}}{\beta_J}\right)
	&\geq \frac{N_{\min}}{N}\cdot \frac{\sum_{m=1}^{J_{\min}}  I\left(\hat \qq_{\mathcal B^m} \nsim \qq^0_{\mathcal B^m}\right)}{J},\\
	\notag
	o_P\left(\frac{\gamma_{J}}{\beta_J}\right)
	&\geq \frac{J_{\min}}{J}\cdot \frac{\sum_{s=1}^{N_{\min}} I(\hat \aaa_{\mathcal N^s}  \neq \aaa^0_{\mathcal N^s} )}{N}.
\end{align}
Under Assumption \ref{cond-id} and according to Lemma \ref{lem-j012}, there is 
$$\frac{J_{\min}}{J} 
\geq \frac{1}{J}\min_{k\in[K]}\sum_{j=1}^J I(\qq^0_j=\ee_k)\geq
\delta_J$$ 
for large enough $J$, and  $N_{\min}/N\geq p_N$ for large enough $N$.
So we have
\begin{align}
\label{eq-qsingle}
\frac{\sum_{m=1}^{J_{\min} }\sum_{j\in\mathcal B^m} I(\hat\qq_j\nsim \qq_j^0)}{J}
    =&~
	\frac{\sum_{m=1}^{J_{\min}}  I\left(\hat \qq_{\mathcal B^m} \nsim \qq^0_{\mathcal B^m}\right)}{J}
	= 
	o_P\left(\frac{\gamma_{J}}{\beta_J}\cdot  \frac{1}{p_N}\right),\\
	\label{eq-rate-a}
	&~
	\frac{\sum_{s=1}^{N_{\min}} I(\hat\aaa_{\mathcal N^s}  \neq \aaa^0_{\mathcal N^s})}{N}= 
    o_P\left(\frac{\gamma_{J}}{\beta_J}\cdot\frac{1}{\delta_J}\right).
\end{align}

\bigskip
\noindent\textbf{Step 6.}
Since the previous \eqref{eq-rate-a} regards those subjects indices $i\in\mathcal N^1,\ldots,\mathcal N^{N_{\min}}$, which is a subset of $[N]$, 
we next further obtain a bound involving all the subject indices $i\in[N]$ using \eqref{eq-pf-qrate}. Denote the set of these remaining subject indices by  $\mathcal N^{\text{rest}} = [N]\setminus\left(\cup_{s=1}^{N_{\min}}\mathcal N^s\right)$. We have
\begin{align*}
	\sum_{i=1}^N \sum_{j=1}^J I\left(\xi(\qq^0_j,\aa^0_i) \neq \xi(\hat\qq_j,\hat\aa_i)\right) 
	\geq &~ \sum_{i\in \mathcal N^{\text{rest}} }\sum_{m=1}^{J_{\min}}\sum_{j\in\mathcal B^m}  I\left( \xi(\qq^0_j,\aa^0_i) \neq \xi(\hat\qq_j,\hat\aa_i)\right) \\
	\geq &~ \sum_{i\in \mathcal N^{\text{rest}} }\sum_{m=1}^{J_{\min}}\sum_{j\in\mathcal B^m} I(\qq^0_{\mathcal B^m} = \hat\qq_{\mathcal B^m} ) I(\xi(I_K,\aa^0_i)\neq \xi(I_K, \hat\aa_i))\\
	= &~ \sum_{i\in \mathcal N^{\text{rest}} }\sum_{m=1}^{J_{\min}}\sum_{j\in\mathcal B^m} I(\qq^0_{\mathcal B^m} = \hat\qq_{\mathcal B^m} ) I(\aa^0_i\nsim \hat\aa_i)\\
	= &~ \sum_{i\in \mathcal N^{\text{rest}} }I(\aa^0_i\nsim \hat\aa_i)\sum_{m=1}^{J_{\min}}\sum_{j\in\mathcal B^m} I(\qq^0_{\mathcal B^m} = \hat\qq_{\mathcal B^m} ),
\end{align*}
which implies
\begin{align}\notag
	\sum_{i\in \mathcal N^{\text{rest}} } I(\aa^0_i\nsim \hat\aa_i)
\leq &~ \frac{\sum_{i=1}^N \sum_{j=1}^J I\left(\xi(\qq^0_j,\aa^0_i) \neq \xi(\hat\qq_j,\hat\aa_i)\right)}{\sum_{m=1}^{J_{\min}}\sum_{j\in\mathcal B^m} I(\qq^0_{\mathcal B^m}= \hat\qq_{\mathcal B^m})},\\ \notag
\frac{1}{N} \sum_{i\in \mathcal N^{\text{rest}} }  I(\aa^0_i\nsim \hat\aa_i)
\leq &~ \frac{\frac{1}{NJ}\sum_{i=1}^N \sum_{j=1}^J I\left(\xi(\qq^0_j,\aa^0_i) \neq \xi(\hat\qq_j,\hat\aa_i)\right)}{1-\frac{1}{J}\sum_{m=1}^{J_{\min}}\sum_{j\in\mathcal B^m} I(\qq^0_{\mathcal B^m}\neq \hat\qq_{\mathcal B^m})} \\ 
\label{eq-arate-final}
=&~ \frac{o_P(\gamma_{J}/\beta_J)}{1-o_P(\gamma_{J}/(\beta_J\cdot p_N))} = o_P\left(\frac{\gamma_{J}}{\beta_J}\right).
\end{align}
Now summing up \eqref{eq-arate-final} and \eqref{eq-rate-a} gives 
\begin{align}\label{eq-arate-2p}
\frac{1}{N}\sum_{i=1}^N I(\aa^0_i\nsim \hat\aa_i)
=o_P\left(\frac{\gamma_{J}}{\beta_J \cdot  \delta_J}\right).	
\end{align}

\noindent\textbf{Step 7.} In this step we further establish 
\begin{equation*}
	\frac{1}{J}\sum_{1\leq j\leq J,\atop j\in \mathcal B^{J,\mult}} I(\hat\qq_j\nsim \qq_j^0) = o_P \left(\frac{\gamma_{J}}{\beta_J}\right).
\end{equation*}
The following inequalities hold,
\begin{align*}
	&~\sum_{i=1}^N \sum_{j=1}^J I\left( \xi(\qq^0_j,\aa^0_i) \neq \xi(\hat\qq_j,\hat\aa_i)\right) 
	\geq \sum_{i=1}^N \sum_{j\in\mathcal B^{J,\mult}} I\left( \xi(\qq^0_j,\aa^0_i) \neq \xi(\hat\qq_j,\hat\aa_i)\right) \\
	= &~ \sum_{j\in\mathcal B^{J,\mult}} \sum_{s=1}^{N_{\min}}\sum_{i\in \mathcal N^s} I\left( \xi(\qq^0_j,\aa^0_i) \neq \xi(\hat\qq_j,\hat\aa_i)\right)\\
	\geq &~ \sum_{j\in\mathcal B^{J,\mult}} \sum_{s=1}^{N_{\min}} I(\hat\aa_{\mathcal N^s}=\aa^0_{\mathcal N^s})
	    I\left( \xi(\qq^0_j,\{0,1\}^K) \neq \xi(\hat\qq_j,\hat\aa_{\mathcal N^s})\right) \\
	\geq &~ \sum_{j\in\mathcal B^{J,\mult}} \sum_{s=1}^{N_{\min}} I(\hat\aa_{\mathcal N^s}=\aa^0_{\mathcal N^s})
	    I( \qq_j^0 \nsim \hat\qq_j) \quad(\text{by \eqref{eq-lemqa-2} in Lemma \ref{lem-qa}})\\
	= &~ \sum_{s=1}^{N_{\min}} I(\hat\aa_{\mathcal N^s}=\aa^0_{\mathcal N^s}) \sum_{j\in\mathcal B^{J,\mult}} I( \qq_j^0 \nsim \hat\qq_j).
\end{align*}
So we further have
\begin{align*}
	\sum_{j\in\mathcal B^{J,\mult}} I( \qq_j^0 \nsim \hat\qq_j)
\leq &~ \frac{\sum_{i=1}^N \sum_{j=1}^J I\left( \xi(\qq^0_j,\aa^0_i) \neq \xi(\hat\qq_j,\hat\aa_i)\right)}{\sum_{s=1}^{N_{\min}} I(\hat\aa_{\mathcal N^s}=\aa^0_{\mathcal N^s}) },\\
	\frac{1}{J}\sum_{j\in\mathcal B^{J,\mult}} I( \qq_j^0 \nsim \hat\qq_j)
\leq &~ \frac{\frac{1}{NJ}\sum_{i=1}^N \sum_{j=1}^J I\left( \xi(\qq^0_j,\aa^0_i) \neq \xi(\hat\qq_j,\hat\aa_i)\right)}{1-\frac{1}{N}\sum_{s=1}^{N_{\min}} I(\hat\aa_{\mathcal N^s}\neq\aa^0_{\mathcal N^s}) }.
\end{align*}
The numerator of the above display is $o_P(\gamma_{J})$, and in the denominator, the previous Step 5 guarantees that $\frac{1}{N}\sum_{s=1}^{N_{\min}} I(\hat\aa_{\mathcal N^s}\neq\aa^0_{\mathcal N^s})=o_P(\gamma_{J})$, so we further obtain
\begin{equation}\label{eq-qmult}
	\frac{1}{J}\sum_{j\in\mathcal B^{J,\mult}} I( \qq_j^0 \nsim \hat\qq_j) \leq \frac{o_P(\gamma_{J}/\beta_J)}{1-o_P(\gamma_{J}/\delta_J)} = o_P \left(\frac{\gamma_{J}}{\beta_J}\right).
\end{equation}
Summing up \eqref{eq-qsingle} and \eqref{eq-qmult} gives
\begin{align}\label{eq-pf-qrate}
	\frac{1}{J}\sum_{j\in\mathcal J_0} I(\hat\qq_j^0 \nsim \hat\qq_j) = o_P\left(\frac{\gamma_{J}}{\beta_J\cdot p_N}\right).
\end{align}
The proof of the theorem is now complete.

\bigskip

\color{black}
\subsection{Proof of Theorem \ref{thm-joint-both} for Multi-Parameter Models}
\label{sec-pf2}

The proof is similar in spirit to that in Section \ref{sec-pf1} for two parameter-models and it can also be divided into Steps 1-7. 
We next will focus on discussing the differences. 

\noindent\textbf{Step 1.} This step is the same as Step 1 in the proof of Theorem \ref{thm-joint-both} and hence is omitted.

\noindent\textbf{Step 2.} When bounding the sum of KL-divergences under a multi-parameter SLAM, we need to introduce the following lemma.
\begin{lemma}
	\label{lem-kl}
	Under a multi-parameter SLAM,
	the following event happens with probability at least $1-\delta$,
\begin{equation*}
\sum_{j=1}^J\sum_{a=1}^{L_j} n_{j,a} D(\hat\theta_{j,a}  \| \bar \theta_{j,a}) < \epsilon = N\log (2^K) + J2^{K_0} \log\left(\frac{N}{2^{K_0}} + 1\right) - \log\delta.
\end{equation*}
\end{lemma}

\begin{proof}
Please see Page \pageref{pf-lem5}.
\end{proof}

Similar to the Step 2 in the proof of Theorem \ref{thm-joint-both}, we obtain
\begin{align*}
	\mathbb P (|X-\mathbb E[X]| \geq \epsilon)
\leq &~2\exp\left\{ -\frac{(1/2)\epsilon^2}{\sum_{i=1}^N \sum_{j=1}^N \ME[X_{i,j}^2]+ (2/3)\epsilon\log J } \right\}\\
\leq &~ 2\exp\left\{ -\frac{(1/2)\epsilon^2}{2d^2 MNJ(\log J)^2 + (2/3)\epsilon\log J} \right\}.
\end{align*}

\bigskip
\noindent\textbf{Step 3-4.}
In this step, we prove the following lemma.
\begin{lemma}
	\label{prop-scale}
	Under the following scaling for some small positive constant $c>0$,
\begin{equation}\label{eq-scale}
\sqrt{J}\cdot 2^K = o (N^{1-c}),
\end{equation}
we have $$\frac{1}{NJ}\max_{\ZZ} |\ell(\RR;\,\ZZ) - \ME \ell(\RR;\,\ZZ)| 
	= o_P\left(\frac{\sqrt{M\log(2^K)}}{\sqrt{J}} (\log J)^{1+\epsilon}\right).
	$$
\end{lemma}

\begin{proof}
Please see Page \pageref{pf-lem6}.
\end{proof}

Under the considered multi-parameter SLAM, with the scaling of $N,J, K$ in part (b) of Theorem \ref{thm-joint-both}, there is
\begin{align}\label{eq-lower}
	o_P(NJ\cdot \gamma_{J}) = &~ \bar \ell(\ZZ^0) - \bar \ell(\hat \ZZ) \\ \notag
    =&~\sum_{i=1}^N\sum_{j=1}^J \left[P_{i,j}\log\left(\frac{\theta^{\ZZ^0}_{j,z_i^0}}{\bar\theta^{\hat\ZZ}_{j,\hat z_i}}\right)
    + (1-P_{i,j})\log\left(\frac{1-\theta^{\ZZ^0}_{j,z_i^0}}{1-\bar\theta^{\hat\ZZ}_{j,\hat z_i}}\right) \right]\\ \notag
= &~ \sum_{i=1}^N\sum_{j=1}^J D(P_{i,j}\| \bar\theta^{\hat\ZZ}_{j,\hat z_i}) 
\geq
\sum_{i=1}^N\sum_{j=1}^J 
2\left(P_{i,j} -\bar\theta^{\hat\ZZ}_{j,\hat z_i}\right)^2.
\end{align}
A similar argument establishes the conclusion of part (a) of the Theorem \ref{thm-joint-both} for multi-parameter models.

\bigskip
\noindent\textbf{Steps 5-6.}
Note that when $K_j = \sum_{k=1} q^0_{j,k}= 1$, we have $ \sum_{k=1} \hat q_{j,k}= 1$, so the above constraints reduce to $N_{11}^j N_{00}^j  > N_{10}^j N_{01}^j$ in this case.
Indeed, under Assumption \ref{cond-id}, there exist many blocks of identity submatrix $I_K$'s in the matrix $\QQ^0$, and for each single-attribute $\qq$-vector, the all-effect model behaves exactly like a two-parameter model and there is
	$\bar \ell(\ZZ^0) - \bar \ell(\hat \ZZ) \geq 
	\sum_{i=1}^N \sum_{j\in\mathcal B^{J,\ee}} 2(P_{i,j} -\bar\theta^{\hat\ZZ}_{j,\hat z_i})^2$.
Therefore, for these items, we can just impose the natural constraint \eqref{eq-estid} and proceed as the Step 5 and Step 6 of Theorem \ref{thm-joint-both} to obtain the conclusion similar to the previous \eqref{eq-qsingle} and \eqref{eq-arate-2p}, 
\begin{align}\label{eq-multqa}
	\frac{1}{J}\sum_{1\leq j\leq J\atop j\in\mathcal B^{J,\ee}} I(\hat\qq_j\nsim \qq_j^0)
	=o_P\left(\frac{\gamma_{J}}{\beta_J\cdot p_N}\right),\quad
	\frac{1}{N}\sum_{i=1}^N I(\aa^0_i\nsim \hat\aa_i)
=o_P\left(\frac{\gamma_{J}}{\beta_J \cdot  \delta_J}\right).
\end{align}
Under an all-effect SLAM, for an item $j\in[J]$ with $\sum_{k=1}^K q^0_{j,k} = K_j$, there are $L_j = 2^{K_j}$ potentially distinct item parameters associated with it. 
Under $\hat\ZZ$, similar to those under a two-parameter SLAM, we impose the following natural constraints to prevent label swapping, 
\begin{align}
\label{eq-natcons2}
&\bar\theta_{j,a} < \bar\theta_{j,b}~~\text{if}~~\theta_{j,a} < \theta_{j,b}.
\end{align} 
Without loss of generality, assume the index $L_j = 2^{K_j}$ corresponds to the latent class to which the all-one latent attribute pattern $\aaa=\one_K$ belongs to, both under the true $(\QQ^0, \AA^0)$ and the estimated $(\hat\QQ, \hat\AA)$. 
Note that when $K_j = \sum_{k=1} q^0_{j,k} = 1$, we have $ \sum_{k=1} \hat q_{j,k}= 1$ by the definition of the joint MLE, so the above constraints reduce to $N_{11}^j N_{00}^j  > N_{10}^j N_{01}^j$ in this case.
Given $\hat\ZZ$-induced $\{\hat\phi_{i,j}\}$, define
\begin{align*}
	N^j_{a,b} = \sum_{i=1}^N I(\phi^0_{i,j} = a)I(\hat\phi_{i,j} = b),\quad a,~b\in\{1,\ldots,L_j\},
\end{align*}
then there is
\begin{align}\notag
	\bar\theta_{j,b}^{\hat\ZZ} 
	=&~
	\frac{\sum_{i=1}^N I(\hat\phi_{i,j}=b)P_{i,j}}{\sum_{i=1}^N I(\hat\phi_{i,j}=b)} \\
	=&~
	\frac{\sum_{i=1}^N I(\hat\phi_{i,j}=b)\cdot [\sum_{a=1}^{L_j} I(\phi_{i,j}=a)\cdot\theta^0_{j,a}]}{\sum_{i=1}^N I(\hat\phi_{i,j}=b)}
	=
	\frac{\sum_{a=1}^{L_j} N_{a,b}^j\theta_{j,a}^0}{\sum_{a=1}^{L_j} N_{a,b}^j}.
\end{align}
Now we have
%
\begin{align}\notag
\bar\ell(\ZZ^0) - \bar\ell(\hat\ZZ) 
=\sum_{i=1}^N\sum_{j=1}^J D(P_{i,j}\| \bar\theta^{\,\hat\ZZ}_{j,\hat z_i})  
\geq &~ \sum_{i=1}^N \sum_{j\not\in \mathcal B^{J,\ee}}
 I(\hat \aa_i=\aa_i^0)  
        D(P_{i,j}\| \bar\theta^{\,\hat\ZZ}_{j,\hat z_i}).
\end{align}
For each item $j$, introduce notation 
\begin{align}\label{eq-deffj}
  F_j = \sum_{i=1}^N  I(\hat \aa_i=\aa_i^0)  D(P_{i,j}\| \bar\theta^{\,\hat\ZZ}_{j,\hat z_i}),
\end{align}
then the previous inequality gives $\bar\ell(\ZZ^0) - \bar\ell(\hat\ZZ) \geq \sum_{j=1}^J F_j$.
We next show that for any item $j\not\in \mathcal B^{J,\ee}$,  if $\hat\qq_j \neq \qq_j^0$, then 
$F_j \geq \Omega(N)\cdot \beta_J\cdot(1-o_P(1))$.
To this end, we consider two different cases separately: case (a) $\hat\qq_j\preceq \qq_j^0$ and $\hat\qq_j\neq \qq_j^0$; or case (b)  $\hat\qq_j\npreceq \qq_j^0$.

\medskip
\noindent\textbf{Case (a)}, $\hat\qq_j\preceq \qq_j^0$. This case can be viewed as an ``underfitting'' case, because multiple original patterns are grouped into a larger and coarser group under $\hat\qq_j \preceq \qq_j^0$.
For each item $j$, define
\begin{align*}
    N_{\succeq\qq_j^0,\;\succeq\hat\qq_j} 
    =&~ 
     \sum_{i=1}^N I (\aa_i^0 \succeq \qq_j^0,\; \aa_i^0 \succeq \hat\qq_j)
    =
    \sum_{i=1}^N I (\aa_i^0 \succeq \qq_j^0),
    \\
    N_{\nsucceq\qq_j^0,\; \succeq\hat\qq_j} 
    =&~  \sum_{i=1}^N I (\aa_i^0 \nsucceq \qq_j^0,\; \aa_i^0 \succeq \hat\qq_j).
\end{align*}
%
Then
\begin{align}\notag
    F_j 
    \geq &~
     \sum_{i=1}^N I(\hat\aa_i = \aa_i^0)
     I(\hat\aa_i \succeq \qq_j^0) D\Big(\theta_{j,\one} \Big{\|} \frac{ N_{\succeq\qq_j^0,\;\succeq\hat\qq_j}  \theta_{j,\one} +  N_{\nsucceq\qq_j^0,\; \succeq\hat\qq_j} \widetilde{\theta}_{j,\aaa} }{ N_{\succeq\qq_j^0,\;\succeq\hat\qq_j}   +  N_{\nsucceq\qq_j^0,\; \succeq\hat\qq_j}} \Big)
     \\ \label{eq-fjrhs}
     \geq &~
     \underbrace{D\Big(\theta_{j,\one} \Big{\|} \frac{ N_{\succeq\qq_j^0,\;\succeq\hat\qq_j}  \theta_{j,\one} +  N_{\nsucceq\qq_j^0,\; \succeq\hat\qq_j} \widetilde{\theta}_{j,\aaa} }{ N_{\succeq\qq_j^0,\;\succeq\hat\qq_j}   +  N_{\nsucceq\qq_j^0,\; \succeq\hat\qq_j}} \Big)}_{\text{(I)}}
     \underbrace{\sum_{i=1}^N I(\hat\aa_i = \aa_i^0)
     I(\aa_i^0 \succeq \qq_j^0)}_{\text{(II)}}.
\end{align}
The $\widetilde{\theta}_{j,\aaa}$ in the above display is
\begin{align*}
    \widetilde{\theta}_{j,\aaa}
    = \frac{\sum_{i=1}^N I(\aa_i^0 \nsucceq\qq_j^0,\; \aa_i^0 \succeq\hat\qq_j) P_{ij}}{N_{\nsucceq\qq_j^0,\; \succeq\hat\qq_j}},
\end{align*}
where each $P_{ij}$ in the summation in the numerator of $\widetilde{\theta}_{j,\aaa}$ satisfies $P_{ij} < \theta_{j,\one}$ because $\aa_i^0 \nsucceq\qq_j^0$.
First consider the first factor (I) in the right hand side (RHS) of \eqref{eq-fjrhs},
\begin{align*}
\text{(I) in \eqref{eq-fjrhs}}
= &~
    D\Big(\theta_{j,\one} \Big{\|} \frac{ N_{\succeq\qq_j^0,\;\succeq\hat\qq_j}  \theta_{j,\one} +  N_{\nsucceq\qq_j^0,\; \succeq\hat\qq_j} \widetilde{\theta}_{j,\aaa} }{ N_{\succeq\qq_j^0,\;\succeq\hat\qq_j}  +  N_{\nsucceq\qq_j^0,\; \succeq\hat\qq_j}} \Big)
    \\
    \geq &~
    2 \left(\theta_{j,\one} - \frac{ N_{\succeq\qq_j^0,\;\succeq\hat\qq_j}  \theta_{j,\one} +  N_{\nsucceq\qq_j^0,\; \succeq\hat\qq_j} \widetilde{\theta}_{j,\aaa} }{ N_{\succeq\qq_j^0,\;\succeq\hat\qq_j}  +  N_{\nsucceq\qq_j^0,\; \succeq\hat\qq_j}} \right)^2
    \\
    \geq &~
    2 \frac{(N_{\nsucceq\qq_j^0,\; \succeq\hat\qq_j})^2 (\theta_{j,\one} - \widetilde{\theta}_{j,\aaa})^2}{\left(N_{\succeq\qq_j^0,\;\succeq\hat\qq_j}  +  N_{\nsucceq\qq_j^0,\; \succeq\hat\qq_j}\right)^2}
    \\
    \geq &~
    2\beta_J\cdot \left(\frac{N2^{-K_0}\cdot\epsilon}{N}\right)^2
    = 2^{1-2K_0}\epsilon^2\beta_J
\end{align*}
where the last inequality above holds because $\hat\qq_j \preceq \qq_j^0$ and hence $N_{\nsucceq\qq_j^0,\; \succeq\hat\qq_j} 
\geq 
2^{K-K_0} \cdot N\cdot p_N \geq N2^{-K_0}\cdot\epsilon$ under Assumption \ref{cond-id} and that $N_{\succeq\qq_j^0,\;\succeq\hat\qq_j}  +  N_{\nsucceq\qq_j^0,\; \succeq\hat\qq_j}\leq N$.
Next consider the second factor (II) in the RHS of \eqref{eq-fjrhs},
\begin{align*}
\text{(II) in \eqref{eq-fjrhs}}
=&~
    \sum_{i=1}^N I(\hat\aa_i = \aa_i^0) I(\aa_i^0 \succeq \qq_j^0)  \\
=&~ \sum_{i=1}^N  I(\aa_i^0 \succeq \qq_j^0) -\sum_{i=1}^N I(\aa_i^0 \succeq \qq_j^0,\; \hat\aa_i \neq \aa_i^0)
    \\
    \geq &~
    \sum_{i=1}^N  I(\aa_i^0 \succeq \qq_j^0) -\sum_{i=1}^N I( \hat\aa_i \neq \aa_i^0)
    \\
    \geq &~
    \frac{N\cdot p_N}{2^{K-K_0}} - o_P\left( \frac{N\cdot\gamma_J}{\beta_J\cdot \delta_J} \right)
    \\
    \geq &~ 
    N \epsilon 2^{K_0} - o_P\left( \frac{N\cdot\gamma_J}{\beta_J\cdot \delta_J} \right) \quad(\text{due to Assumption }\ref{cond-id}\text{ on }p_N).
\end{align*}
Therefore $F_j$ in \eqref{eq-fjrhs} can be lower bounded as follows,
\begin{align*}
    F_j
    \geq  &~
    2^{1-2K_0}\epsilon^2\beta_J
    \cdot 
    \left[N \epsilon 2^{K_0} - o_P\left( \frac{N\cdot\gamma_J}{\beta_J\cdot \delta_J} \right)\right]\\
    =&~
    N\beta_J\left[2^{1-K_0}\epsilon^3
    -
    o_P\left( \frac{\gamma_J}{\beta_J\cdot \delta_J} \right)\right]
    =\Omega(N)\cdot\beta_J\cdot(1-o_P(1)).
\end{align*}

\medskip
\noindent\textbf{Case (b)}, $\hat\qq_j\npreceq \qq_j^0$.
Recalling the constraint that $\sum_{k=1}^K \hat q_{j,k}\leq \sum_{k=1}^K q^0_{j,k}$ in the definition of the joint MLE, in this case there must also be $\hat\qq_j\nsucceq \qq_j^0$.
We call this case the ``misfitting'' scenario as the fitted $\hat q_{j,k}$ measures some additional attributes not measured by $\qq_j^0$.
%
Note that since $\hat\qq_j\npreceq \qq_j^0$ and $\hat\qq_j\nsucceq \qq_j^0$, the vector $\tilde\qq_{j}$ contains strictly fewer entries of ``1'' than either $\hat\qq_{j}$ or $\qq_j^0$.
In the definition \eqref{eq-deffj} of $F_j$, the summation involves all the subjects indexed from $1$ to $N$. We next also only consider those subjects $i$ satisfying $\aa_i^0 \succeq \qq_j^0$.

Define $K_{j}^{\text{dif}} = \sum_{k=1}^K I(\hat q_{j,k}=1, q_{j,k}^0=0)$. 
Without loss of generality, suppose 
\begin{align}\label{eq-qq}
    \qq_j^0 &= 
    (\underbrace{1,\ldots,1,}_{K_j^0\text{ entries}}~
    \underbrace{1,\ldots,1,}_{K_j^{\text{com}}\text{ entries}}~
    \underbrace{0,\ldots,0,}_{K_j^{\text{dif}}\text{ entries}}~
    \underbrace{0,\ldots,0}_{K_j^{\text{rem}}\text{ entries}})\\ \notag
    \hat\qq_j &=  (\underbrace{0,\ldots,0,}_{K_j^0\text{ entries}}~
    \underbrace{1,\ldots,1,}_{K_j^{\text{com}}\text{ entries}}~
    \underbrace{1,\ldots,1,}_{K_j^{\text{dif}}\text{ entries}}~
    \underbrace{0,\ldots,0}_{K_j^{\text{rem}}\text{ entries}});
\end{align}
That is, $K_j^0$ is the number of attributes measured by $\qq_j^0$ but not by $\hat\qq_j$, $K_j^{\text{com}}$ is the number of attributes commonly measured by both $\qq_j^0$ and $\hat\qq_j$, $K_j^{\text{dif}}$ is the number of attributes measured by $\hat\qq_j$ but not by $\qq_j^0$, and $K_j^{\text{rem}}$ is the number of remaining entries not measured by either $\hat\qq_j$ or $\qq_j^0$.
In the current case (b) with $\hat\qq_j\npreceq \qq_j^0$ and $\hat\qq_j\nsucceq \qq_j^0$, there must be $K_j^{\text{dif}}>0$ and $K_j^0 > 0$. Also define $K_{j}^{\max} = K_j^0 + K_j^{\text{com}} + K_j^{\text{dif}}$, and by the definition of the joint MLE there is $K_{j}^{\max}\leq 2K_0$, where $K_0$ is the upper bound of the number of attributes measured by either $\qq_j^0$ or $\hat\qq_j$.
Based on the structures of $\qq_j^0$ and $\hat\qq_j$, we define $\ell_j := 2^{K_j^{\text{dif}}}$ attribute patterns $ \aaa_{j,1},\aaa_{j,2},\ldots,\aaa_{j,\ell_j}$ which coincide in the first $K_j^0$ attributes and the last $K_j^{\text{dif}} + K_j^{\text{rem}}$ attributes as follows,
\begin{align}\label{eq-ajm}
    \aaa_{j,m} = (\underbrace{1,\ldots,1,}_{K_j^0\text{ entries}}~
    \underbrace{1,\ldots,1,}_{K_j^{\text{com}}\text{ entries}}~
    \underbrace{*,\ldots,*,}_{K_j^{\text{dif}}\text{ entries}}~
    \underbrace{0,\ldots,0}_{K_j^{\text{rem}}\text{ entries}}),
    \quad \forall m\in\{1,\ldots,\ell_j=2^{K_j^{\text{dif}}}\}.
\end{align}
For the middle $K_j^{\text{dif}}$ entries, the $ \aaa_{j,1},\aaa_{j,2},\ldots,\aaa_{j,\ell_j}$ range over all the $\ell_j$ possible binary vector configurations; for example, $\aaa_{j,1}$ has the middle $K_j^{\text{dif}}$ entries being all zeros, and the $\aaa_{j,\ell_j}$ has the middle $K_j^{\text{dif}}$ entries being all ones, etc.
Then
\begin{align*}
    F_j 
    \geq &~
    \sum_{i=1}^N I(\hat\aa_i = \aa_i^0) 
    I(\hat\aa_i \succeq \qq_j^0)
    D( \theta_{j,\one} {\|} \bar\theta_{j,\hat z_i}^{\hat\ZZ} )
    \\
    = &~
    \sum_{m=1}^{\ell_j}
    \sum_{i=1}^N 
    I(\hat\aa_i = \aa_i^0) 
    I(\hat\aa_i \succeq \qq_j^0)
    I(\aa_i^0\text{ equals }\aaa_{j,m}\text{ in the first }K_j^{\max}\text{ entries})D(\theta_{j,\one} {\|} \bar\theta_{j,\hat z_i}^{\hat\ZZ}),
\end{align*}
where the last equality above holds because of the definitions in \eqref{eq-qq} and \eqref{eq-ajm}.
Next, note that for any $\aa_i^0\text{ which equals }\aaa_{j,m}\text{ in the first }K_j^{\max}$ entries, 
there must be 
\begin{align}\label{eq-ajm2}
\bar\theta_{j,\hat z_i}^{\hat\ZZ} 
=&~ \bar\theta_{j,\aaa_{j,m}}^{\hat\ZZ}
    \\ \notag
=&~
\frac{\sum_{i} \theta_{j,\one} I(\aa_i^0\text{ equals }\aaa_{j,m}\text{ in all entries }1,\ldots,K_j^0+K_j^{\text{com}})}
{\sum_{i} I(\aa_i^0\text{ equals }\aaa_{j,m}\text{ in all entries }K_j^0+1,\ldots,K_j^{\text{max}})}
\\ \notag
&~+
\frac{\sum_{i} \tilde\theta_{j,\aaa} I(\aa_i^0\text{ doesn't equal }\aaa_{j,m}\text{ in some entries in }1,\ldots,K_j^0+K_j^{\text{com}})}
{\sum_{i} I(\aa_i^0\text{ equals }\aaa_{j,m}\text{ in all entries }K_j^0+1,\ldots,K_j^{\text{max}})},
\end{align}
where $\tilde\theta_{j,\aaa}$ is defined as
\begin{align*}
&~ \tilde\theta_{j,\aaa}
\sum_i I(\aa_i^0\text{ doesn't equal }\aaa_{j,m}\text{ in some entries in }1,\ldots,K_j^0+K_j^{\text{com}})
\\
=&~ \sum_{i} P_{ij} I(\aa_i^0\text{ doesn't equal }\aaa_{j,m}\text{ in some entries in }1,\ldots,K_j^0+K_j^{\text{com}}).
\end{align*}
Therefore for each $m\in\{1,\ldots,\ell_j\}$ there is
\begin{align*}
&~ D\Big( \theta_{j,\one} \Big{\|} \bar\theta_{j,\aaa_{j,m}}^{\hat\ZZ} \Big)
\\
\geq &~
\left(
\frac{\sum_i I(\aa_i^0\text{ doesn't equal }\aaa_{j,m}\text{ in some entries in }1,\ldots,K_j^0+K_j^{\text{com}}) (\theta_{j,\one} - \tilde\theta_{j,\aaa})}
{\sum_{i} I(\aa_i^0\text{ equals }\aaa_{j,m}\text{ in all entries }K_j^0+1,\ldots,K_j^{\text{max}})}
\right)^2
\\
\geq &~
\left( 
\frac{N\cdot p_N\cdot (2^{K_j^0 + K_j^{\text{com}}} - 1) 2^{ K - K_j^0 - K_j^{\text{com}}} }{N}
\right)^2 \cdot (\theta_{j,\one} - \tilde\theta_{j,\aaa})^2
\\
\geq &~
\epsilon^2 \cdot (2^{K_j^0 + K_j^{\text{com}}} - 1)^2
\cdot 2^{- 2K_j^0 - 2K_j^{\text{com}}}
\cdot \beta_J
\end{align*}
Now insert the above lower bound back into the lower bound for $F_j$ and we have
\begin{align*}
    F_j 
    \geq &~
    \sum_{m=1}^{\ell_j}
    \sum_{i=1}^N 
    I(\hat\aa_i = \aa_i^0) 
    I(\hat\aa_i \succeq \qq_j^0)
    I(\aa_i^0\text{ equals }\aaa_{j,m}\text{ in the first }K_j^{\max}\text{ entries})
    \\
    &~ \qquad\qquad \times
    \epsilon^2 \cdot (2^{K_j^0 + K_j^{\text{com}}} - 1)^2
\cdot 2^{- 2K_j^0 - 2K_j^{\text{com}}}
\cdot \beta_J
\\
= &~ \epsilon^2 \cdot (2^{K_j^0 + K_j^{\text{com}}} - 1)^2
\cdot 2^{- 2K_j^0 - 2K_j^{\text{com}}}
\cdot \beta_J
\cdot \sum_{i=1}^N I(\hat\aa_i = \aa_i^0) 
    I(\hat\aa_i \succeq \qq_j^0)
\\
\geq &~ \epsilon^2 \cdot (2^{K_j^0 + K_j^{\text{com}}} - 1)^2
\cdot 2^{- 2K_j^0 - 2K_j^{\text{com}}}
\cdot \beta_J
\cdot 
\left[
\sum_{i=1}^N
    I(\hat\aa_i \succeq \qq_j^0) 
    -
    \sum_{i=1}^N
    I(\hat\aa_i \neq \aa_i^0) 
\right]
\\
\geq &~ \Omega(N)\cdot\beta_J\cdot(1-o_P(1)),
\end{align*}
where the last inequality follows from a similar argument in the last step of the proof of the previous case (a).
    
\medskip
Now summarizing case (a) and case (b), we have that $F_j \geq I(\hat\qq_j \neq \qq_j^0)\cdot\Omega(N)\cdot\beta_J\cdot(1-o_P(1))$, therefore
\begin{align*}
o_P(NJ\cdot \gamma_{J}) 
=
    \sum_{j=1}^J F_j 
    \geq &~ 
    \sum_{1\leq j\leq J\atop j\not\in\mathcal B^{J,\ee}} I(\hat\qq_j \neq \qq_j^0) \cdot \Omega(N) \cdot\beta_J\cdot(1-o_P(1))
    \\
    = &~ \frac{1}{J}\sum_{1\leq j\leq J\atop j\not\in\mathcal B^{J,\ee}} I(\hat\qq_j \neq \qq_j^0) \cdot \Omega(NJ)\cdot\beta_J\cdot(1-o_P(1)).
\end{align*}
This implies that
\begin{align*}
    \frac{1}{J}\sum_{1\leq j\leq J\atop j\not\in\mathcal B^{J,\ee}} I(\hat\qq_j \neq \qq_j^0) 
    = o_P\left( 
    \frac{\gamma_J}{\beta_J}
    \right)
\end{align*}
Combining \eqref{eq-multqa} and the above gives
\begin{align}\label{eq-pf-qrate}
	\frac{1}{J}\sum_{j=1}^J I(\hat\qq_j^0 \nsim \hat\qq_j) = o_P\left(\frac{\gamma_{J}}{\beta_J\cdot p_N}\right).
\end{align}
This completes the proof of the theorem.


\color{black}

\bigskip

\subsection{{A Toy Example Illustrating Assumption \ref{as-weak}}} 
\label{sec-pftoy}
	Consider a single item $j$ with $\qq_j^0=(1,1,0)$ and rows of $\AA^0$ being $\aa^0_1=(0,0,0),\,\aa^0_2=(0,0,1), \aa^0_3=(0,1,0),\,\aa^0_4=(0,1,1),\, \aa^0_5=(1,0,0),\,\aa^0_6=(1,0,1),\, \aa^0_7=(1,1,0),\,\aa^0_8=(1,1,1)$.
	We claim that $f_j(\ZZ^0) = \min_{\ZZ=(\qq_j,\AA)}
		f_j(\ZZ)$ holds for this item $j\in\mathcal E_0$. 
	Under the multi-parameter SLAM, for $\qq_j^0$ there are four item parameters for this item $j$: $\theta_{j,\,(000)}=\theta_{j,\,(001)}$, $\theta_{j,\,(010)}=\theta_{j,\,(011)}$, $\theta_{j,\,(100)}=\theta_{j,\,(101)}$, and $\theta_{j,\,(110)}=\theta_{j,\,(111)}$; we denote these four parameters by $\theta_{j,\,(00*)}$, $\theta_{j,\,(01*)}$, $\theta_{j,\,(10*)}$, and $\theta_{j,\,(11*)}$, respectively.
Define
\begin{align*}
& L_j=\min_{ab\in\{00,\,01,\,10\}} (\theta_{j,\,(11*)} - \theta_{j,\,(a b*)})^2;\\
& U_j
		=\max_{ab,cd\in\{00,\,01,\,10\}\atop ab\neq cd} D(\theta_{j,\,(ab*)}\| \theta_{j,\,(cd*)}).
\end{align*}
	First, the $f_j(\ZZ^0)$ in the right hand side of \eqref{eq-less} can be upper bounded as follows,
	\begin{align}\label{eq-lessright}
		&~f_j(\ZZ^0)=\sum_{i=1}^N D\Big(P^{\mult}_{i,j} \Big\| P^{2,\ZZ^0}_{i,j} \Big) \\ \notag
		=&~ \sum_{ab\in\{00,\,01,\,10\}} 2D\Big(\theta_{j,\,(ab*)}  \Big\| \frac{\theta_{j,\,(00*)} + \theta_{j,\,(01*)} + \theta_{j,\,(01*)}}{3}\Big)\\ \notag
		\leq &~ \sum_{ab\in\{00,\,01,\,10\}} 2\frac{\sum_{cd\in\{00,\,01,\,10\}} D(\theta_{j,\,(ab*)}\| \theta_{j,\,(cd*)})}{3}\\ \notag
		\leq &~ \frac{2}{3}\sum_{ab,cd\in\{00,\,01,\,10\}\atop ab\neq cd}D(\theta_{j,\,(ab*)}\| \theta_{j,\,(cd*)})\\ \notag
		\leq &~ 4\max_{ab,cd\in\{00,\,01,\,10\}\atop ab\neq cd} D(\theta_{j,\,(ab*)}\| \theta_{j,\,(cd*)})
		=4 U_j,
	\end{align} 
	where the last but third inequality is due to the convexity of the KL divergence with respect to its second argument.
	Next we consider $f_j(\ZZ)$ regarding an arbitrary $\ZZ$ in the left hand side of \eqref{eq-less}. There is 
	\begin{align*}
		f_j(\ZZ)
		=&~ \sum_{i=1}^N D\Big(P^{\mult}_{i,j} \Big\| P^{2,\ZZ}_{i,j} \Big) 
		\geq %
		\sum_{i: \,\aa_i\succeq\qq_j^0} D\Big(\theta_{j,\,(11*)} \Big\| P^{2,\ZZ}_{i,j} \Big).
	\end{align*}
	Note that $\ZZ$ under $\sum_{k=1}^K q_{j,k}\leq \sum_{k=1}^K q^0_{j,k}=2$ induces a partition of the $N=8$ subjects into at most $2^2=4$ latent classes. If $\ZZ\neq\ZZ^0$, then the partition induced by $\ZZ$ is different from those under $\ZZ^0$. 
		Consider two possible cases, (1)  subjects $i=7,8$ belong to the same latent class under $\ZZ$, (2) subjects $i=7,8$ belong to two different latent classes under $\ZZ$. In case (1), denote the number of other subjects falling in the same cluster of $i=7,8$ by $m$. Since we assume $\ZZ\neq\ZZ^0$, we must have $m\geq 1$ and 
	\begin{align*}
	  f_j(\ZZ,\,\text{case (1)})
	\geq &~ 2 D\Big(\theta_{j,\,(11*)} \Big \| \frac{2\theta_{j,\,(11*)} + \sum_{\ell=1}^m\theta_{j,\,(a_\ell b_\ell*)}}{2+m}\Big)\\
	\geq &~ 
	\frac{4m^2}{(2+m)^2} \min_{ab\in\{00,\,01,\,10\}} (\theta_{j,\,(11*)} - \theta_{j,\,(a b*)})^2
	=: \frac{4m^2}{(2+m)^2} L_j 
	\geq \frac{4}{9} L_j.
	\end{align*}
	In case (2), the two subjects $i=7,8$ belong to two different latent classes under $\ZZ$, and we denote by $m_1, m_2$ the number of other subjects assigned to these two different clusters, respectively. Since a two-parameter approximation under $\ZZ$ only contain two latent classes, there must be $m_1\geq 1$ or $m_2\geq 1$, then similar to case (1) there is
	\begin{align*}
	f_j(\ZZ,\,\text{case (2)})
	\geq &~ \left[\frac{2m_1^2}{(1+m_1)^2} + \frac{2m_2^2}{(1+m_2)^2} \right] \min_{ab\in\{00,\,01,\,10\}} (\theta_{j,\,(11*)} - \theta_{j,\,(a b*)})^2
	\geq L_j.
	\end{align*}
	Combining cases (1) and (2), we obtain that 
	\begin{align*}
		\min_{\ZZ:\, \ZZ\neq\ZZ^0} f_j(\ZZ) 
		\geq \frac{4}{9} L_j,\quad & L_j:=\min_{ab\in\{00,\,01,\,10\}} (\theta_{j,\,(11*)} - \theta_{j,\,(a b*)})^2;\\
		 f_j(\ZZ_0) 
		\leq 4 U_j,\quad & U_j:
		=\max_{ab,cd\in\{00,\,01,\,10\}\atop ab\neq cd} D(\theta_{j,\,(ab*)}\| \theta_{j,\,(cd*)}).
	\end{align*}
Therefore in order to have $\min_{\ZZ} f_j(\ZZ) \geq f_j(\ZZ^0)$, it suffices to have $L_j\geq 9 U_j$ for this item $j$. 
In summary, by working out this toy example, we shed light on the intuition behind \eqref{eq-less} in Assumption \ref{as-weak}. That is, $L_j/U_j \geq C $ for some constant $C$ for all $j$ would intuitively lead to \eqref{eq-less}. It is worth noting that the two-parameter DINA model has $L_j>0$ and $U_j= 0$ and hence $L_j/U_j=\infty$. Our derivation here shows that a multi-parameter model with $U_j\neq 0$ can have a behavior that the ``oracle'' two-parameter approximation is the best among all the possible two-parameter approximations.
\color{black}

\subsection{Proof of Theorem \ref{thm-mis-weak} and Theorem \ref{thm-mis-strong}}
\label{sec-pf34}
We combine the proofs of Theorem \ref{thm-mis-weak} and Theorem \ref{thm-mis-strong} here because they share a same first step in analyzing the misspecificed log likelihood. After such a step 1, we will go on to separately discuss the different scenarios in the two theorems in Case (1) and Case (2), respectively.

\bigskip
\noindent\textbf{Step 1.} 
Recall the true probability of observing $r_{i,j}=1$ under the true data-generating multi-parameter model (e.g., GDINA) by $P_{i,j}^{\true}$.
	Denote the log-likelihood under a two-parameter model (e.g., DINA) by $\ell^{2\approx}(\RR\mid \QQ,\AA)$, and denote its expectation with respect to the distribution of the true data generating mechanism by $\ME_{\true} [\ell^{2\approx}(\RR\mid \QQ,\AA)]$.  Given $\ZZ=(\QQ,\AA)$, define $\xi_{i,j} = \prod_{k=1}^K a_{i,k}^{q_{j,k}}$ and
	\begin{equation}\label{eq-hatbar}
		\hat\theta_{j,a} = \frac{\sum_{i=1}^N I(\xi_{i,j}=a) r_{i,j}}{\sum_{i=1}^N I(\xi_{i,j}=a) },
		\quad
		\bar\theta_{j,a} = \frac{\sum_{i=1}^N I(\xi_{i,j}=a) P_{i,j}^{\true}}{\sum_{i=1}^N I(\xi_{i,j}=a) },~~a=0,1.
	\end{equation}
Also define $n_{j,a} = \sum_{i=1}^N I(\xi_{i,j}=a)$.

\begin{lemma}\label{prop-expr-mis}
The following display holds 
\begin{align*}
	&~\ell^{2\approx}(\RR\mid \QQ,\AA) - \ME_{\true} [\ell^{2\approx}(\RR\mid \QQ,\AA)]\\
   =&~\sum_{j=1}^J\sum_{a=0,1} n_{j,a} D(\hat\theta_{j,a}\|\bar\theta_{j,a}) + \sum_{i=1}^N \sum_{j=1}^J (r_{i,j} - P_{i,j}^{\true})\log\left( \frac{\bar\theta_{j,\xi_{i,j}}}{1-\bar\theta_{j,\xi_{i,j}}} \right).
\end{align*}
Furthermore, 
\begin{align*}
	\frac{1}{NJ}\max_{\ZZ} |\ell^{2\approx}(\RR;\, \ZZ) - \ME_{\true} [\ell^{2\approx}(\RR;\, \ZZ)]| 
	=o_P\left(\gamma_{J}\right).
\end{align*}
\end{lemma}

\begin{proof}
Please see Page \pageref{pf-lem8}.
\end{proof}

We continue with the proof of the theorem. Denote the true latent structure that generates the data $\RR$ by $\ZZ^0$ and the estimator obtained from maximizing the misspecified likelihood \eqref{eq-prob1} by $\hat\ZZ^{2\approx}$. Consider the following difference of expected log-likelihoods,
\begin{align}\label{eq-diff-eg}
	&~\ME_{\true} [\ell^{2\approx}(\RR;\, \ZZ^0)] - \ME_{\true} [\ell^{2\approx}(\RR;\, \hat\ZZ^{2\approx})] \\ \notag
=   &~\ME_{\true} [\ell^{2\approx}(\RR;\, \ZZ^0)] - \ell^{2\approx}(\RR;\, \ZZ^0)
     +\ell^{2\approx}(\RR;\, \hat\ZZ^{2\approx}) - \ME_{\true} [\ell^{2\approx}(\RR;\, \hat\ZZ^{2\approx})]
     \\ \notag
     &~
     +\underbrace{\ell^{2\approx}(\RR;\, \ZZ^0) - \ell^{2\approx}(\RR;\, \hat\ZZ^{2\approx})}_{\leq 0}
      \\ \notag
\leq &~ 2 \max_{\ZZ} |\ell^{2\approx}(\RR;\, \ZZ) - \ME_{\true} [\ell^{2\approx}(\RR;\, \ZZ)]|
=o_P(NJ\cdot\gamma_{J}),
\end{align}
where the last but second inequality $\ell^{2\approx}(\RR;\, \ZZ^0) - \ell^{2\approx}(\RR;\, \hat\ZZ^{2\approx})\leq 0$ follows from the definition that $\hat\ZZ^{2\approx}$ maximizes $\ell^{2\approx}(\RR;\, \ZZ)$.
Recall the definition 
\begin{equation}\label{eq-sing}
	\mathcal E_0=\{j\in[J]:\, \qq^0_j = \ee_k~~\text{for some}~~k\in[K]\}.
\end{equation} 
From now on, we slightly abuse  the notation and denote by $\xi_{i,j} = \xi_{i,j}(\hat\ZZ^{2\approx})$ the ideal response structure under the misspecified MLE $\hat\ZZ^{2\approx} = (\hat\QQ^{2\approx},\hat\AA^{2\approx})$.
We next decompose $\ME_{\true} [\ell^{2\approx}(\RR;\, \ZZ^0)] - \ME_{\true} [\ell^{2\approx}(\RR;\, \hat\ZZ^{2\approx})]$ in \eqref{eq-diff-eg} into two parts for $j\in\mathcal E_0$ and $j\in[J]\setminus\mathcal E_0$,
\begin{align*}
	&~\text{Eq.}~\eqref{eq-diff-eg} \\
	= &~\sum_{j\in\mathcal E_0}\sum_{i=1}^N \Big\{ P_{i,j}^{\true}\log(P_{i,j}^{\true}) + (1-P_{i,j}^{\true})\log(1-P_{i,j}^{\true}) 
	\\
	&~\qquad 
	- \Big[ P_{i,j}^{\true}\log(\bar\theta^{2\approx}_{j,\,\hat\xi_{i,j}}) + (1-P_{i,j}^{\true})\log(1-\bar\theta^{2\approx}_{j,\,\hat\xi_{i,j}}) \Big] \Big\}\\
	+&~\sum_{j\notin\mathcal E_0}\sum_{i=1}^N \Big\{ P_{i,j}^{\true}\log(\bar\theta^{2\approx}_{j,\,\xi^0_{i,j}}) + (1-P_{i,j}^{\true})\log(1-\bar\theta^{2\approx}_{j,\,\xi^0_{i,j}}) 
	\\
	&~ \qquad 
	- \Big[ P_{i,j}^{\true}\log(\bar\theta^{2\approx}_{j,\,\hat\xi_{i,j}}) + (1-P_{i,j}^{\true})\log(1-\bar\theta^{2\approx}_{j,\,\hat\xi_{i,j}}) \Big] \Big\} \\
	=&~\sum_{j\in\mathcal E_0}\sum_{i=1}^N D(P_{i,j}^{\true} \| \bar\theta^{2\approx}_{j,\,\hat\xi_{i,j}})
	+ \sum_{j\notin\mathcal E_0}\sum_{i=1}^N \Big[ {-D(P_{i,j}^{\true} \| \bar\theta^{2\approx}_{j,\,\xi^0_{i,j}})}
	   +D(P_{i,j}^{\true} \| \bar\theta^{2\approx}_{j,\,\hat\xi_{i,j}}) \Big]\\
	=&~ \sum_{i=1}^N\sum_{j=1}^J D(P_{i,j}^{\true} \| \bar\theta^{2\approx}_{j,\,\hat\xi_{i,j}}) - 
	   {
	   \underbrace{\sum_{j\notin\mathcal E_0}\sum_{i=1}^N D(P_{i,j}^{\true} \| \bar\theta^{2\approx}_{j,\,\xi^0_{i,j}})}_{\text{(NE)}},
	    }
\end{align*}
where there is
\begin{align}\label{eq-def-ne}
{
\text{(NE)}=\sum_{j\notin\mathcal E_0}\sum_{i=1}^N D(P_{i,j}^{\true} \| \bar\theta^{2\approx}_{j,\,\xi^0_{i,j}})
}
= &~\sum_{j\notin\mathcal E_0}\sum_{i=1}^N I(\aa^0_i\succeq\qq^0_j) D(\theta^{\true}_{j,\qq^0_j}\|\theta^{\true}_{j,\qq^0_j}) \\
\notag
&~+  \sum_{j\notin\mathcal E_0}\sum_{i=1}^N I(\aa^0_i\nsucceq\qq^0_j) D\Big(\theta^{\true}_{j,\aa_i^0}\Big\|\frac{\sum_{m=1}^N I(\aa^0_m\nsucceq\qq^0_j)\theta^{\true}_{j,\aa_i}}{\sum_{m=1}^N I(\aa^0_m\nsucceq\qq^0_j)}\Big) \\
\notag
=&~ \sum_{j\notin\mathcal E_0}\sum_{b=1}^{L_j-1}|S_{j,\aaa_b}| D\Big(\theta^{\true}_{j,\aaa_b} \Big\| \frac{\sum_{c=1}^{L_j-1}|S_{j,\aaa_c}|\theta^{\true}_{j,\aaa_c}}{\sum_{c=1}^{L_j-1}|S_{j,\aaa_c}|} \Big)
\end{align}
Note Eq.~\eqref{eq-diff-eg}$=o_P(NJ\cdot\gamma_{J})=\sum_{i=1}^N\sum_{j=1}^J D(P_{i,j}^{\true} \| \bar\theta^{2\approx}_{j,\,\hat\xi_{i,j}}) - \text{(NE)}$. 

\bigskip
\noindent\textbf{Step 2.} In this step we separately consider the different scenarios of Theorem \ref{thm-mis-weak} and Theorem \ref{thm-mis-strong}, respectively.

\noindent\textbf{Case (1).}
In this part we prove Theorem \ref{thm-mis-weak}.
Under Assumption \ref{as-weak}, there is
\begin{align}\label{eq-case2}
	\sum_{j\notin\mathcal E_0}\sum_{i=1}^N \Big[ {D(P_{i,j}^{\true} \| \bar\theta^{2\approx}_{j,\,\xi^0_{i,j}})}
	   -D(P_{i,j}^{\true} \| \bar\theta^{2\approx}_{j,\,\hat\xi_{i,j}}) \Big] = 
	   o(NJ \cdot \eta_J).
\end{align} 
Then we have that
\begin{align*}
\text{Eq.~}\eqref{eq-diff-eg}
&=\sum_{i=1}^N\sum_{j=1}^J D(P_{i,j}^{\true} \| \bar\theta^{2\approx}_{j,\,\hat\xi_{i,j}}) - 
	   \sum_{j\notin\mathcal E_0}\sum_{i=1}^N D(P_{i,j}^{\true} \| \bar\theta^{2\approx}_{j,\,\xi^0_{i,j}})
	    \\
&=\sum_{i=1}^N\sum_{j\in\mathcal E_0}^J D(P_{i,j}^{\true} \| \bar\theta^{2\approx}_{j,\,\hat\xi_{i,j}}) - 
	   \sum_{j\notin\mathcal E_0}\sum_{i=1}^N 
	   \left[ D(P_{i,j}^{\true} \| \bar\theta^{2\approx}_{j,\,\xi^0_{i,j}}) - D(P_{i,j}^{\true} \| \bar\theta^{2\approx}_{j,\,\hat\xi_{i,j}})\right],
\end{align*}
therefore
\begin{align*}
\sum_{i=1}^N\sum_{j\in\mathcal E_0}^J D(P_{i,j}^{\true} \| \bar\theta^{2\approx}_{j,\,\hat\xi_{i,j}}) 
&= \text{Eq.~}\eqref{eq-diff-eg} + \sum_{j\notin\mathcal E_0}\sum_{i=1}^N 
	   \left[ D(P_{i,j}^{\true} \| \bar\theta^{2\approx}_{j,\,\xi^0_{i,j}}) - D(P_{i,j}^{\true} \| \bar\theta^{2\approx}_{j,\,\hat\xi_{i,j}})\right]\\
&= o_P(NJ\cdot\gamma_{J}) + o_P(NJ\cdot\eta_J)
=o_P(NJ\cdot(\gamma_J\vee\eta_J)).
\end{align*}
Note that for any $j\in\mathcal E_0$, the $\qq_j=\ee_k$ for some $k\in[K]$ and the multi-parameter model  reduces to a two-parameter model for this $j$. 
Therefore the above display can be equivalently rewritten as a bound for $\sum_{i=1}^N\sum_{j\in\mathcal E_0}^J D(P^{2\approx}_{i,j} \| \bar\theta^{2\approx}_{j,\,\hat\xi_{i,j}})$, that is, a bound under the two-parameter model.
Then under Assumption \ref{as-weak} with $\min_{j\in\mathcal E_0}\left(\theta^0_{j,\one_K} - \theta^0_{j,\zero_K}\right)^2 \geq \zeta_J$, following a similar argument as Steps 5-6 in the proof of Theorem \ref{thm-joint-both}, we obtain
$$
	\frac{1}{J}\sum_{j\in\mathcal E_0} I(\hat\qq^{2\approx}_j \neq \qq_j^0) 
	= o_P\left(\frac{\gamma_J \vee \eta_J}{\zeta_J\cdot p_N}\right),
	\qquad
	\frac{1}{N}\sum_{i=1}^N I(\aa_i^0 \neq \hat\aa^{2\approx}_i)
    = o_P\left(\frac{\gamma_J \vee \eta_J}{\zeta_J\cdot\delta_J}\right),
$$
holds up to a permutation of the $K$ attributes.
Here $p_N$ and $\delta_J$ are those specified in Assumption \ref{cond-id}.
This proves the conclusion  of Theorem \ref{thm-mis-weak}.

\bigskip
\noindent\textbf{Case (2).}
In this part we prove Theorem \ref{thm-mis-strong}.
Under Assumption \ref{as-strong}, there is
$$
\text{(NE)}=
\sum_{j\notin\mathcal E_0}\sum_{i=1}^N {D\Big(P_{i,j}^{\true} \Big\| P^{2,\ZZ^0}_{i,j}\Big)}
=o_P(NJ\cdot\eta'_J).
$$
In this case, Eq.~\eqref{eq-diff-eg}$=o_P(NJ\cdot\gamma_{J})=\sum_{i=1}^N\sum_{j=1}^J D(P_{i,j}^{\true} \| \bar\theta^{2\approx}_{j,\,\hat\xi_{i,j}}) - \text{(NE)}$ obtained prior to Case (1) indicates that 
$\sum_{i=1}^N\sum_{j=1}^J D(P_{i,j}^{\true} \| \bar\theta^{2\approx}_{j,\,\hat\xi_{i,j}}) = O_P(NJ\cdot(\gamma_J\vee\eta_J))$.
The expression of (NE) in \eqref{eq-def-ne} implies that
\begin{align*}
\text{(NE)}\leq &~ \sum_{j\notin\mathcal E_0}
     \frac{1}{N - |S_{j,\aaa_{L_j}}|} \sum_{b\neq c\in[L_j-1]} |S_{j,\aaa_b}| |S_{j,\aaa_c}| D(\theta^{\true}_{j,\aaa_b} \| \theta^{\true}_{j,\aaa_c}),
\end{align*}
where the inequality results from the convexity of the KL-divergence with respect to its second argument. 
First, $\sum_{j=1}^J\sum_{i=1}^N D(P_{i,j}^{\true} \| \bar\theta^{2\approx}_{j,\,\hat\xi_{i,j}}) = o(NJ\cdot(\gamma_J\vee\eta_J))$ indicates 
\begin{align*}
	o(NJ\cdot(\gamma_J\vee\eta_J)) = &~
	\sum_{j\in\mathcal E_0}\sum_{i=1}^N D(P_{i,j}^{\true} \| \bar\theta^{2\approx}_{j,\,\hat\xi_{i,j}})
	= \sum_{j\in\mathcal E_0}\sum_{a=0,1} \sum_{i=1}^N I(\xi^0_{i,j} = a) D(\theta^{\true}_{j,a} \| \bar\theta^{2\approx}_{j,\,\hat\xi_{i,j}}),
\end{align*}  
then a similar argument as Steps 5-6 in the proof of Theorem \ref{thm-joint-both} gives 
\begin{align}\label{eq-med}
	\frac{1}{J}\sum_{m=1}^{J_{\min} }\sum_{j\in\mathcal B^m} I(\hat\qq^{2\approx}_j\neq \qq_j^0)
    =&~
	o_P\left(\frac{\gamma_J\vee\eta'_J}{\Delta_J\cdot p_N}\right),\quad
	\frac{1}{N}\sum_{i=1}^N I(\hat\aa^{2\approx}_i\neq\aa^0_i)
   =o_P\left(\frac{\gamma_J\vee\eta'_J}{\Delta_J \cdot  \delta_J}\right).
\end{align}
Recall that for $j\in\mathcal B^m$ in the above display, the $\qq_j^0$ is some single-attribute vector and it remains to show the convergence of other multi-attribute $\qq$-vectors.
Second, we claim that for $j\in[J]$ such that $\hat\qq^{2\approx}_j\neq\qq_j^0$, there is
\begin{align}\label{eq-cl1}
	\sum_{i=1}^N  D(P_{i,j}^{\true} \| \bar\theta^{2\approx}_{j,\,\hat\xi_{i,j}}) = \Omega(N)\cdot \Delta_J.
\end{align}
If the above Claim \eqref{eq-cl1} is true, then
\begin{align*}
o_P(NJ\cdot(\gamma_J\vee\eta'_J))=
	\sum_{j=1}^J \sum_{i=1}^N  D(P_{i,j}^{\true} \| \bar\theta^{2\approx}_{j,\,\hat\xi_{i,j}})
\geq \sum_{j=1}^J I(\hat\qq^{2\approx}_j\neq\qq_j^0)\cdot  \Omega(N)\cdot \Delta_J,
\end{align*}
and further $(1/J) \sum_{j=1}^J I(\hat\qq^{2\approx}_j\neq\qq_j^0) = o_P((\gamma_J\vee\eta'_J)/\Delta_J)$. We next prove Claim \eqref{eq-cl1}. For notational simplicity, we simply write $\hat\qq_j^{2\approx}$ as $\hat\qq_j$. If $\hat\qq_j\neq\qq_j^0$, then
\begin{align*}
&~\sum_{i=1}^N  D(P_{i,j}^{\true} \| \bar\theta^{2\approx}_{j,\,\hat\xi_{i,j}})
\geq \sum_{i=1}^N 
I(\aa^0_i\succeq\qq^0_j)  
I(\hat\aa_i\succeq\hat\qq_j)
D(\theta^{\true}_{j,\aaa_{L_j}} \| \bar\theta^{2\approx}_{j,\,\hat\xi_{i,j}}) \\
= &~ \sum_{i=1}^N 
I(\aa^0_i\succeq\qq^0_j)
I(\hat\aa_i\succeq\hat\qq_j) D\Big(\theta^{\true}_{j,\aaa_{L_j}} \Big\| \frac{\sum_{m=1}^N I(\hat\xi_{m,j} = \hat\xi_{i,j}) P^{\true}_{m,j}  }{\sum_{m=1}^N I(\hat\xi_{m,j} = \hat\xi_{i,j}) } \Big) \\
= &~ \sum_{i=1}^N 
  I(\aa^0_i\succeq\qq^0_j)
  I(\hat\aa_i\succeq\hat\qq_j) 
  D\Big(\theta^{\true}_{j,\aaa_{L_j}} \Big\| \frac{\sum_{m=1}^N I(\hat\aa_m\succeq\hat\qq_j) P^{\true}_{m,j}  }{\sum_{m=1}^N I(\hat\aa_m\succeq\hat\qq_j) } \Big) \\
\geq &~ \sum_{i=1}^N 
  I(\aa^0_i\succeq\qq^0_j)
  I(\hat\aa_i\succeq\hat\qq_j) 
  \frac{ 2 \left[\sum_{m=1}^N I(\hat\aa_m\succeq\hat\qq_j,\, \aa^0_m\nsucceq\qq^0_j)(\theta^{\true}_{j,\aaa_{L_j}} - P^{\true}_{m,j})\right]^2}{\left[\sum_{m=1}^N I(\hat\aa_m\succeq\hat\qq_j)\right]^2} \\
= &~ \sum_{i=1}^N 
  I(\aa^0_i\succeq\qq^0_j)
  I(\hat\aa_i\succeq\hat\qq_j) 
  \frac{ 2\left[\sum_{m=1}^N I(\hat\aa_m\succeq\hat\qq_j,\, \aa^0_m\nsucceq\qq^0_j)\right]^2 (\theta^{\true}_{j,\aaa_{L_j}} - P^{\true}_{m,j})^2}{ \left[\sum_{m=1}^N I(\hat\aa_m\succeq\hat\qq_j)\right]^2 }\\
\geq &~ 
\frac{2}{N^2}\sum_{i=1}^N  I(\aa^0_i\succeq\qq^0_j)   I(\hat\aa_i\succeq\hat\qq_j) 
   \left[\sum_{m=1}^N I(\hat\aa_m\succeq\hat\qq_j,\, \aa^0_m\nsucceq\qq^0_j)\right]^2 (\theta^{\true}_{j,\aaa_{L_j}} - P^{\true}_{m,j})^2 \\
%
%
\geq &~ 
\frac{2}{N^2}\sum_{i=1}^N  
I(\aa^0_i = \hat\aa_i)
I(\aa^0_i\succeq\qq^0_j)   
I(\aa^0_i\succeq\hat\qq_j) 
   \left[\sum_{m=1}^N I(\aa^0_m = \hat\aa_m) I(\aa^0_m\succeq\hat\qq_j) I(\aa^0_m\nsucceq\qq^0_j)\right]^2 \cdot \Delta_J\\
\geq &~ \frac{2}{N^2}\cdot N\min\{1-o_P(1), 2^{K-2K_0}p_N\}\cdot N^2\min\{1-o_P(1), 2^{2K-2K_0}p_N^2\}\cdot\Delta_J \\
\geq &~ 
\Omega(N)\cdot\Delta_J\quad \text{with probability tending one},
\end{align*}
where the last but second inequality holds as long as $\hat\qq_j\neq\qq^0_j$. 
Now that we have proved Claim \eqref{eq-cl1}, the argument right after \eqref{eq-cl1} gives $(1/J) \sum_{j=1}^J I(\hat\qq_j\neq\qq_j^0) = o_P((\gamma_J\vee\eta'_J)/\Delta_J)$.
Combined with \eqref{eq-med}, we have shown
	\begin{align*}
		\frac{1}{J}\sum_{j=1}^J I(\hat\qq^{2\approx}_j\nsim \qq^0_j) = o_P\left(\frac{\gamma_J\vee\eta'_J}{\Delta_J\cdot p_N}\right),\quad
		\frac{1}{N}\sum_{i=1}^N I(\hat\aa^{2\approx}_i\nsim \aa^0_i) = o_P\left(\frac{\gamma_J\vee\eta'_J}{\Delta_J\cdot \delta_J}\right).
	\end{align*}
This completes the proof of Theorem \ref{thm-mis-strong}.

\bigskip
\subsection{Proofs of Technical Lemmas}\label{sec-lemma}
\begin{proof}[Proof of Lemma \ref{lem-express}]\label{pf-lem1}
Given a fixed $\ZZ$, denote $n^{(\ZZ)}_{j,a} = \sum_{i=1}^N Z_{i,a}$. The maximizing properties of $\hat \theta_{j,a}$ and $\bar \theta_{j,a}$ in \eqref{eq-zmle} imply that
\begin{equation}\label{eq-mprop}
	n_{j,a}\hat\theta_{j,a} = \sum_{i=1}^N Z_{i,a}r_{i,j},\quad
    n_{j,a}\bar\theta_{j,a} = \sum_{i=1}^N Z_{i,a}P_{i,j}.
\end{equation}
Recall $L=2^K$ denotes the number of latent class.
Using \eqref{eq-mprop}, we have the following,
\begin{align*}
 &~\ell(\RR;\, \ZZ) - \mathbb E[\ell (\RR;\,\ZZ)]\\
=&~ \sum_{j=1}^J \sum_{i=1}^N \sum_{a=1}^{L_j} Z_{i,a} [r_{i,j} \log\hat\theta_{j,a} + (1-r_{i,j})\log(1-\hat\theta_{j,a})]
 \\ 
&~ \qquad\qquad 
- \sum_{j=1}^J \sum_{a=1}^{L_j} Z_{i,a} [P_{i,j} \log\bar\theta_{j,a} + (1-P_{i,j})\log(1-\bar\theta_{j,a})]\\
=&~ \sum_{j=1}^J \sum_{a=1}^{L_j} n_{j,a}[\hat\theta_{j,a}\log \hat\theta_{j,a} + (1-\hat\theta_{j,a})\log(1-\hat\theta_{j,a})]\\
&~ \qquad\qquad - 
   \sum_{j=1}^J \sum_{a=1}^{L_j} n_{j,a}[\bar\theta_{j,a}\log \bar\theta_{j,a} + (1-\bar\theta_{j,a})\log(1-\bar\theta_{j,a})] \\
=&~ \sum_{j=1}^J \sum_{a=1}^{L_j} n_{j,a} \Big\{[\hat\theta_{j,a}\log \hat\theta_{j,a} + (1-\hat\theta_{j,a})\log(1-\hat\theta_{j,a})] 
 -  [\hat\theta_{j,a}\log \bar\theta_{j,a} + (1-\hat\theta_{j,a})\log(1-\bar\theta_{j,a})]\Big\} \\
 &  +  \sum_{j=1}^J \sum_{a=1}^{L_j} n_{j,a} \Big\{[\hat\theta_{j,a}\log \bar\theta_{j,a} + (1-\hat\theta_{j,a})\log(1-\bar\theta_{j,a})] 
 -   [\bar\theta_{j,a}\log \bar\theta_{j,a} + (1-\bar\theta_{j,a})\log(1-\bar\theta_{j,a})] \Big\}\\
  =&~  \sum_{j=1}^J \sum_{a=1}^{L_j} n_{j,a} D(\hat\theta_{j,a} \| \bar \theta_{j,a}) 
    + \sum_i \sum_j \Big\{[r_{i,j}\log \bar\theta_{j,z_i} + (1-r_{i,j})\log(1-\bar\theta_{j,z_i})] \\
    &\qquad\qquad\qquad\qquad\qquad\qquad\qquad  - [P_{i,j}\log \bar\theta_{j,z_i} + (1-P_{i,j})\log(1-\bar\theta_{j,z_i})] \Big\} \\
 =&~ \sum_{a=1}^{L_j} n_{j,a} \sum_{j=1}^J D(\hat\theta_{j,a} \| \bar \theta_{j,a}) 
    + \sum_{i=1}^N \sum_{j=1}^J r_{i,j} \log\Big( \frac{\bar\theta_{j,z_i}}{1-\bar\theta_{j,z_i}} \Big)
    - \sum_{i=1}^N \sum_{j=1}^J P_{i,j} \log\Big( \frac{\bar\theta_{j,z_i}}{1-\bar\theta_{j,z_i}} \Big).
\end{align*}
Define the random variable 
\begin{equation}
\label{eq-defx}	
X=\sum_{i=1}^N \sum_{j=1}^J r_{i,j} \log( \bar\theta_{j,z_i}/(1-\bar\theta_{j,z_i} )),
\end{equation}
then $X$  depends on $\ZZ$ and the above display equals the sum of $\sum_{a=1}^{L_j} n_{j,a} \sum_j D(\hat\theta_{j,a} \| \bar \theta_{j,a})$ and $X-\mathbb E[X]$.
This establishes \eqref{eq-lemma} in Lemma \ref{lem-express}.
\end{proof}

\begin{proof}[Proof of Lemma \ref{lem-kl-2p}]\label{pf-lem2}
	Given any fixed latent class memberships $\ZZ$, every $\hat\theta_{j,a}$ is an average of $n_{j,a}$ independent Bernoulli random variables $R_{1,j},\ldots, R_{N,j}$ with mean $\bar\theta_{j,a}$. We apply the Chernoff-Hoeffding theorem to obtain
\begin{equation}
	\mathbb P(\hat\theta_{j,a} \geq \bar \theta_{j,a} + t) 
	\leq e^{ -n_{j,a} D(\bar\theta_{j,a} + t \| \bar \theta_{j,a}) },\quad
	\mathbb P(\hat\theta_{j,a} \leq \bar \theta_{j,a} + t) \leq e^{ -n_{j,a} D(\bar\theta_{j,a} - t \| \bar \theta_{j,a}) }.
\end{equation}
Note that given a fixed $\ZZ$, each $\hat\theta_{j,a}$ can take values only in the finite set $\{0, 1/n_{j,a}, 2/n_{j,a}$,$\ldots$, $(n_{j,a}-1)/n_{j,a}, 1\}$ of cardinality $n_{j,a}+1$. 
We denote this range of $\hat\theta_{j,a}$ by $\hat\Theta^{j,a}$. Then $$\mathbb P(\hat\theta_{j,a} = \vartheta)\leq \exp\{ -n_{j,a} D(\vartheta \| \bar \theta_{j,a})\}$$ for any $\vartheta \in \hat\Theta^{j,a}$. 
Then 
$\mathbb P(\hat\theta_{j,a}\in\hat\Theta^{j,a})\leq e^{ -n_{j,a} D(\hat\theta_{j,a} \| \bar \theta_{j,a}) }.$
Further denote the range of the matrix $\hat\TT=(\hat\theta_{j,a})$ by $\hat\Theta$. Since entries of $\RR$ are independent given $\ZZ$, the following holds for any $\tilde\TT\in\hat\Theta$,
\begin{equation}\label{eq-that}
\mathbb P(\hat\TT = \tilde\TT\mid \ZZ) \leq
\exp\Big\{ 
-\sum_{j=1}^J\sum_{a=0,1} n_{j,a} D(\tilde \theta_{j,a} \| \bar \theta_{j,a})
\Big\}.
\end{equation}
Now consider the cardinality of the set $\hat\Theta$ given $\ZZ$. Since for each of the $J\times L$ entries in $\hat\TT$, $\hat\theta_{j,a}$ can independently take on $n_{j,a}+1$ different values, there is $|\hat\Theta| = [\prod_{a=0,1}(n_{j,a}+1)]^J$. Considering the natural constraint $\sum_{a=0,1} n_{j,a} = N$, we have 
\begin{equation}\label{eq-That}
	|\hat\Theta| = \prod_{j=1}^J\prod_{a=0,1}(n_{j,a}+1)
	\leq \Big(\frac{N}{2}+1\Big)^{2J}.
\end{equation}
Define the event
	$\hat\Theta_\epsilon = \{\hat\TT\in\hat\Theta:\, \sum_{j=1}^J\sum_{a=0,1} n_{j,a} D(\hat\theta_{j,a}  \| \bar \theta_{j,a} ) \geq \epsilon\}$, and combine \eqref{eq-that} and \eqref{eq-That} to obtain
\begin{align*}
&~\mathbb P\left(\sum_{j=1}^J\sum_{a=0,1} n_{j,a} D(\hat\theta_{j,a}  \| \bar \theta_{j,a} ) \geq \epsilon \mid\ZZ\right) \\
= &~ \sum_{\tilde\TT\in\hat\Theta_{\epsilon}} \mathbb P\left(\hat\TT=\tilde\TT,~\sum_{j=1}^J\sum_{a=0,1} n_{j,a} D(\tilde\theta_{j,a}  \| \bar \theta_{j,a})  \geq \epsilon \mid\ZZ\right) \\
\leq &~ |\hat\Theta_{\epsilon}| \exp\left(- \sum_{j=1}^J\sum_{a=0,1} n_{j,a} D(\tilde\theta_{j,a}  \| \bar \theta_{j,a})\right)
\leq \Big(\frac{N}{2}+1\Big)^{2J} e^{-\epsilon}.
\end{align*}
The above result holds for fixed $\ZZ$, we apply a union bound over all the $L^N$ possible assignment $\ZZ$ and obtain 
$$
\mathbb P\left(\sum_{j=1}^J\sum_{a=0,1} n_{j,a} D(\hat\theta_{j,a}  \| \bar \theta_{j,a} ) \geq \epsilon \right) \leq L^N \Big(\frac{N}{2}+1\Big)^{2J} e^{-\epsilon}.
$$
Now take $\delta = (2^K)^N\Big(\frac{N}{2^K}+1\Big)^{J2^K} e^{-\epsilon}$, then $\epsilon = N\log(2^K) + JL\log(\frac{N}{2^K} + 1) - \log\delta$. Therefore the following event happens with probability at least $1-\delta$,
\begin{equation*}
\sum_{j=1}^J\sum_{a=0,1}  n_{j,a} D(\hat\theta_{j,a}  \| \bar \theta_{j,a}) < \epsilon = N\log (2^K) + 2J\log\Big(\frac{N}{2} + 1\Big) - \log\delta.
\end{equation*}
This concludes the proof of Lemma \ref{lem-kl-2p}.
\end{proof}

\begin{proof}[Proof of Lemma \ref{lem-j012}]\label{pf-lem3}
For $j\in\mathcal J_1$, there is $\min\{N^j_{10},N^j_{00}\} + \min\{N^j_{11},N^j_{01}\} = N_{10}^j + N_{11}^j = \sum_{i=1}^N I(\xi^0_{i,j} = 1)\geq 2^{K-K_j}\cdot Np_N$ (under a two-parameter SLAM).
Now for an arbitrary positive constant $b\in(0,1)$, we look at 
	\begin{align*}
		\MP \left(|\mathcal J_1| \geq b\frac{J\cdot\gamma_{J}}{\beta_J}\right)
	\leq &~ \MP\left(\sum_{j\in\mathcal J_1} (N_{10}^j + N_{11}^j) \geq b\frac{J\cdot\gamma_{J}}{\beta_J}\cdot 2^{K-K_j}\cdot N p_{N}\right) \\
	\leq &~ \MP\left( \bar\ell(\ZZ^0) - \bar\ell(\hat\ZZ) \geq b N J\gamma_{J}\cdot 2^{K-K_j}\cdot p_{N}\right) \\
	\leq &~ \MP\left( \bar\ell(\ZZ^0) - \bar\ell(\hat\ZZ) \geq b N J\gamma_{J}\cdot 2^{K-K_0}\cdot p_{N}\right) \\
	\leq &~  \MP\left( \bar\ell(\ZZ^0) - \bar\ell(\hat\ZZ) \geq \frac{b\epsilon}{2^{K_0}} N J\gamma_{J}\right).
	\end{align*}
The conclusion of Step 4 gives $\bar\ell(\ZZ^0) - \bar\ell(\hat\ZZ) = o_P(NJ\gamma_{J})$ with $\gamma_{J}\to 0$ as $N,J,K\to\infty$.
Now that $b,\epsilon,K_0$ are constants, we obtain that $\MP (|\mathcal J_1| \geq b\cdot J\gamma_{J}/\beta_J) = o(1)$ and $|\mathcal J_1| = o_P(J\gamma_{J}/\beta_J)$. Similar arguments gives $|\mathcal J_2| = o_P(J\gamma_{J}/\beta_J)$.
Since $\mathcal J_0 = [J]\setminus (\mathcal J_1 \cup \mathcal J_2)$, we have $1 - |\mathcal J_0|/J =o_P(\gamma_{J}/\beta_J)$. This completes the proof of Lemma \ref{lem-j012}.
\end{proof}

\begin{proof}[Proof of Lemma \ref{lem-qa}]\label{pf-lem4}
	\label{pf-lem2}
	
	\noindent
	\textbf{Part (a) of the lemma:}
	Let $\aaa^1=\mathbf 0_K, \aaa^2 = \ee_1,\ldots, \aaa^{2^K} = \mathbf 1_K$ denote the distinct $2^K$ number of $K$-dimensional binary vectors in $\{0,1\}^K$. Suppose the two $K\times 2^K$ matrices $\xi(I_K, \{\aaa^1,\aaa^2,\ldots,\aaa^{2^K}\}) = \xi(\mathbf Q, \{\aa^1,\aa^2,\ldots,\aa^{2^K}\})$. First, since $\xi(I_K, \aaa^K) = \aaa^K$, the matrix $\xi(I_K, \{\aaa^1,\aaa^2,\ldots,\aaa^{2^K}\})$ has $2^K$ distinct column vectors arranged in 
	\begin{align*}\setlength{\dashlinegap}{1pt}
		\xi(I_K, \{\aaa^1,\aaa^2,\ldots,\aaa^{2^K}\}) = 
		\left(\begin{array}{c:c:c:c}
			\aaa^1 &\aaa^2 & \cdots &\aaa^{2^K}
		\end{array}\right).
	\end{align*}
	First, if the set $\{\aa^1,\ldots,\aa^{2^K}\}$ contain some identical vectors $\aa^m = \aa^\ell$, then their corresponding columns in the ideal response matrix must be identical as well, $\xi(\QQ, \aa^m) = \xi(\QQ,\aa^\ell)$ for any $Q$. So without loss of generality, we next consider the case where $\aa^1,\ldots,\aa^{2^K}$ are distinct, so $\{\aa^1,\ldots,\aa^{2^K}\}=\{0,1\}^K$.
We next show that if $\QQ\nsim I_K$, the $\xi(I_K, \{\aaa^1,\ldots,\aaa^{2^K}\})$ must contain identical column vectors.
If $Q\nsim I_K$, then there must exist some $k\in[K]$ such that vector $\ee_k$ does not belong to the set of row vectors of $Q$. Consider the $m,\ell\in[2^K]$ such that $\aa^m=\zero_K, \aa^\ell=\ee_k$, then 
	\begin{equation}\label{eq-ir}
		\xi(\QQ,\aa^m)= \xi(\QQ,\zero) = \xi(\QQ,\ee_k) = \xi(\QQ,\aa^\ell),
	\end{equation}
This is because the two attribute patterns $\aa^m$ and $\aa^\ell$ will have identical ideal response for any item with a $\qq$-vector not equal to $\ee_k$, and that $Q$ does not have any row vector $\ee_k$. This shows $Q$ must equal $I_K$ up to a column permutation and $\{\aa^1,\ldots,\aa^{2^K}\}$, proving part (a) of the lemma.

\medskip\noindent
\textbf{Part (b) of the lemma:} Suppose $\xi(\qq_j, \{0,1\}^K) = \xi(\tilde\qq_j, \{0,1\}^K)$ and $\qq_j \nsim \tilde\qq_j$. Consider two scenarios: (1) $\qq_j\nsucceq \tilde\qq_j$ and (2) If $\xi(\qq_j, \{0,1\}^K) = \xi(\tilde\qq_j, \{0,1\}^K)$. First, if $\qq_j\nsucceq \tilde\qq_j$, we can just take an attribute pattern $\aa^m=\qq_j$. Since $\aa^m\succeq \qq_j$ and $\aa^m\nsucceq \tilde\qq_j$, the following holds by the definition of $\xi$,
$$
\xi(\qq_j,\aa^m) = 1\neq 0 = \xi(\tilde\qq_j,\aa^m).
$$
This implies $\xi(\qq_j,\{0,1\}^K)\neq \xi(\tilde\qq_j,\{0,1\}^K)$ and contradicts the assumption of part (b). So we must have $\tilde\qq_j\nsim \qq_j$. This proves part (b) of the lemma.
\end{proof}

\begin{proof}[Proof of Lemma \ref{lem-kl}]\label{pf-lem5}
Following a similar argument as the proof of Lemma \ref{lem-kl-2p}, we have 
\begin{equation}\label{eq-that}
\mathbb P(\hat\TT = \tilde\TT\mid \ZZ) \leq
\exp\Big\{ 
-\sum_{j=1}^J\sum_{a=1}^{L_j} n_{j,a} D(\tilde \theta_{j,a} \| \bar \theta_{j,a})
\Big\}.
\end{equation}
Now consider the cardinality of the set $\hat\Theta$ given $\ZZ$. Since for each of the $J\times L$ entries in $\hat\TT$, $\hat\theta_{j,a}$ can independently take on $n_{j,a}+1$ different values, there is $|\hat\Theta| = [\prod_{a=1}^{L_j}(n_{j,a}+1)]^J$. Considering the natural constraint $\sum_{a=1}^{L_j} n_{j,a} = N$ and also $L_j=2^{K_j}\leq 2^{K_0}$ by Assumption \ref{cond-id}, we have 
\begin{equation}\label{eq-That}
	|\hat\Theta| = \prod_{j=1}^J\prod_{a=1}^{L_j}(n_{j,a}+1)
	\leq \Big(\frac{N}{2^{K_0}}+1\Big)^{J2^{K_0}}.
\end{equation}
Define the event
	$\hat\Theta_\epsilon = \{\hat\TT\in\hat\Theta:\, \sum_{j=1}^J\sum_{a=1}^{K_0} n_{j,a} D(\hat\theta_{j,a}  \| \bar \theta_{j,a} ) \geq \epsilon\}$, and combine \eqref{eq-that} and \eqref{eq-That} to obtain
\begin{align*}
&~\mathbb P\left(\sum_{j=1}^J\sum_{a=1}^{L_j} n_{j,a} D(\hat\theta_{j,a}  \| \bar \theta_{j,a} ) \geq \epsilon \mid\ZZ\right) 
\leq \Big(\frac{N}{2^{K_0}}+1\Big)^{J2^{K_0}} e^{-\epsilon}.
\end{align*}
The above result holds for fixed $\ZZ$, we apply a union bound over all the $L^N$ possible assignment $\ZZ$ and obtain 
$$
\mathbb P\left(\sum_{j=1}^J\sum_{a=1}^{2^{K_0}} n_{j,a} D(\hat\theta_{j,a}  \| \bar \theta_{j,a} ) \geq \epsilon \right) \leq (2^K)^N \Big(\frac{N}{2^{K_0}}+1\Big)^{J2^{K_0}} e^{-\epsilon}.
$$
Therefore the following event happens with probability at least $1-\delta$,
\begin{equation*}
\sum_{j=1}^J\sum_{a=1}^{L_j} n_{j,a} D(\hat\theta_{j,a}  \| \bar \theta_{j,a}) < \epsilon = N\log (2^K) + J2^{K_0}\log\Big(\frac{N}{2^{K_0}} + 1\Big) - \log\delta.
\end{equation*}
This concludes the proof of Lemma \ref{lem-kl}.
\end{proof}


\begin{proof}[Proof of Lemma \ref{prop-scale}]\label{pf-lem6}
Combining the results of Step 2 and Step 3, since that there are $(2^K)^N$ possible assignments of $\ZZ$, we apply the union bound to obtain
\begin{align}\label{eq-deltanj}
	&~\MP (\max_{\ZZ} |\ell(\RR;\,\ZZ) - \ME \ell(\RR;\,\ZZ)| \geq 2\epsilon \delta_{NJ})\\ \notag
\leq &~ L^N \MP \left[
\left\{\sum_{j=1}^J\sum_{a=1}^{L_j} n_{j,a} D(\hat\theta_{j,a} \| \bar \theta_{j,a})\geq \epsilon \delta_{NJ}\right\}
\cup
\left\{|X - \ME[X]| \geq \epsilon \delta_{NJ}\right\}
\right]\\ \notag
\leq &~  \exp\Big\{N\log(2^K) + J2^{K_0}\log\Big(\frac{N}{2^{K_0}} + 1\Big) - \epsilon \delta_{NJ}\Big\} \\ \notag
&~  + 2\exp\Big\{N\log(2^{K}) -\frac{\epsilon^2 \delta_{NJ}}{4(MNJ/\delta_{NJ})(\log J)^2 + (4/3)\epsilon\log J) } \Big\}.
\end{align}
In order for the second term on the right hand side of the above display to go to zero, the following of $\delta_{NJ}$ would suffice,
\begin{equation}
	\label{eq-scale}
	\delta_{NJ} \succsim N\sqrt{MJ\log(2^{K})} \log J.
\end{equation}
We take $\delta_{NJ} = N\sqrt{MJ\log(2^{K})} (\log J)^{1+\epsilon}$ for a small positive constant $\epsilon$. Further, under this $\delta_{NJ}$, in order for the first term on the right hand side of \eqref{eq-deltanj} to go to zero, 
Then the right hand side of \eqref{eq-deltanj} goes to zero as $N, J$ go large. 
Then the scaling $\sqrt{J}\cdot 2^{K_0} = O(\sqrt{M} N^{1-c})$ and $K=o(MJ\log J)$ described in the theorem
 yields $\MP (\max_{\ZZ} |\ell(\RR;\,\ZZ) - \ME \ell(\RR;\,\ZZ)| \geq 2\epsilon \delta_{NJ}) = o(1)$, which implies
\begin{equation*}
	\frac{1}{NJ}\max_{\ZZ} |\ell(\RR;\,\ZZ) - \ME \ell(\RR;\,\ZZ)| 
	=o_P\left(\frac{\sqrt{M\log (2^{K_0})}}{\sqrt{J}} (\log J)^{1+\epsilon}\right).
\end{equation*}
This completes the proof of the lemma.
\end{proof}

\begin{proof}[Proof of Lemma \ref{prop-expr-mis}]\label{pf-lem8}
Recall that $\ZZ = (\QQ, \AA)$ and $\xi_{i,j} = \prod_{k=1}^K a_{i,k}^{q_{j,k}}$, then
\begin{align*}
	&~\ell^{2\approx}(\RR;\, \ZZ) - \ME_{\true} [\ell^{2\approx}(\RR;\, \ZZ)]\\
=   &~\sum_{i=1}^N\sum_{j=1}^J\sum_{a=0,1} I(\xi_{i,j}=a)\Big[ r_{i,j}\log(\hat\theta_{j,\xi_{i,j}}) + (1-r_{i,j})\log(1-\hat\theta_{j,\xi_{i,j}})\Big] \\
 &\quad -\sum_{i=1}^N\sum_{j=1}^J\sum_{a=0,1} I(\xi_{i,j}=a)\Big[ P_{i,j}^{\true}\log(\bar\theta_{j,\xi_{i,j}}) + (1-P_{i,j}^{\true})\log(1-\bar\theta_{j,\xi_{i,j}})\Big] \\
=   &~\sum_{j=1}^J\sum_{a=0,1} n_{j,a}\left[\hat\theta_{j,a}\log\left(\frac{\hat\theta_{j,a}}{\bar\theta_{j,a}}\right) + (1-\hat\theta_{j,a})\log\left(\frac{1-\hat\theta_{j,a}}{1-\bar\theta_{j,a}}\right)\right] \\
 &\quad +\sum_{j=1}^J\sum_{a=0,1} n_{j,a}\left[(\hat\theta_{j,a}-\bar\theta_{j,a})\log(\bar\theta_{j,a}) - (\hat\theta_{j,a}-\bar\theta_{j,a})\log(1-\bar\theta_{j,a})\right] \\
= &~ \sum_{j=1}^J\sum_{a=0,1} n_{j,a} D(\hat\theta_{j,a}\|\bar\theta_{j,a}) + \sum_{i=1}^N \sum_{j=1}^J (r_{i,j} - P_{i,j}^{\true})\log\left( \frac{\bar\theta_{j,\xi_{i,j}}}{1-\bar\theta_{j,\xi_{i,j}}} \right).
\end{align*}
Given any fixed $\ZZ$, every $\hat\theta_{j,a}$ is an average of $n_{j,a}$ independent Bernoulli random variables $r_{1,j},\ldots, r_{N,j}$ with mean $\bar\theta_{j,a}$ because $\ME_{\true}[r_{i,j}] = P_{i,j}^{\true}$. Following a similar argument as the proof of Lemma \ref{lem-kl-2p}, 
the following event happens with probability at least $1-\delta$,
\begin{equation*}
\sum_{j=1}^J\sum_{a=0,1}  n_{j,a} D(\hat\theta_{j,a}  \| \bar \theta_{j,a}) < \epsilon = N\log (2^K) + 2J\log\left(\frac{N}{2} + 1\right) - \log\delta.
\end{equation*}
Further, a similar argument as Step 3 of the proof of Theorem \ref{thm-joint-both} gives
\begin{align*}
	&~\MP\left( \max_{\ZZ} |\ell^{2\approx}(\RR;\, \ZZ) - \ME_{\true} [\ell^{2\approx}(\RR;\, \ZZ)]| > 2\epsilon\delta_{}NJ\right)
	\\ \notag
\leq &~  \exp\Big\{N\log (2^K) + 2J\log\Big(\frac{N}{2} + 1\Big) - \epsilon \delta_{NJ}\Big\} \\ \notag
&~  \qquad + 2\exp\Big\{N\log(2^K) -\frac{\epsilon^2 \delta_{NJ}}{2d^2(MNJ/\delta_{NJ})(\log J)^2 + (4/3)\epsilon\log J) } \Big\}.
\end{align*}
Therefore under the scaling $\sqrt{J} = O(\sqrt{M}N^{1-c})$ and $K=o(MJ\log J)$, we reach the conclusion of the lemma.
\end{proof}

\section{Additional Algorithms}
\label{sec-addalgo}
In this section, we provide two additional algorithms, Algorithm \ref{algo-adg-saem} and Algorithm \ref{algo-screen-miss}. Algorithm \ref{algo-adg-saem} is an alternating direction Gibbs stochastic-approximation-EM algorithm mentioned in Section \ref{sec-est-2p} in the main text. This algorithm applies the stochastic approximation to both $\QQ$ and $\AA$ in each iteration, instead of only to $\QQ$ as in Algorithm \ref{algo-screen} presented in the main text. In practice, we found through simulations that in cases where $N$ and $J$ are very large relative to $K$, this Algorithm \ref{algo-adg-saem} yields better estimation accuracy than Algorithm \ref{algo-screen}. The theoretical investigations of the properties of the algorithms are left to the future study.
Algorithm \ref{algo-screen-miss} is for estimating $\QQ$ and $\AA$ with missing entries in the data matrix $\RR$, as mentioned in Section \ref{sec-real} in the main text.

\begin{algorithm}[h!]
\caption{ADG-SAEM: Alternating Direction Gibbs Stochastic Approximation EM}
\label{algo-adg-saem}

\SetKwInOut{Input}{input}
\SetKwInOut{Output}{Output}

\KwData{Response matrix $\RR=(r_{i,j})_{N\times J}\in\{0,1\}^{N\times J}$ and number of attributes $K$.}
 Initialize $\AA=(a_{i,k})_{N\times K}\in\{0,1\}^{N\times K}$ and  $\QQ=(q_{j,k})_{J\times K}\in\{0,1\}^{J\times K}$.\\
 Initialize parameters $\ttt^+$ and $\ttt^-$. 
 Set $t=1$, ~$\AA^{\text{ave}}=\zero$.\\
 \While{not converged}{
 
   \lFor{$(i,j)\in[N]\times [J]$}{
      $
      \psi_{i,j} \leftarrow 
      r_{i,j}\log[\theta^+_j / \theta^-_j] + 
      (1-r_{i,j})\log[(1-\theta^+_j) / (1-\theta^-_j)]
      $
  }
  
  $\AA^{\text{new}}\leftarrow\zero$, \quad $\QQ^{\text{new}}\leftarrow\zero$.\\
\For{$r\in[C]$}{
    \For{$(i,k)\in[N]\times[K]$}{
    Draw $a_{i,k}\sim\text{Bernoulli}\Big(\sigma\Big(-\sum_{j} q_{j,k} \prod_{m\neq k} a_{i,m}^{q_{j,m}} \psi_{i,j} \Big)\Big)$
    }
     $\AA^{\text{new}} \leftarrow \AA^{\text{new}}+\AA$;    
    }

  \For{$r\in[C]$}{
    \For{$(j,k)\in[J]\times[K]$}{
    Draw $q_{j,k}\sim\text{Bernoulli}\Big(\sigma\Big(\sum_{i} (1-a_{i,k}) \prod_{m\neq k} a_{i,m}^{q_{j,m}} \psi_{i,j} \Big)\Big)$
    }
     $\QQ^{\text{new}} \leftarrow \QQ^{\text{new}}+\QQ$;
}

  \begin{align*}
  &\AA^{\text{ave}} \leftarrow \frac{1}{t}\AA^{\text{new}}/C
  + \Big(1-\frac{1}{t}\Big)\AA^{\text{ave}};
  \quad \AA= I\Big(\AA^{\text{ave}}>\frac{1}{2}\Big)~~\text{element-wisely;}
   \\
  &\QQ^{\text{ave}} \leftarrow \frac{1}{t}\QQ^{\text{new}}/C
  + \Big(1-\frac{1}{t}\Big)\QQ^{\text{ave}};
  \quad \QQ= I\Big(\QQ^{\text{ave}}>\frac{1}{2}\Big)~~\text{element-wisely;}
  \end{align*}

  $\mathbf M = (M_{i,j})_{N\times J} =\Big( \prod_k a_{i,k}^{q_{j,k}} \Big)_{N\times J}$;

   \For{$j\in [J]$}{
\begin{align*}
&\theta_{j,\text{new}}^+ \leftarrow \frac{\sum_{i}r_{i,j} M_{i,j}}{\sum_{i}M_{i,j}},
\qquad\qquad
\theta_{j,\text{ave}}^+ \leftarrow 
\frac{1}{t}\theta_{j,\text{new}}^+ +\Big(1-\frac{1}{t}\Big)\theta_{j,\text{ave}}^+;
\\
&\theta_{j,\text{new}}^- \leftarrow \frac{\sum_{i}r_{i,j} (1-M_{i,j})}{\sum_{i}(1-M_{i,j})},
\qquad
\theta_{j,\text{ave}}^- \leftarrow
\frac{1}{t}\theta_{j,\text{new}}^- + \Big(1-\frac{1}{t}\Big)\theta_{j,\text{ave}}^-;
\end{align*}
   }
   
   $t\leftarrow t+1;$
 }

\Output{$\hat\QQ$ and $\hat\AA$.}
\end{algorithm}


\begin{algorithm}[h!]
\caption{ADG-EM with missing data}\label{algo-screen-miss}

\SetKwInOut{Input}{input}
\SetKwInOut{Output}{Output}

\KwData{Responses $\RR$ with the set of indices of observed entries $\Omega\subseteq[N]\times [J]$.}
 Initialize attribute patterns $(a_{i,k})_{N\times K}\in\{0,1\}^{N\times K}$; and structural matrix $(q_{j,k})_{J\times K}\in\{0,1\}^{J\times K}$.\\
 Initialize parameters $\ttt^+$ and $\ttt^-$. \quad
 Set $t=1$, ~$\AA^{\text{ave}}=\zero$.\\
 \While{not converged}{
 
   \lFor{$(i,j)\in\Omega$}{
         $
      \psi_{i,j} \leftarrow 
      r_{i,j}\log[\theta^+_j / \theta^-_j] + 
      (1-r_{i,j})\log[(1-\theta^+_j) / (1-\theta^-_j)]
      $
  }
  
  $\AA^{\text{s}}\leftarrow\zero$, \quad $\QQ^{\text{s}}\leftarrow\zero$.\\
\For{$r\in[C]$}{
    \For{$(i,k)\in[N]\times[K]$}{
    $\text{Draw~~}a_{i,k}\sim\text{Bernoulli}\Big(\sigma\Big(-\sum_{j:\,(i,j)\in\Omega} q_{j,k} \prod_{m\neq k} a_{i,m}^{q_{j,m}} \psi_{i,j} \Big)\Big);$
    }
     $\AA^{\text{s}} \leftarrow \AA^{\text{s}}+\AA$    
    }
    
  $\AA^{\text{ave}} \leftarrow t^{-1}\AA^{\text{s}}/C
  + \Big(1-t^{-1}\Big)\AA^{\text{ave}};\quad
  t\leftarrow t+1.$

  \For{$r\in[C]$}{
    \For{$(j,k)\in[J]\times[K]$}{
    $\text{Draw~~}q_{j,k}\sim\text{Bernoulli}\Big(\sigma\Big(\sum_{i:\,(i,j)\in\Omega} (1-a_{i,k}) \prod_{m\neq k} a_{i,m}^{q_{j,m}} \psi_{i,j} \Big)\Big);$
    }
    $\QQ^{\text{s}} \leftarrow \QQ^{\text{s}}+\QQ$
}
  
  $\QQ= I(\QQ^{\text{s}}/C > 1/2)$ element-wisely;\quad
  $\mathbf I^{\text{ave}} =\Big( \prod_k \{a_{i,k}^{\text{ave}}\}^{q_{j,k}} \Big)_{N\times J}$;
  
   \For{$j\in [J]$}{
   \begin{align*}
   	\theta_j^+ 
   	\leftarrow \frac{\sum_{i:\,(i,j)\in\Omega}r_{i,j} I^{\text{ave}}_{i,j}}{\sum_{i:\,(i,j)\in\Omega}I_{i,j}^{\text{ave}}},\quad
   	\theta_j^- 
   	\leftarrow \frac{\sum_{i:\,(i,j)\in\Omega}r_{i,j} (1-I^{\text{ave}}_{i,j})}{\sum_{i:\,(i,j)\in\Omega}(1-I_{i,j}^{\text{ave}})};
   \end{align*}
   }
 }

$\hat \AA = I(\AA^{\text{ave}} > 1/2)$ element-wisely.

\Output{$\hat\QQ$ and $\hat\AA$.}
\end{algorithm}


%
%

\end{document}